\documentclass[a4paper]{article}
\pdfoutput=1
\usepackage{amsmath,amsfonts,amsthm,amssymb}
\usepackage{mathtools}
\usepackage[usenames,dvipsnames]{color}
\usepackage[aligntableaux=top]{ytableau}
\usepackage{tikz-cd}
\usetikzlibrary{matrix,arrows,calc}
\usepackage{calc}
\usepackage[colorlinks=true,linkcolor=blue,citecolor=red,urlcolor=Violet,bookmarksdepth=subsection,final]{hyperref}

\makeatletter
\def\tikzcdset{\pgfqkeys{/tikz/commutative diagrams}}
\tikzcdset{
  @shiftabletopath/.style={
    /tikz/execute at begin to={%
      \begingroup%
        \def\tikz@tonodes{coordinate[pos=0,commutative diagrams/@starttransform/.try](tikzcd@nodea) %
                          coordinate[pos=1,commutative diagrams/@endtransform/.try](tikzcd@nodeb)}%
        \path (\tikztostart) \tikz@to@path;%
      \endgroup%
      \def\tikztostart{tikzcd@nodea}%
      \def\tikztotarget{tikzcd@nodeb}%
      \tikzset{insert path={(tikzcd@nodea)}}},
    /tikz/commutative diagrams/@shiftabletopath/.code={}},
  shift left/.style={
    /tikz/commutative diagrams/@shiftabletopath,
    /tikz/execute at begin to={%
      \pgfpointnormalised{%
        \pgfpointdiff{\pgfpointanchor{tikzcd@nodeb}{center}}{\pgfpointanchor{tikzcd@nodea}{center}}}%
      \pgfgetlastxy{\tikzcd@x}{\tikzcd@y}%
      \pgfmathparse{#1}%
      \ifpgfmathunitsdeclared\else
        \pgfmathparse{\pgfmathresult*\pgfkeysvalueof{/tikz/commutative diagrams/shift left/.@def}}
      \fi
      \coordinate (tikzcd@nodea) at ([shift={(\pgfmathresult*\tikzcd@y,-\pgfmathresult*\tikzcd@x)}]tikzcd@nodea);%
      \coordinate (tikzcd@nodeb) at ([shift={(\pgfmathresult*\tikzcd@y,-\pgfmathresult*\tikzcd@x)}]tikzcd@nodeb);%
      \tikzset{insert path={(tikzcd@nodea)}}}},
  shift right/.style={
    /tikz/commutative diagrams/shift left={-(#1)}},
  equal/.code={\tikzcdset{double line}\pgfsetarrows{-}},
  shift left/.default=+0.56ex,
  shift right/.default=1
}
\makeatother
\allowdisplaybreaks[1] % Allow page breaks in equation blocks.

\newtheorem{thm}{Theorem}
\newtheorem{prop}[thm]{Proposition}
\newtheorem{lem}[thm]{Lemma}
\newtheorem{cor}[thm]{Corollary}
\theoremstyle{remark}
\newtheorem{rem}{Remark}

	\newcommand{\yt}{\ytableaushort}
	\newcommand{\yd}{\ydiagram}
  \newcommand{\del}{\partial}
  \newcommand{\oo}{\infty}
  
  \newcommand{\ad}{\mathrm{ad}}
  \newcommand{\Ad}{\mathrm{Ad}}
  \newcommand{\Aut}{\operatorname{Aut}}

\renewcommand{\d}{\mathrm{d}}
  \DeclareMathOperator{\odim}{\overline{\dim}}
  \newcommand{\D}{\Secs_c'}
  \newcommand{\Do}{\Secs'}
  \newcommand{\eps}{\varepsilon}
  
  \newcommand{\id}{\mathrm{id}}
  \newcommand{\im}{\operatorname{im}}
  \newcommand{\Isom}{\mathrm{Isom}}
  \newcommand{\tr}{\operatorname{tr}}
  \newcommand{\sso}{\subset}
  \newcommand{\sse}{\subseteq}
  \newcommand{\Secs}{\Gamma}
  \newcommand{\supp}{\operatorname{supp}}
  \newcommand{\End}{\operatorname{End}}
  \newcommand{\K}{\mathcal{K}}
  \newcommand{\KY}{\mathcal{KY}}
  \newcommand{\Lie}{\mathcal{L}}
  \newcommand{\lf}{\mathit{lf}}
  \newcommand{\F}{\mathcal{F}}
  \newcommand{\g}{\mathfrak{g}}
  
  \newcommand{\GG}{\mathcal{G}}
  \newcommand{\GL}{\mathrm{GL}}
\renewcommand{\H}{\mathcal{H}}
  \newcommand{\R}{\mathbb{R}}
  \newcommand{\rk}{\operatorname{rk}}
\renewcommand{\S}{\mathcal{S}}
\renewcommand{\sc}{\mathit{sc}}
  
  \newcommand{\tf}{\mathit{tf}}
  \newcommand{\U}{\mathcal{U}}
  \newcommand{\Y}{\mathrm{Y}}

\title{The Calabi complex and Killing sheaf cohomology}
\author{Igor Khavkine\\
Department of Mathematics, University of Trento,\\
and TIFPA-INFN, Trento,\\
I--38123 Povo (TN) Italy\\
\texttt{igor.khavkine@unitn.it}}

\begin{document}
\maketitle

\begin{abstract}
It has recently been noticed that the degeneracies of the Poisson
bra\-cket of linearized gravity on constant curvature Lorentzian manifold
can be described in terms of the cohomologies of a certain complex of
differential operators. This complex was first introduced by Calabi and
its cohomology is known to be isomorphic to that of the (locally
constant) sheaf of Killing vectors. We review the structure of the
Calabi complex in a novel way, with explicit calculations based on
representation theory of $\GL(n)$, and also some tools for studying its
cohomology in terms of of locally constant sheaves. We also conjecture
how these tools would adapt to linearized gravity on other backgrounds
and to other gauge theories. The presentation includes explicit formulas
for the differential operators in the Calabi complex, arguments for its
local exactness, discussion of generalized Poincar\'e duality, methods
of computing the cohomology of locally constant sheaves, and example
calculations of Killing sheaf cohomologies of some black hole and
cosmological Lorentzian manifolds.
\end{abstract}

\tableofcontents

\section{Introduction}\label{sec:intro}
The Calabi complex is a differential complex that was introduced in by
E.~Calabi in 1961~\cite{calabi}. It has an extended pre-history, though.
One way to characterize it is as a canonical formal compatibility
complex (the \emph{second Spencer sequence}) of the Killing equation on
(pseudo-)Riemannian manifolds of constant curvature. The solutions of
the Killing equation are (possibly only locally defined) infinitesimal
isometries. In the special context of classical linear elasticity
theory, the same operator also maps between the displacement and strain
fields~\cite{pommaret-mech,eastwood,stven}. It is well known that for
flat spaces (zero curvature) a complete set of formal compatibility
conditions for the Killing equation is given by the linearized Riemann
curvature operator, also known as the Saint-Venant compatibility
operator in the context of
elasticity~\cite{pommaret-mech,eastwood,stven}. Subsequent compatibility
conditions are furnished by the Bianchi identities. Thus, it would also
be reasonable to refer to it as the Killing-Riemann-Bianchi complex.

Calabi's interest in the eponymous complex stemmed from the isomorphism
between the cohomology of its global sections and the cohomology of the
sheaf of Killing vectors. Given a fine resolution of a sheaf, like one
provided by a locally exact sequence of differential operators on
sections of vector bundles, the general machinery of homological algebra
implies that the sheaf cohomology is in fact isomorphic to the
cohomology of this global sections of its resolution, the resolution of
the sheaf of locally constant functions by the de~Rham complex of
differential forms on a manifold being the canonical example. The bulk
of Calabi's original article was in fact spent proving that the
hypotheses needed for applying this general result actually hold,
thus providing a way to represent the cohomology of the Killing sheaf.
It is the latter object that was of intrinsic interest, as it was in
subsequent works by others~\cite{bbl,weil,berger}, motivated by the well
known interpretations of its lowest cohomology groups: in degree-$0$ as
the Lie algebra of global isometries and in degree-$1$ as the space of
non-trivial infinitesimal deformations of the metric under the constant
curvature restriction. Later, the Calabi complex was also seen as a
non-trivial example of a formally exact compatibility
complex~\cite{goldschmidt-calabi,gasqui-goldschmidt-fr,gasqui-goldschmidt,pommaret,eastwood}
constructed for the Killing operator by the methods of the formal theory
of partial differential equations developed by the school of
D.~C.~Spencer~\cite{spencer-deform1,spencer-deform2,spencer,quillen,goldschmidt-lin}.

More recently, the Calabi complex resurfaced in mathematical physics, in
the context of the (pre-)symplectic and Poisson structure of
relativistic classical field theories. In~\cite{kh-big,kh-peierls}, the
author has shown that the degeneracy subspaces of the naturally defined
pre-symplectic $2$-form and Poisson bivector on the infinite dimensional
phase space of relativistic classical field theories with possible
constraints and gauge invariance are controlled by the cohomology of
some differential complexes. In the case of Maxwell-like
theories~\cite[Sec.4.2]{kh-peierls}, this role is played by the de~Rham
complex, while in the case of linearized
gravity~\cite[Sec.4.4]{kh-peierls}, this role is played by the formal
compatibility complex of the Killing operator. In other words, for
linearization backgrounds of constant curvature (important examples
include Minkowski and de~Sitter spaces, as well as quotients thereof),
this is precisely the Calabi complex. When the linearization background
is merely locally symmetric, rather than of constant curvature, the
right complex to use is a slightly different one that was constructed by
Gasqui and Goldschmidt~\cite{gasqui-goldschmidt-fr,gasqui-goldschmidt}.
However, a discussion of the latter is beyond the scope of this work.
The construction of similar complexes adapted to other background
geometries appears to be an open problem. The degeneracy subspace of the
Poisson bivector of a classical field theory is of importance because it
translates almost directly into violations of a (strict) notion of
locality of the corresponding quantum field theory, a subject that has
recently been under intense
investigation~\cite{dl,sdh,bds,fewster-hunt,fl,benini,hack-lingrav,bss}.

The goal of this paper is to exploit the connection between the Calabi
complex and Killing sheaf cohomologies, in a direction opposite the
original one of Calabi, for the purpose of obtaining results relevant to
the above mentioned applications in mathematical physics. More
precisely, we consider the computation of certain sheaf cohomologies
much simpler than constructing quotient spaces of kernels of complicated
differential operators. Thus, the ability to equate the Calabi
cohomology groups, which for us are of primary interest, with Killing
sheaf cohomology groups is a significant technical simplification. Along
the way, we collect some relevant facts about the Calabi complex that
are either difficult or impossible to find in the existing literature,
along with other little known tools from the theory of differential
complexes~\cite{tarkhanov} needed to prove the desired equivalence and
to introduce cohomologies with compact supports. It is our hope that
this treatment of the Calabi complex could serve as a model for the
treatment of other differential complexes that are of importance in
mathematical physics.

In Section~\ref{sec:calabi}, we discuss the explicit form of the Calabi
complex on any constant curvature pseudo-Riemannian manifold. The tensor
bundles and differential operators between them are defined using
notation and identities from the representation theory of the general
linear group, which are reviewed in Appendix~\ref{app:yt-bkg}. The
differential cochain homotopy operators defined in
Section~\ref{sec:calabi-ops} and Appendix~\ref{app:yt-calabi}, and the
relation of the formal adjoint Calabi complex to the Killing-Yano
operator presented in Section~\ref{sec:calabi-adj} are likely new. Then,
Section~\ref{sec:sheaves} recalls some general notions from sheaf
cohomology, with emphasis on locally constant sheaves. It also covers
the relation between the Calabi cohomology, with various supports, and
the cohomologies of the sheaf of Killing vectors and the sheaf of
Killing-Yano tensors. In Section~\ref{sec:killing} we discuss several
methods for effectively computing the cohomologies of the Killing sheaf,
also outside the constant curvature context. An important application of
the above results is described in Section~\ref{sec:appl}, which uses the
Calabi cohomology to determine the degeneracy subspaces of presymplectic
and Poisson structures of linearized gravity on constant curvature
backgrounds. This application, and its generalizations, constitutes the
main motivation for this work. Finally, Section~\ref{sec:discuss}
concludes with a discussion of the presented results and of how they
could be generalized to other differential complexes of interest in the
mathematical theory of classical and quantum gauge field theories in
physics.

It should be emphasized that the Killing sheaf cohomology can be
identified with the cohomology of the Calabi complex only on
pseudo-Riemannian spaces of constant curvature, where the latter complex
is actually defined. The Killing sheaf itself has a wider domain of
definition. In terms of applications to linearized gravity, the
differential complexes that are to replace the Calabi complex on other
background geometries are still expected to have isomorphic cohomology
to that of the Killing sheaf. So, from that perspective, the Calabi
complex is a particular case study and the Killing sheaf is an object of
more permanent value.

\section{The Calabi complex}\label{sec:calabi}
Below, in Sections~\ref{sec:calabi-tens} and~\ref{sec:calabi-ops}, we
shall explicitly describe the Calabi complex as a complex of
differential operators between tensor bundles on a pseudo-Riemannian
manifold $(M,g)$. Further more, we will explicitly list a corresponding
sequence of differential operators that constitute a cochain homotopy
from the Calabi complex to itself. The cochain maps induced by the
homotopy operators will have the same principal symbol as the tensor
Laplacian $\nabla_a \nabla^a$ induced by the Levi-Civita connection of
the metric tensor $g$, though will differ from it by lower order terms.
This geometric structure is very similar to that of the Hodge theory of
the de~Rham complex on a Riemannian manifold. This structure is used in
the later Section~\ref{sec:sh-res} to show the complex's local
exactness.  Finally, in Section~\ref{sec:calabi-adj}, we will describe
the formal adjoint Calabi complex, with the formal adjoint cochain maps
and homotopies playing roles analogous to the original ones. It turns
out that, just as the Calabi complex resolves the sheaf of Killing
vectors on $(M,g)$, its formal adjoint complex resolves the sheaf of
rank-$(n-2)$ Killing-Yano tensors.

A non-negligible amount of
work~\cite{bbl,gasqui-goldschmidt-fr,goldschmidt-calabi,gasqui-goldschmidt,eastwood,pommaret},
though certainly not a large one, has been done on this differential
complex since the original work~\cite{calabi} of Calabi in 1961. Its
original presentation was in terms of Cartan's moving frame formalism
and much of the subsequent work did not put a strong emphasis on
explicit formulas. Thus, it is a little difficult to find its
presentation in terms of standard covariant derivatives on
tensor bundles in the existing literature. We give such formulas below,
together with a complete sequence of cochain homotopy operators from the
complex to itself and their corresponding cochain maps. These formulas
are apparently new, as their role was played by a more generic, but
somewhat less natural, construction applicable to general elliptic
complexes in~\cite{calabi,bbl,gasqui-goldschmidt-fr,goldschmidt-calabi}.
The advantage of our version is the connection of the homotopy and
cochain maps with the equations of linearized gravity and coincidence,
in low degrees, with other well known related operators, which include
the Killing, linearized Riemann, Bianchi, de~Donder and Ricci trace
operators. One could also argue that our resulting homotopy and cochain
maps are simpler, because they never exceed second differential order
(in contrast to fourth differential order). Further more, we find that
the tensor bundles that constitute the nodes of the complex are best
described as having fibers that carry irreducible representations of
$\GL(n)$, where $n$ is the dimension of the base manifold; moreover, the
principal symbols of the differential operators in the complex are
$\GL(n)$ equivariant maps.  Hence they are independent of the background
metric, which is no longer true for subleading terms. This observation
appears to have escaped the attention of earlier works, thus requiring
some seemingly ad-hoc constructions~\cite{calabi}. A notable exception
is Eastwood~\cite{eastwood}, who also identified the principal symbol
complex as an instance of the general notion of BGG
resolutions~\cite{bgg} in representation theory. Taking advantage of
this connection with representation theory, we explicitly describe the
tensor bundles of the complex and the equivariant principal symbol maps
between them in terms of Young diagrams.

\subsection{Tensor bundles and Young symmetrizers}\label{sec:calabi-tens}
As was mentioned in the Introduction, it is convenient to describe
various tensor bundles involved in the Calabi complex, as well as
various maps between then, in terms of irreducible representations
(\emph{irreps}) of group $\GL(n)$, where $n = \dim M$ is the dimension
of the base manifold $M$. Irreps of $\GL(n)$ are concisely presented
using Young diagrams and corresponding Young tensor symmetrizers. An
excellent reference on this topic is the book~\cite{fulton}, where we
refer the reader for complete details. For an uninitiated reader, we
have briefly summarized the relevant concepts and formulas in
Appendix~\ref{app:yt-bkg}. For the expert reader, it is recommended to
skim the same appendix for the particulars of our notation.

\begin{table}
\caption{%
The table below lists the tensor bundles of the Calabi complex, the
corresponding irreducible $\GL(n)$ representations (labeled by Young
diagrams), and their fiber ranks, for $\dim M = n$. The rank is given by
the famous \emph{hook formula}, which is discussed in
Appendix~\ref{app:yt-bkg}.}
\label{tbl:bundles}
\begin{center}
\begin{tabular}{ccc}
	bundle & Young diagram & fiber rank \\
	\hline \rule[0.5ex]{0pt}{2.5ex} \\[-3ex]
	$C_0M \cong T^*M$
		& \yd{1} & $n$
		\\ \\[-1ex]
	$C_1M \cong S^2M$
		& \yd{2} & $\frac{n(n+1)}{2}$ \\ \\[-1ex]
	$C_2M \cong RM$
		& \yd{2,2} & $\frac{n^2(n^2-1)}{12}$ \\ \\
	$C_3M \cong BM$
		& \yd{2,2,1} & $\frac{n^2(n^2-1)(n-2)}{24}$ \\ \\[-2ex]
	\hline \rule[0.5ex]{0pt}{2.5ex} \\[-3ex]
	$C_lM$
		& \yt{{1}{},{2}{},{\none[\raisebox{-.5ex}{\vdots}]},{l}}
		& $\frac{n^2(n^2-1)(n-2)\cdots(n-l+1)}{2(l+1)l(l-2)!}$ \\ \\[-2ex]
	\hline \rule[0.5ex]{0pt}{2.5ex}
\end{tabular}
\end{center}
\end{table}

Given a base manifold $M$ of dimension $n=\dim M$, we can construct
tensor bundles over $M$ whose fibers carry irreducible representations
of $\GL(n)$. Indeed, we will consider Young symmetrized sub-bundles
$\Y^d T^*M$ of the bundle of covariant $k$-tensors $(T^*)^{\otimes k}
M$, where $d$ is a Young diagram type with $k$ cells.

The Calabi complex, to be introduced in the next section, is a complex
of differential operators between certain tensor bundles over $M$. Let
us denote the corresponding sequence of vector bundles by $C_l M$. More
precisely,
\begin{equation}
	C_0 = T^*, ~~ C_1 = \Y^{(2)} T^*, ~~ C_2 = \Y^{(2,2)} T^*, ~~
	C_l = \Y^{(2,2,1^{l-2})} T^* ~ (l>2) .
\end{equation}
Note that the bundle $C_1 M$ corresponds to symmetric $2$-tensors, which
we will also denote $S^2 M$. Also, as mentioned in the preceding
section, since the bundle $C_2 M$ corresponds to $4$-tensors with the
algebraic symmetries of the Riemann tensor, we will also denote it $RM$.
And the bundle $C_3 M$, also denoted $BM$, corresponds to $5$-tensors
with symmetries of the image of the Bianchi operator applied to a
section or $RM$. The index $l$ essentially counts the number of rows in
the corresponding Young diagram. So, for $l>n$, the number of rows
exceeds the base dimension and the $C_l M$ bundles become trivial. These
tensor bundles, the corresponding Young diagrams and their fiber ranks
are illustrated in Table~\ref{tbl:bundles}.

\subsection{Differential operators}\label{sec:calabi-ops}
Below, given any $n$-dimensional pseudo-Riemannian manifold $(M,g)$ of
constant curvature $k$ (normalized so that the Ricci scalar curvature%
	\footnote{We follow~\cite{wald} for conventions regarding the
	definitions of curvature tensors and scalars.} %
is equal to $k$), we give explicit formulas for the differential
operators, constituting the Calabi complex, as well as formulas for the
differential operators that constitute a cochain homotopy from the
complex to itself and the corresponding induced cochain maps. In
Calabi's original work~\cite{calabi}, the corresponding differential
operators were constructed using an orthogonal coframe formalism. Thus,
it has been difficult to find explicit formulas for these operators in
the tensor formalism that is more prevalent in the physics literature on
relativity. The cochain homotopy operators and the induced cochain maps
coincide, in low degrees, with differential operators well known in the
relativity literature. However, their explicit form in all degrees
appears to be new. Furthermore, we explicitly demonstrate all the
identities between these differential operators that lead to their
homological algebra interpretations. We use a mixture of elementary
arguments, as well as equivariance and standard $\GL(n)$-representation
theoretic identities, unlike Calabi's original proofs~\cite{calabi} that
relied on a somewhat ad hoc algebraic constructions, and unlike the
derivation of Gasqui and
Goldschmidt~\cite{gasqui-goldschmidt-fr,gasqui-goldschmidt} that relied
on the sophisticated theory of Spencer sequences.

First, we define a number of differential operators that will be
convenient for our purposes. For homogeneous differential operators with
constant coefficients, the operator is completely determined by the
principal symbol. In general that is not the case, yet the presence of a
preferred connection on tensor bundles (the $g$-compatible Levi-Civita
connection) still allows us to specify operators by their principal
symbols: the covariant derivative is applied to a tensor $k$-times, the
derivative indices are full symmetrized, and the principal symbol is
applied to the result.

The principal symbol of a $k$-th order differential operator between two
Young symmetrized bundles $\Y T^*$ and $\Y' T^*$ is a linear map between
them that depends polynomially on a covector $p\in T^*$. If the operator
(or just its principal symbol) is $\GL(n)$-equivariant, then the
principal symbol actually corresponds to an intertwiner between the
$\Y^{(k)}\otimes \Y$ and $\Y'$ representations. Such an intertwiner is
non-zero only if $\Y'$ appears in the irrep decomposition of the tensor
product. Moreover if $\Y'$ appears with single multiplicity, the
intertwiner (and hence the principal) symbol is determined uniquely up
to a scalar factor. It is an old result due to Pieri~\cite{fulton} that,
in fact, the decomposition of the product of $\Y^{(k)}\otimes \Y$ into
irreps has only single multiplicities. Not all principal symbols of
importance to us are equivariant. The main source of the lack of
equivariance is the dependence on the metric $g$. However, if the metric
itself is also allowed to transform, the principal symbol becomes
equivariant again.  For instance, if the operator is equivariant in this
way and depends linearly on the metric in covariant form, it corresponds
to an intertwiner between the representations $\Y^{(2)}\otimes \Y^{(k)}
\otimes \Y$ and $\Y'$. Because of the presence of a double tensor
product, Pieri's rule doesn't always apply, so sometimes more
information is necessary to specify the desired intertwiner
unambiguously. As a rule, these ambiguities will be resolved by giving
explicit formulas.

Observe that the all tensor fields defined in
Section~\ref{sec:calabi-tens} correspond to Young diagrams with at most
two columns. We shall refer to the columns as \emph{left} and
\emph{right}. Let $\d_L$ and $\d_R$, the \emph{left} and \emph{right
exterior differentials}, be differential operators that increase by one
the number of boxes in the, respectively, left or right column. They
have equivariant principal symbols. We also define several operators
whose principal symbols involve the metric. Two operators of order $0$
are the \emph{trace} $\tr$ and the \emph{metric exterior product}
$(g\odot -)$, respectively, decreasing (contracting indices between the
two columns) or increasing (multiplying by $g$ and symmetrizing) by one
the number of boxes in each column. Two operators of order $1$ are
\emph{left} and \emph{right codifferentials} $\delta_L$ and $\delta_R$,
which decrease (taking a covariant divergence and resymmetrizing, if
necessary) by one the number of boxes in, respectively, the left or
right column. Finally, we have the \emph{tensor Laplacian} $\square$, a
differential operator of order $2$ that does not alter the Young
symmetry. Explicit formulas for these operators, along with proofs that
they respect the corresponding Young symmetries, are given in
Appendices~\ref{app:yt-ops}, \ref{app:yt-comp}, \ref{app:yt-calabi}.

The differential operators constituting the Calabi complex, as well as
cochain self-homotopy and the induced cochain self-maps fit into the
following diagram:
\begin{equation}\label{eq:calabi-diag}
\begin{tikzcd}
	0 \arrow{r} &
	C_0 \arrow{r}{B_1} \arrow{d}[swap]{P_0} &
	C_1 \arrow{r}{B_2} \arrow{d}[swap]{P_1} \arrow[dashed]{dl}[swap]{E_1} &
	C_2 \arrow{r}{B_3} \arrow{d}[swap]{P_2} \arrow[dashed]{dl}[swap]{E_2}&
	C_3 \cdots \arrow{r}{B_{n}} \arrow{d}[swap]{P_3} \arrow[dashed]{dl}[swap]{E_3}&
	C_{n} \arrow{r} \arrow{d}[swap]{P_n}
	\arrow[dashed]{dl}[swap]{E_{n}} &
	0 \\
	0 \arrow{r} &
	C_0 \arrow{r}{B_1} &
	C_1 \arrow{r}{B_2} &
	C_2 \arrow{r}{B_3} &
	C_3 \cdots \arrow{r}{B_{n}} &
	C_{n} \arrow{r} &
	0
\end{tikzcd} ,
\end{equation}
where for simplicity we have used the symbol $C_l$ to stand for the
space of sections $\Secs(C_lM)$. The operators $B_l$ constitute a
complex, because $B_{l+1} \circ B_l = 0$. The solid arrows in the
diagram commute, $P_{l+1} \circ B_{l+1} = B_l \circ P_l$, so that the
$P_l$ are cochain maps from the complex to itself. These cochain maps,
$P_l = E_{l+1} \circ B_{l+1} + B_l \circ E_l$, are induced by the
homotopy operators $E_n$, which appear as dashed arrows. Below, we give
explicit formulas for each of these operators, discuss these identities,
and relate them to well known differential operators from the literature
on relativity. We follow the notational conventions of
Appendices~\ref{app:yt-bkg}, \ref{app:yt-ops}. In particular, we use $:$
to separate fully antisymmetric tensor index groups belonging to
different columns of the Young diagram, which characterizes the symmetry
type of a given tensor.  However, for simplicity, we also write $g_{ab}
= g_{a:b}$ and $h_{ab} = h_{a:b}$.

\begin{align}
\notag
	B_1[v]_{a:b} &= K[v]_{a:b}
		= \nabla_a v_b + \nabla_b v_a , \\
\notag
	B_2[h]_{ab:cd} &= -2\tilde{R}[h]_{ab:cd} \\
\notag
		&= \left(
		\nabla_{(a}\nabla_{c)} h_{bd} -
		\nabla_{(b}\nabla_{c)} h_{ad} -
		\nabla_{(a}\nabla_{d)} h_{bc} +
		\nabla_{(b}\nabla_{d)} h_{ac} \right) \\
		& \qquad {}
			+ \frac{k}{n(n-1)} (g\odot h)_{ab:cd} , \\
\notag
	B_3[r]_{abc:de}
		&= \bar{B}[r]_{abc:de}
		= \d_L[r]_{abc:de}
		= 3\nabla_{[a} r_{bc]:de} \\
		&= \nabla_a r_{bc:de} + \nabla_b r_{ca:de} + \nabla_c r_{ab:de} , \\
\notag
	B_4[b]_{abcd:ef}
		&= \d_L[b]_{abcd:ef}
		= 4\nabla_{[a} b_{bcd]:ef} \\
		&= \nabla_a b_{bcd:ef} - \nabla_b b_{cda:ef}
			- \nabla_c b_{dab:ef} - \nabla_d b_{abc:ef} , \\
	B_l[b]_{a_1\cdots a_l:bc}
		&= \d_L[b]_{a_1\cdots a_l:bc}
		= l\nabla_{[a_1} b_{a_2\cdots a_l]:bc} \quad (l \ge 3) .
\end{align}

Note that we have introduced some suggestive alternative notations for
operators $B_l$ of low rank. In particular, $B_1 = K$ is the
\emph{Killing} operator. Then, $B_2 = -2\tilde{R}$ is the linearized
\emph{corrected Riemann} curvature operator,%
	\footnote{The same corrected curvature tensor can be obtained by
	linearizing the mixed form $R[g]_{ab}{}^{cd}$ of the Riemann tensor
	and then lowering all indices with the background metric. This
	linearized mixed Riemann tensor was previously used to isolate the
	gauge invariant metric perturbations on de~Sitter space
	in~\cite{roura}. That the linearized corrected Riemann tensor
	annihilates the Killing operator also follows from the classical
	analysis in~\cite{sw-pert}, which noted that the linearization of any
	tensor built only out of the metric and vanishing on the background
	spacetime is invariant under linearized diffeomorphisms.} %
where $R[g+\lambda h] - \bar{R}[g+\lambda h] = \lambda \tilde{R}[h] +
O(\lambda^2)$, with $\bar{R}[g]_{ab:cd} = \frac{k}{n(n-1)} (g_{ac}
g_{bd} - g_{bc} g_{ad})$, cf.~Equation~\eqref{eq:Rb-def} in
Appendix~\ref{app:yt-comp}. The precise relation with the linearized
Riemann tensor operator is
\begin{equation}\label{eq:linR-def}
\dot{R}[h] = -\frac{1}{2}B_2[h] + k \frac{2}{n(n-1)}(g\odot h).
\end{equation}
Note that the identity $R[g] - \bar{R}[g] = 0$ holds precisely when the
metric $g$ has constant curvature $k$.  Finally, $B_3 = \bar{B}$ is the
background \emph{Bianchi} operator, which also happens to coincide with
the left exterior differential $\d_L$. It satisfies the well known
Bianchi identity $\bar{B}[R[g]] = 0$. The operators $B_l$ for $l>3$,
which we may call \emph{higher Bianchi} operators, do not appear to have
been studied in the literature on relativity. So, as mentioned in the
Introduction, the Calabi complex might also be legitimately referred to
as the Killing-Riemann-Bianchi complex.

Now we give mostly elementary arguments for the composition identities
$B_{l+1}\circ B_l = 0$. Recall that if $v$ is a vector field (identified
with a section of $C_0M \cong T^*M$ using the metric), then the Lie
derivative of the metric along $v$ is given by the Killing operator,
$\Lie_v g = K[v]$. Now, suppose that $T[g]$ is any tensor field
covariantly constructed out of the metric and its derivatives. Consider
its linearization $T[g+\lambda h] = T[g] + \lambda \dot{T}[h] +
O(\lambda^2)$. The linearization $\dot{T}$ annihilates the Killing
operator if $T[g] = 0$~\cite{sw-pert}. This fact follows from the fact
that $T[g]$ is itself a tensor field, so that
\begin{equation}
	\Lie_v T[g] = \dot{T}[\Lie_v g] = \dot{T}\circ K[v] .
\end{equation}
Letting $T[g]_{ab:cd} = R[g]_{ab:cd} - \bar{R}[g]_{ab:cd}$ we obtain the
identity $B_2\circ B_1 = -2\tilde{R}\circ K = 0$, since $T[g] = 0$ by
reason of $g$ being of constant curvature equal to $k$. Further, note
that, since the metric is covariantly constant, $\nabla g = 0$, it is
trivial to check that $\bar{B}[\bar{R}[g]] = 0$, for any $g$. Combining
this observation with the Bianchi identity, we find that
$\bar{B}[R[g]-\bar{R}[g]] = 0$, for any $g$. Making the dependence of
$\bar{B} = \bar{B}_g$ on $g$ explicit, the linearization of this
identity gives
\begin{multline}
	\bar{B}_{g+\lambda h}[R[g+\lambda h] - \bar{R}[g+\lambda h]] \\
	= \bar{B}_g[R[g]-\bar{R}[g]]
		+ \lambda(\bar{B}[\tilde{R}[h]] + \dot{B}[h,R[g]-\bar{R}[g]])
		+ O(\lambda^2)
	= 0 ,
\end{multline}
where $\bar{B}_{g+\lambda h}[T] = \bar{B}[T] + \lambda \dot{B}[h,T] +
O(\lambda^2)$. At first order in $\lambda$, we obtain the desired
identity $B_3 \circ B_2 = -2\bar{B} \circ \tilde{R} = 0$. The remaining
identities, $B_{l+1}\circ B_l = \d_L^2 = 0$ for $l>2$, follow from
abstract representation theoretic reasons, described in more detail in
Appendix~\ref{app:yt-comp} and~\ref{app:yt-calabi}.

\begin{align}
\label{eq:calabi-homot-start}
	E_1[h]_a
		&= D[h]_a = \nabla^b h_{ab} - \frac{1}{2} \nabla_a h , \\
	E_2[r]_{a:b}
		&= \tr[r]_{a:b} = r_{ac:b}{}^c , \\
\notag
	E_3[b]_{ab:cd}
		&= \nabla^e b_{e ab:cd}
			+ \frac{1}{2}\nabla^e (b_{c ab:d e} - b_{d ab:c e}) \\
\notag
		& \quad {}
			-\frac{1}{2} (\nabla_c b_{ab e:d}{}^e - \nabla_d b_{ab e:c}{}^e) \\
\notag
		& \quad {} - \frac{1}{2}
			(\nabla_a b_{c b e:d}{}^e - \nabla_a b_{d b e:c}{}^e \\
		& \qquad {}
			+\nabla_b b_{a c e:d}{}^e - \nabla_b b_{a d e:c}{}^e) , \\
\notag
	E_4[b]_{abc:de}
		&= \nabla^f b_{f abc:de}
			+ \frac{1}{3}\nabla^f (b_{d abc:e f} - b_{e abc:d f}) \\
\notag
		& \quad {}
			+ \frac{1}{3} (\nabla_d b_{abc f:e}{}^f - \nabla_e b_{abc f:d}{}^f) \\
\notag
		& \quad {} + \frac{1}{6}
			(\nabla_a b_{d bc f:e}{}^f - \nabla_a b_{e bc f:d}{}^f \\
\notag
		& \qquad {}
			+\nabla_b b_{a d c f:e}{}^f - \nabla_b b_{a e c f:d}{}^f \\
		& \qquad {}
			+\nabla_c b_{ab d f:e}{}^f - \nabla_c b_{ab e f:d}{}^f) , \\
\notag
	E_5[b]_{abcd:ef}
		&= \nabla^i b_{i abcd:ef}
			+ \frac{1}{4}\nabla^i (b_{e abcd:f i} - b_{f abcd:e i}) \\
\notag
		& \quad {}
			- \frac{1}{4}(\nabla_e b_{abcd i:f}{}^i - \nabla_f b_{abcd i:e}{}^i) \\
		& \quad {} - \frac{1}{12}
			(\nabla_{\{e\}} b_{\{abcd\} i:f}{}^i
			-\nabla_{\{f\}} b_{\{abcd\} i:e}{}^i) , \\
\label{eq:calabi-homot-end}
	E_{l+1}[b]_{a_1\cdots a_l:bc}
		&= (\delta_L[b]
			- (-1)^l l^{-1} \d_R \circ \tr[b])_{a_1\cdots a_l:bc} 
			\quad (l\ge 2) .
\end{align}
The notation used in the formula for $E_5$ is defined in
Appendix~\ref{app:yt-bkg}. Note that $E_1 = D$ is the \emph{de~Donder}
operator, used as a linearized gauge fixing condition in the literature
on relativity. Also, if $R[g]$ is the Riemann tensor of the metric $g$,
then $E_2[R[g]] = \tr[R[g]]$ is the corresponding Ricci tensor. The
higher homotopy operators $E_l$ for $l>2$ do not seem to have previously
appeared in the literature on relativity.

\begin{align}
	P_0[v]_a
		&= \square v_a + k\frac{1}{n} v_a , \\
	P_1[h]_{ab}
		&= \square h_{ab} 
			- k \frac{2}{n(n-1)} h_{ab}
			+ 2k \frac{g_{ab} \tr[h]}{n(n-1)}  , \\
	P_2[r]_{ab:cd} 
		&= \square r_{ab:cd}
			- k\frac{2}{n} r_{ab:cd}
			+ 2k\frac{(g\odot \tr[r])_{ab:cd}}{n(n-1)}  , \\
	P_3[b]_{abc:de}
		&= \square b_{abc:de}
			- k\frac{(3n-7)}{n(n-1)} b_{abc:de}
			- 2k\frac{(g\odot \tr[b])_{abc:de}}{n(n-1)}  , \\
	P_4[b]_{abcd:ef}
		&= \square b_{abcd:ef}
			- k\frac{(4n-14)}{n(n-1)} b_{abcd:ef}
			+ 2k\frac{(g\odot \tr[b])_{abcd:ef}}{n(n-1)}  , \\
\notag
	P_l[b]_{a_1\cdots a_l:bc}
		&= \square b_{a_1\cdots a_l:bc}
			- k\frac{(ln-l^2+2)}{n(n-1)} b_{a_1\cdots a_l:bc} \\
		& \qquad {}
			+ (-)^l 2k \frac{(g\odot \tr[b])_{a_1\cdots a_l:bc}}{n(n-1)}
			\quad (l\ge 3) .
\end{align}

Note our notation $\square = \nabla^a\nabla_a$ for the tensor Laplacian,
which is also known as the d'Alambertian in Lorentzian signature. The
operator $P_0 = D\circ K$ is gives the wave-like residual gauge
condition such that the perturbation $h = K[v]$ satisfies the de~Donder
gauge condition $D[h] = 0$ in linearized gravity. The operator $P_1 =
\tr{} \circ (-2\tilde{R}) + K \circ D$ is the wave-like operator of the
linearized Einstein equations for gravitational perturbations $h$ in
de~Donder gauge $D[h] = 0$. These two operators are well known and can
be found (or their close analogs can) for instance in
\cite[Sec.7.5]{wald} and more they appeared in
in~\cite{fewster-hunt,hack-lingrav,bdm}. The higher cochain maps and the
corresponding identities appear to be new. Though, the identity $P_2 =
E_3 \circ \bar{B} - 2\tilde{R}\circ E_2$ is related to the non-linear
wave equations satisfied by the Riemann and Weyl tensors on any vacuum
background, sometimes known as the \emph{Penrose wave equation}. For
linearized fields, a related equation is sometimes known as the
\emph{Lichnerowicz Laplacian}. For more details, see
references~\cite{ryan}, \cite[Sec.1.3]{lichnerowicz},
\cite[Sec.7.1]{chr-kl}, \cite[Exr.15.2]{mtw}, \cite[Eq.35]{bcjr}.

\begin{rem}\label{rmk:elliptic}
It is worth noting that we refer to the operators $P_l$ as wave-like
because the principal symbol of $P_l$ has the same principal symbol as
the tensor Laplacian $\square = \nabla^a \nabla_a$, on Lorentzian
manifolds also known as the d'Alambertian or wave operator, which is a
hyperbolic differential operator. Note that the principal symbol of
$P_l$ is determined only by the principal symbols of the $B_l$ and
$E_l$. The principal symbols of $B_l$ are metric independent, while
those of $E_l$ depend on the metric $g$ of the background
pseudo-Riemannian manifold $(M,g)$. However, we are actually free to
choose any metric, say $g'$ that is different from $g$, to construct the
cochain homotopy operators, say $E'_l$. The principal symbol induced
cochain maps $P'_l = E'_{l+1}\circ B_{l+1} + B_l \circ E'_l$ will then
still only depend on one metric, $g'$, and be equal to the principal
symbol of the tensor Laplacian $\square'$ defined with respect to $g'$.
Thus, if we choose $g'$ to be Riemannian, we can induce cochain
homotopies $P'_l$ that are elliptic. The operators $P'_l$ will of course
differ from the $P_l$ by terms of lower differential order that would
depend on both $g$ and $g'$. This remark will be very useful in
Proposition~\ref{prp:calabi-exact} in the discussion of the local
exactness of the Calabi complex.
\end{rem}

\subsection{Formal adjoint complex}\label{sec:calabi-adj}
Given a linear differential operator $f\colon \Secs(E)\to \Secs(F)$,
between vector bundles $E\to M$ and $F\to M$, its \emph{formal adjoint}
is a linear differential operator $f^*\colon \Secs(\tilde{F}^*) \to
\Secs(\tilde{E}^*)$, where where we have used the notation for the
bundle $\tilde{V}^* \cong V^* \otimes_M \Lambda^n M$ of \emph{dual
densities} of a vector bundle $V\to M$, defined as the tensor product of
the its linear dual bundle $V^*\to M$ with that of densities $\Lambda^n
M\to M$ on the base manifold if dimension $\dim M = n$. The formal
adjoint operator is defined to be the unique differential operator such
that a \emph{Green formula} holds,
\begin{equation}\label{eq:gen-adj}
	\psi \cdot f[\xi] - f^*[\psi] \cdot \xi = \d G[\psi,\xi] ,
\end{equation}
for any $\psi \in \Secs(\tilde{F}^*)$, $\xi\in \Secs(E)$, and some
bilinear bidifferential operator
\begin{equation}
	G\colon \Secs(\tilde{F}^* \times_M E) \to \Secs(\Lambda^{n-1} M).
\end{equation}
A formal adjoint operator always exists and is
unique~\cite{anderson-small,anderson-big,tarkhanov}.

In the presence of background pseudo-Riemannian metric $g$ on $M$, we
can canonically identify the trivial bundle $\R\times M$ with $\Lambda^n
M$, via multiplication by the canonical volume form $\eps_{a_1\cdots
a_n}$ with respect to $g$ ($\eps \in \Omega^n(M)$), and also $V \cong
V^*$ for any tensor bundle $V\to M$, by lowering and raising indices
with $g$, thus also canonically identifying $V \cong \tilde{V}^*$.
Below, we will take formal adjoints with respect to this identification.
Recall the identity~\cite{levi-civita}
\begin{equation}
	\eps^{a a_2\cdots a_n} \eps_{b a_2\cdots a_n} = (-1)^s(n-1)! \delta^a_b
\end{equation}
(where $s$ counts the number
of minuses in the signature of the metric $g$, with $s=1$ for Lorentzian
metrics with mostly-plus convention) and define
\begin{equation}
	G^a = \frac{(-1)^s}{(n-1)!} \eps^{a a_2 \cdots a_n} G_{a_2\cdots a_n}
\end{equation}
so that $G_{a_2\cdots a_n} = \eps_{a a_2 \cdots a_n} G^a$. The right
hand side of the formal adjoint equation~\eqref{eq:gen-adj} can then be
rewritten as
\begin{equation}
	(\d G)_{a_1\cdots a_n}
	= \frac{(-1)^s}{n!} \eps_{a_1\cdots a_n} \eps^{a b_2\cdots b_n}
		n \nabla_{a} G_{b_2\cdots b_n}
	= \eps_{a_1\cdots a_n} \nabla_a G^a,
\end{equation}
with the whole equation becoming
\begin{equation}\label{eq:adj}
	\psi\cdot f[\xi] - f^*[\psi]\cdot \xi = \nabla_a G^a[\psi,\xi] ,
\end{equation}
where the dot indicates contraction of indices using the metric $g$
between two tensors of the same index structure.

With this notation, the formal adjoint Calabi complex
$(C_\bullet,B_\bullet^*)$ fits into the following diagram:
\begin{equation}\label{eq:dual-calabi-diag}
\begin{tikzcd}
	0 \arrow[<-]{r} &
	C_0 \arrow[<-]{r}{B_1^*} \arrow[<-]{d}[swap]{P_0^*} &
	C_1 \arrow[<-]{r}{B_2^*} \arrow[<-]{d}[swap]{P_1^*} \arrow[<-,dashed]{dl}[swap]{E_1^*} &
	C_2 \arrow[<-]{r}{B_3^*} \arrow[<-]{d}[swap]{P_2^*} \arrow[<-,dashed]{dl}[swap]{E_2^*}&
	C_3 \cdots \arrow[<-]{r}{B_{n}^*} \arrow[<-]{d}[swap]{P_3^*} \arrow[<-,dashed]{dl}[swap]{E_3^*}&
	C_{n} \arrow[<-]{r} \arrow[<-]{d}[swap]{P_n^*} \arrow[<-,dashed]{dl}[swap]{E_{n}^*} &
	0 \\
	0 \arrow[<-]{r} &
	C_0 \arrow[<-]{r}{B_1^*} &
	C_1 \arrow[<-]{r}{B_2^*} &
	C_2 \arrow[<-]{r}{B_3^*} &
	C_3 \cdots \arrow[<-]{r}{B_{n}^*} &
	C_{n} \arrow[<-]{r} &
	0
\end{tikzcd} ,
\end{equation}
where we have identified $\tilde{C}_i^* \cong C_i$ using the background
metric. Note that all the analogous identities are satisfied, the solid
arrows in the diagram commute and the dashed arrows are homotopy
operators inducing the vertical cochain maps, $P_i^* = B_{i+1}^* \circ
E_{i+1}^* + E_i^* \circ B_i^*$. The main difference is that the
$B_\bullet$ now decrease the degree index by one instead of decreasing
it. The usual numbering convention can be achieved by relabelling, but
we shall not do so here, expecting that no confusion will arise.

Recall that the final differential operator $B_n$ of the Calabi complex
is
\begin{equation}
	B_n[b]_{a_1 \cdots a_n : bc}
		= \d_L[b]_{a_1\cdots a_n : bc}
		= n \nabla_{[a_1} b_{a_2 \cdots a_n]:bc} ,
\end{equation}
where $b\in \Secs(C_{n-1} M)$. To compute its formal adjoint, let $c\in
\Secs(C_n M)$ and consider first the identity, derived in
Appendix~\ref{app:yt-adj},
\begin{multline}\label{eq:dlDl-adj}
	\nabla_a (c^{a a_2\cdots a_n : bc} b_{a_2\cdots a_n : bc}) \\
	= \frac{1}{n} c^{a a_2\cdots a_n : bc} \d_L[b]_{a a_2\cdots a_n : bc}
			+ \delta_L[c]^{a_2\cdots a_n : bc} b_{a_2\cdots a_n : bc} .
\end{multline}
Note that the operators $\d_L$ and $\delta_L$ specifically produce
tensors of the appropriate Young type. Therefore, the formal adjoint
operator $B_n^*$ is given by the formula
\begin{align}
	B_n^*[c]_{a_2\cdots a_n:bc}
	&= -\frac{1}{n}\delta_L[c]_{a_2\cdots a_n:bc} \\
	&= -\frac{1}{n}\nabla^a c_{a a_2\cdots a_n:bc}
		- \frac{2}{n(n-1)} \nabla^a c_{[b|a_2\cdots a_n:|c]a} ,
\end{align}
with the Green form represented by $G^a[c,b] = \frac{1}{n} c^{a
a_2\cdots a_n:bc} b_{a_2\cdots a_n:bc}$.

While this operator $B_n^*$ may look unfamiliar, after a further local
invertible transformation the equation $B_n^*[c] = 0$ becomes equivalent
to the well known \emph{rank-$(n-2)$ Killing-Yano equation}. Let us
define a rank-$(n-2)$ anti-symmetric tensor $y^{c_3\cdots c_n}$ such
that
\begin{align}
	c_{a_1\cdots a_n:bc}
		&= \eps_{a_1\cdots a_n} y^{c_3\cdots c_n} \eps_{bc c_3\cdots c_n} , \\
	y^{c_3\cdots c_n}
		&= \frac{1}{2(n-2)!(n-1)!} \eps^{a_1\cdots a_n}
			\, c_{a_1\cdots a_n:bc} \,
			\eps^{bc c_3\cdots c_n} .
\end{align}
It is straightforward to check using the hook formula
(Appendix~\ref{app:young}) that the tensor $c$ of Young type
$(2,2,1^{n-2})$ has the same number of independent components as the
tensor $y$ of Young type $(1^{n-2})$. To transform the equation
satisfied by $c$ into the Killing-Yano equation satisfied by $y$, we
will need the following identities, which follow from the general
properties of the $\eps$ tensor~\cite{levi-civita}:
\begin{align}
	\eps^{a a_2\cdots a_n} c_{a' a_2\cdots a_n:bc} \eps^{bc c_3\cdots c_n}
		&= 2(n-2)!(n-1)! \delta^a_{a'} y^{c_3\cdots c_n} , \\
	\eps^{a a_2\cdots a_n} c_{b a_2\cdots a_n:a'c} \eps^{bc c_3\cdots c_n}
		&= (n-1)!^2 y^{b_3\cdots b_n}
			\delta^{[a}_{a'} \delta^{c_3}_{b_3} \cdots \delta^{c_n]}_{b_n} .
\end{align}
Contracting one $\eps$ tensor with each index group of the equation
$B_n^*[c] = 0$ we get
\begin{align}
	0
	&= \eps^{a a_2\cdots a_n} B^*_n[c]_{a_2\cdots a_n:bc}
		\eps^{bc c_3\cdots c_n} \\
\notag
	&= -\frac{1}{n}\nabla^{a'} \eps^{a a_2\cdots a_n}
			c_{a' a_2\cdots a_n:bc} \eps^{bc c_3\cdots c_n} \\
	&\quad {}
		- \frac{2}{n(n-1)} \nabla^{a'} \eps^{a a_2\cdots a_n}
			c_{b a_2\cdots a_n:ca'} \eps^{bc c_3\cdots c_n} \\
	&= -\frac{2}{n}(n-1)!(n-2)! \left(\nabla^a y^{c_3\cdots c_n}
		- \nabla^{[a} y^{c_3\cdots c_n]}\right) .
\end{align}
Note that the derivative $\nabla^a y^{c_3\cdots c_n}$ takes values in
the tensor bundle of Young type $(1)\otimes (1^{n-2})$. Using the
well-known Littlewood-Richardson rules~\cite{fulton,lrr} this
representations decomposes into the direct sum $(1)^{n-1} \oplus
(2,1^{n-3})$. Note that the antisymmetrization of the above equation
gives zero. Thus, the independent components of the equation satisfied
by $y$ take values in a tensor bundle of Young type $(2,1^{n-3})$, which
has two columns, of lengths $n-2$ (filled with indices belonging to $y$) and
$1$ (filled with index belonging to $\nabla$). It is also well-known that
this representation can be isolated by antisymmetrizing along the columns
and symmetrizing any two indices between the columns. In our case, the
antisymmetrization has no effect ($y$ is already antisymmetric) and the
symmetrization, after lowering all indices, gives the equation
\begin{equation}
	KY[y]_{a c_3 \cdots c_n} = \nabla_{(a} y_{c_3)\cdots c_n} = 0 ,
\end{equation}
which is none other than the \emph{rank-$(n-2)$ Killing-Young equation},
whose solutions are called \emph{rank-$(n-2)$ Killing-Young} tensors or
\emph{Killing $(n-2)$-forms}~\cite{stepanov}. We refer to the
differential operator $KY$ as the Killing-Young operator. So, in the
same sense that the Calabi complex constitutes the compatibility complex
of the Killing equation on a constant curvature background, so does the
formal adjoint Calabi complex for the rank-$(n-2)$ Killing-Yano
equation on the same background.
% XXX: reference for adjoint of formally exact => formally exact?

\subsection{Equations of finite type, twisted de~Rham complex}\label{sec:tw-dr}
The Killing and Killing-Yano equations, which lie at the base of the
Calabi and its formal adjoint differential complexes, are well known
examples of partial differential equations of \emph{finite
type}~\cite{goldschmidt-lin,spencer,pommaret}. That is, in any
neighborhood of a point $x\in M$ they admit only a finite dimensional
space of solutions. Each solution is fully determined by its value and
finitely many derivatives at $x$. For the Killing and Killing-Yano
equations only the first derivatives are required. This is a strong kind
of unique continuation. Such equations are called \emph{regular} if the
dimension of the solution space in a sufficiently small neighborhood of
a point $x\in M$ is independent of $x$. That number may, however, differ
from the dimension of the global solution space, which can be strictly
smaller in the presence of topological or geometric obstructions to
continuing local solutions to global ones.

Regular equations of finite type have a very simple existence theory.
Let $F\to M$ and $E\to M$ be two vector bundles, together with a
differential operator $e\colon \Secs(F) \to \Secs(E)$ of order $l$ such
that the equation $e[\psi] = 0$, for $\psi\in \Secs(F)$, is finite type
and regular. This means that there exists an integer $k$ such that the
knowledge of $j^k\psi(x)$ for any $x\in M$ is sufficient to determine the
components of all higher jets of $\psi$ at $x$. Prolongation of the
equation to order $k$ (Appendix~\ref{app:jets}) gives the bundle map
$p^{k-l} e\colon J^k F \to J^{k-l} E$. By the regularity hypothesis, the
map is of constant rank, so its kernel $V = \ker p^{k-l} e \sse J^k F$
is a vector bundle over $M$. Since all higher derivatives of a solution
$\psi$ at $x$ are uniquely determined by $j^k\psi(x)$ and $j^k\psi$ only
takes values in $V$, there is a unique $n$-dimensional hyperplane in
$T_{x,v}V$ that is tangent to the graph of a solution $\psi$ such that
$j^k\psi(x) = (x,v)$. These hyperplanes define an $n$-dimensional
distribution on the total space of the bundle $V$ and it is
straightforward to check that this distribution is involutive (Lie
brackets of vector fields valued in the distribution remain valued in
the distribution). Thus, by the theorem of Frobenius~\cite{lang}, $V$ is
foliated by $n$-dimensional leaves tangent to the given hyperplane
distribution. Locally, these leaves are precisely the graphs of
solutions to the equation $e[\psi] = 0$. Thus the rank $\rk V$ is
precisely the dimension of the local solution space on any sufficiently
small, connected open set in $M$.

As we have already mentioned, both the Killing and Killing-Yano
operators, $K\colon \Secs(T^*M) \to \Secs(S^2M)$ and $KY\colon
\Secs(\Lambda^{n-2}M) \to \Secs(\Y^{(2,1^{n-1})}T^*M)$, define finite
type equations. By the virtue of their covariance, they are also regular
on any pseudo-Riemannian symmetric space, which includes constant
curvature backgrounds. Furthermore, on constant curvature spaces, the
dimensions of their local solution spaces are $\rk V_K = \rk V_{KY} =
n(n+1)/2$~\cite{stepanov}.

The $n$-dimensional hyperplane distribution on $V$ and the resulting
foliation described above can also be described in another way, namely
as a flat linear connection on $V\sse J^kF$~\cite[Sec.2.1.3]{morita}.
The connection is linear because the original equation $e[\psi] = 0$ is
itself linear. A linear connection on $V\to M$ can alternatively be
described by a first order differential operator $D\colon \Secs(V) \to
\Secs(T^*M \otimes V)$ defined by the property
\begin{equation}
	D[\omega j^k\psi] = \d \omega \otimes j^k\psi ,
\end{equation}
for any $\omega \in C^\oo(M)$ and solution $\psi\in \Secs(F)$ of
$e[\psi]$, where its $k$-jet is treated as a section $j^k\psi \colon M
\to V$. That is, a section $\phi\in \Secs(V)\sse \Secs(J^kF)$ is
constant on an open set $U\sse M$ iff it coincides with the $k$-jet of a
solution of $e[\psi] = 0$ on $U$. So, it is clear that the equations
$e[\psi] = 0$ and $D[\phi] = 0$ are equivalent (their spaces of local
solutions are locally isomorphic). As discussed in
Appendix~\ref{app:jets} this means that there exist differential
operators $f$, $f'$, $g$, $g'$, $p$ and $q$, which fit into the
following diagram (again, for brevity we use the bundle symbols to stand
in for their spaces of sections)
\begin{equation}\label{eq:conn-equiv}
\begin{tikzcd}
	F \ar[swap]{r}{e} \ar[shift left]{d}{f}
		& E \ar[shift left]{d}{f'} \ar[dashed,bend right,swap]{l}{p} \\
	V \ar{r}{D} \ar[shift left]{u}{g}
		& T^*M\otimes V \ar[shift left]{u}{g'} \ar[dashed,bend left]{l}{q}
\end{tikzcd}
\end{equation}
and satisfy the following identities:
\begin{align}
	D \circ f &= f' \circ e , & g \circ f &= \id + p \circ e , \\
	e \circ g &= g' \circ D , & f \circ g &= \id + q \circ D .
\end{align}
We have already seen that on solutions, the map $f$ simply agrees with
the $k$-jet extension operator $j^k$. Thus, as a differential operator
of order $k$, it can be chosen to be any projection of $J^k F$ to its
subspace $V$. The choice of this projection then determines the
differential operator $f'$. The differential operators $g$ and $g'$ are
constructed in similar ways, making sure that $f$ and $g$ are mutual
inverses on solutions. The freedom in the choice of $f$, $f'$, $g$ and
$g'$ also determine the operators $p$ and $q$.

When it comes to a specific case, say the Killing or Killing-Yano
equation, its equivalence to a local constancy condition with respect to
a connection can be made explicit only once the solutions are themselves
explicitly known. Thus this equivalence is mostly of theoretical, though
non-negligible, interest.

Having defined the flat vector bundle $(V,D)$ corresponding to a regular
equation of finite type, there is a standard procedure to construct a
differential complex associated to it. It is called the \emph{twisted
de~Rham complex} associated to $(V,D)$,
\begin{equation}\label{eq:tw-dr}
\begin{tikzcd}[column sep=scriptsize]
	0 \arrow{r} &
	V \arrow{r}{D} &
	\Lambda^1 M \otimes V \arrow{r}{D} &
	\Lambda^2 M \otimes V \cdots \arrow{r}{D} &
	\Lambda^n M \otimes V \arrow{r} &
	0 ,
\end{tikzcd}
\end{equation}
where $D$ has been extended to a twisted de~Rham differential, defined
on sections of $\Lambda^k M \otimes V$ by the condition
\begin{equation}
	D[\omega \otimes \psi]
		= \d\omega \otimes \psi + (-1)^k \omega \wedge D\psi ,
\end{equation}
for any $\omega\in \Secs(\Lambda^k M)$ and $\psi \in \Secs(V)$, where we
recall that $D\psi$ is a section of $T^*M\otimes V = \Lambda^1 M \otimes
V$ and apply the wedge product of forms in the obvious way.

\begin{rem}\label{rmk:tw-dr}
Locally (on sufficiently small contractible open sets), this twisted
de~Rham complex consists $\rk V$ copies of the ordinary de~Rham complex.
Globally, of course, if the base manifold $M$ is not simply connected,
the twisted de~Rham complex $(\Lambda^\bullet M \otimes V, D)$ will
differ from $\rk V$ copies of the ordinary de~Rham complex
$(\Lambda^\bullet M,\d)$ because of the possible non-trivial bundle
structure of $V\to M$ or the non-trivial monodromy $D$ (parallel
transport with respect to $D$ along closed loops). The importance of the
twisted de~Rham complex will become clear in Section~\ref{sec:sheaves}
where we discuss the connection between the cohomology of differential
complexes and sheaf cohomology.
\end{rem}

For later convenience, we shall denote the twisted de~Rham complexes
associated to the Killing and Killing-Yano equations, respectively, by
$(\Lambda^\bullet M\otimes V_K, D_K)$ and $(\Lambda^\bullet M\otimes
V_{KY}, D_{KY})$.

\section{Cohomology of locally constant sheaves}\label{sec:sheaves}

The main reasons for introducing some of the general sheaf and sheaf
cohomology machinery below is are two fold. First, we have made a
connection between the abstract notion of sheaf cohomology and the
cohomology of a differential complex. A priori, computing the cohomology
of differential complex is a very hard problem, because it involves
solving partial differential equations. On the other hand, because of
the flexibility of the general machinery of sheaf cohomology, it may be
computable in some effective way, for instance, by reducing it to a
problem in finite dimensional linear algebra. The canonical example of
where this connection can be leveraged is the computation of de~Rham
cohomology groups of a manifold $M$ using the equivalent (through sheaf
theoretic machinery) computation of the simplicial (or cellular)
cohomology of a finite triangulation (or cell decomposition) of $M$. The
second reason is that the ideas that have been introduced give us some
tools to explicitly show that the cohomologies of two different
differential complexes are isomorphic as long as both complexes are
\emph{formally exact}, \emph{locally exact} and \emph{resolve} the same
sheaf in degree-$0$ (this terminology is introduced below).

\subsection{Locally constant sheaves}\label{sec:sh-lc}
Recall from Section~\ref{sec:tw-dr} that a regular linear differential
equation of finite type has only a finite dimensional space of local
solutions, with this dimension being constant over the base manifold. It
so happens that, from an abstract point of view, it is convenient to
view these local solutions as a \emph{locally constant sheaf} of vector
spaces. A \emph{sheaf} $\F$ of vector spaces on a topological space
$M$~\cite{bredon,ks} is an assignment $U\mapsto \F(U)$ of a vector space
(of \emph{local sections} over $U$, $\F(\varnothing) = 0$) to each open
$U\sse M$ satisfying the following axioms: \emph{(restriction)} for any
inclusion of opens $U\sse V$ there exist linear restriction maps
$\F(V)\to \F(U)$, also written $f\mapsto f|_U$, such that $U\sse U$
induces the identity map and $U\sse V\sse W$ induces $\F(W) \to \F(U)$
in agreement with the composition $\F(W) \to \F(V) \to \F(U)$;
\emph{(descent)} any pair of opens $U$ and $V$ induces an exact sequence
$0 \to \F(U\cup V) \to \F(U)\times \F(V) \to \F(U\cap V) \to 0$, where
the first map is $f \mapsto (f|_U,f|_V)$ and the second one is $(f,g)
\mapsto f|_{U\cap V} - g|_{U\cap V}$. We write $\Secs(\F) = \Secs(M,\F)
= \F(M)$ for the vector space of \emph{global sections} of the sheaf
$\F$. A \emph{sheaf} is called \emph{locally constant} when the number
$\odim \F_x = \max_{U\ni x} \dim \F(U)$, where $U$ ranges over connected
open neighborhoods of $x\in M$, is finite and does not depend on $x$, so
we can write $\odim\F = \odim\F_x$. Since $\dim \F(U)$ can only decrease
for larger connected $U$, for any $x\in M$ there exists a connected
neighborhood $U$ of $x$ such the vector spaces of local sections over
smaller connected neighborhoods stabilize (the restriction map becomes
an isomorphism), so that we can write $\F(U) \cong \bar{F}$ for some
fixed vector space $\bar{F}$ that we call the \emph{stalk} of $\F$.
Clearly, $\dim \bar{F} = \odim \F$.  Also, $\F$ is called
\emph{constant} when it is locally constant and $\Secs(\F) \cong
\bar{F}$.

Given a vector bundle $F\to M$, the assignment $\F(U) = \Secs(F,U)$ of
local sections of $F$ over each open $U\sse M$ defines a sheaf $\F$ on
$M$, called the \emph{sheaf of (germs of) sections} of $F\to M$.
Similarly, it is straightforward to check that, given another vector
bundle $E\to M$ and a linear differential operator $e\colon \Secs(F) \to
\Secs(E)$, the sets $\S_e(U) = \{ \psi\in \Secs(F,U) \mid e[\psi] = 0
\}$ of solutions of the partial differential equation $e[\psi] = 0$ also
define a sheaf $\S_e$ on $M$, called the \emph{solution sheaf} of
$e\colon \Secs(F) \to \Secs(E)$. Following the preceding discussion of
equations of finite type, it should be clear that solution sheaves $\K =
\S_K$ (the \emph{Killing sheaf}) and $\KY = \S_{KY}$ (the
\emph{Killing-Yano sheaf}) of the Killing and Killing-Yano equations are
locally constant, provided the background pseudo-Riemannian manifold is
chosen such that these equations are regular. Another important example
is the constant sheaf $\R_M = \S_\d$ of locally constant functions,
which solve the equation $\d f = 0$, $f\in C^\oo(M)$ and $\d$ the
de~Rham differential.

Sheaves are important because every sheaf $\F$ (of vector spaces) on $M$
automatically comes with an abstract notion of \emph{sheaf cohomology}
(vector spaces) $H^p(M,\F)$, called the $p$-th or degree-$p$ cohomology
of $\F$, or of $M$ with coefficients in $\F$. Moreover, all classical
cohomology theories from algebraic topology can be identified with the
cohomologies of certain sheaves. Further, some superficially different
looking cohomologies theories may be connected through the fact that
they are both equivalent to the sheaf cohomology of the same sheaf. In
particular, the classical simplicial, cellular, singular, \v{C}ech and
de~Rham cohomologies of a manifold $M$ all
coincide~\cite{bott-tu,bredon,ks} because they are each equivalent to
the cohomology of $M$ with coefficients in the sheaf $\R_M$ of locally
constant functions.

The intrinsic definition of sheaf cohomology is somewhat involved and
not entirely intuitive (unless one is already intimately familiar with
\v{C}ech cohomology and the notion of local coefficients). Fortunately,
the intrinsic definition can be relegated to standard
references~\cite{bredon,ks} in favor of an equivalent but more practical
definition using \emph{acyclic resolutions}. To explain further, we need
to introduce some terminology. A \emph{complex} of sheaves of vector
spaces
\begin{equation}\label{eq:sh-cplx}
\begin{tikzcd}
	\cdots \arrow{r} &
	\F_i \arrow{r} &
	\F_{i+1} \arrow{r} &
	\cdots
\end{tikzcd}
\end{equation}
consists of an assignment of linear maps $\F_i(U) \to \F_{i+1}(U)$ to
each open $U\sse M$, in a way consistent with restriction maps, such
that we have a complex of vector spaces of local sections (two
successive maps compose to zero)
\begin{equation}
\begin{tikzcd}
	\cdots \arrow{r} &
	\F_i(U) \arrow{r} &
	\F_{i+1}(U) \arrow{r} &
	\cdots
\end{tikzcd}
\end{equation}
for each open $U\sse M$. A local section in $\F_i(U)$ that is in the
kernel of the corresponding map is called a \emph{cocycle} and a local
section in $\F_i(U)$ that is in image of the corresponding map is called
a \emph{coboundary}. A sheaf complex is \emph{exact} when, for each
$x\in M$, open neighborhood $U\sse M$ of $x$ and local section $\alpha
\in \F_i(U)$, there exists a possibly smaller and $\alpha$-dependent
open neighborhood $U'\sse U$ of $x$ such that $\alpha|_{U'}$ is a
coboundary. For a complex of sheaves, like~\eqref{eq:sh-cplx}, we could
define its \emph{cohomology sheaves} $\H^i(\F_\bullet)$ (distinct from
\emph{sheaf cohomology}, to be defined later), by starting with the
assignment $\H^i(\F_\bullet)(U) = \ker(\F_i(U)\to \F_{i+1}(U)) / \im
(\F_{i-1}(U) \to \F_i(U))$, which may not produce a sheaf but only a
\emph{presheaf}, and applying the \emph{sheafification} construction to
it. We will not go into the details of how sheafification turns
presheaves into sheaves here, but they can be found in standard
references~\cite{bredon,ks}. It suffices to point out that given a sheaf
complex in non-negative degrees, $0 \to \F_0 \to \F_i \to \cdots$, the
vector space $\H^0(\F_\bullet)(U) \sse \F_0(U)$ consists of all cocycle
local sections. In the sequel, we shall only need to refer to such
cohomology sheaves in degree-$0$. Given a sheaf $\F$, if $\F_i \to
\F_{i+1}$ is a complex of sheaves such that $\F_i = 0$ for $i < 0$,
$\H^0(\F_\bullet) = \F$, and $\H^i(\F_\bullet) = 0$ for $i > 0$, we call
it a \emph{resolution} of the sheaf $\F$.

In the sequel, we shall only consider sheaves of sections of vector
bundles or of solution of some liner PDE and only complexes of sheaves
where maps between the vector spaces of local sections are induced by
restrictions of differential operators, for which the compatibility with
restrictions is automatically satisfied.

\subsection{Acyclic resolution by a differential complex}\label{sec:sh-res}
The de~Rham complex~\cite{bott-tu} is the canonical example of a complex
of sheaves of sections of vector bundles (differential forms on $M$),
with maps induced by differential operators (de~Rham differentials). The
Poincar\'e lemma then demonstrates that this complex of sheaves is
exact. For simplicity, we shall call a \emph{differential complex}
$(F_\bullet,f_\bullet)$ a sequence of vector bundles $F_i\to M$ and
differential operators $f_i \colon \Secs(F_{i-1}) \to \Secs(F_i)$
satisfying $f_i \circ f_{i-1} = 0$, while implicitly setting $F_{-1} =
0$ and $f_0 = 0$. Given a differential complex, it is natural to define
its cohomology vector spaces to be the cohomology of the cochain complex
of global sections, $H^i(F_\bullet,f_\bullet) =
H^i(\Secs(F_\bullet),f_\bullet)$, which we also refer to as the
cohomology with \emph{unrestricted supports}. Since differential
operators do not increase supports, we can equally consider the
cohomology of the differential complex with \emph{compact supports},
defined as $H_c^i(F_\bullet,f_\bullet) =
H^i(\Secs_c(F_\bullet),f_\bullet)$. A differential complex naturally
define a complex $0 \to \F_0 \to \F_1 \to \cdots$ of sheaves of sections
of these bundles, $\F_i(U) = \Secs(F_i,U)$. A differential complex is
said to be \emph{locally exact} if it defines an exact complex of
sheaves. Local exactness is a very strong property that is crucial in
the relation of the cohomology of a differential complex to sheaf
cohomology, which we discuss next.

In general, given a complex of sheaves $\F_i \to \F_{i+1}$, we call it
an \emph{injective resolution} of a sheaf $\F$ if it is a
\emph{resolution} of $\F$ (namely, $\F_i = 0$ for $i<0$, it is exact
except for $\H^0(\F_\bullet) = \F$), and each $\F_i$ is
\emph{injective}. The injectivity condition is somewhat technical. The
same can be said for the fact that every sheaf has an injective
resolution. So we will not go into them here and defer to standard
references instead~\cite{bredon,ks}. We will need these notions only for
the following definition.  The \emph{degree-$i$ sheaf cohomology} vector
spaces $H^i(\F) = H^i(M,\F)$, also called the \emph{degree-$i$
cohomology of $M$ with coefficients in $\F$}, as the cohomology vector
space of the complex of global sections of any injective resolution
$\F_i \to \F_{i+1}$ of $\F$, $H^i(\F) = H^i(\Secs(\F_\bullet))$. It is
important to note that sheaf cohomology is well defined. It does not
depend on the chosen injective resolution, because the injectivity
condition implies the existence of a homotopy equivalence between the
complexes of global sections of any two such resolutions, thus forcing
their cohomologies to be isomorphic. This is another technical fact that
we shall not go into here.

Instead, we make note of yet another technical fact that provides a
practical way to compute sheaf cohomology. For that, we need two more
definitions. A sheaf $\F$ is called \emph{acyclic} if $H^i(\F) = 0$ for
all $i>0$, though as usual the degree-$0$ cohomology $H^0(\F) \cong
\Secs(\F)$ is isomorphic to the vector space of global sections of $\F$.
A sheaf $\F$ on $M$ is called \emph{soft} if for any closed $A\sse M$
the restriction maps $\F(M) \to \F(A)$ are surjective, where $\F(A) =
\bigcap_{U\sse A} \F(U)$ with $U$ ranging over all open sets that
contain the closed set $A$. In other words, given an open $U\sse M$ and
a closed subset $A \sse U$, a local section on $U$ can always be
extended to a global one on $M$ without modification on $A$, but
possibly modified on $U\setminus A$. What is really important for us is
the following
\begin{prop}\label{prp:sh-res}
(i) If $\F$ is a sheaf on $M$, and $\F_i \to \F_{i+1}$ is a resolution
of $\F$ by acyclic sheaves (\emph{acyclic resolution}), then $H^i(M,\F)
\cong H^i(\Secs(M,\F_\bullet))$. (ii) Any soft sheaf on $M$ is acyclic.
(iii) Given a vector bundle $F\to M$, the sheaf $\F$ of sections of $F$
is soft.
\end{prop}
\begin{proof}
Any standard discussion of sheaf cohomology establishes (i) and
(ii)~\cite{bredon,ks}. On the other hand, (iii) is simply a restatement
of the well known Whitney extension theorem for smooth
functions~\cite[Thm.2.3.6]{hoermander-I}.
\end{proof}
Note that the complex of sheaves corresponding to a differential complex
then automatically consists of acyclic sheaves. The above proposition
essentially tells us that, given a resolution of some sheaf $\F$ on a
manifold $M$ by a locally exact differential complex
$(F_\bullet,f_\bullet)$, the sheaf cohomology of $\F$ and the cohomology
of the differential complex will coincide, $H^i(\F) \cong
H^i(F_\bullet,f_\bullet)$. This observation will be particularly
important later in Corollary~\ref{cor:calabi-sheaf}.

Next, we discuss some conditions ensuring that the cohomologies of two
given differential complexes are isomorphic. As we have now seen, local
exactness is a very strong and useful property, unfortunately it can be
difficult to check in practice. Two weaker notions of exactness exist
that are easier to check in practice.  To formulate them, we refer to
the notions of \emph{jets} and \emph{jet bundles}, together with
associated constructions like \emph{prolongations} and \emph{principal
symbols}, all briefly recalled in Appendix~\ref{app:jets}. Given a
sequence of vector bundles $F_i$ and a complex of linear differential
operators $f_{i}\colon F_{i-1} \to F_i$, each of order $k_i$, their
prolongations define a complex of vector bundle morphisms,
\begin{equation}
\begin{tikzcd}[column sep=large]
	\cdots \arrow{r} &
	J^l F_{i-1} \arrow{r}{p^{l_i} f_i} &
	J^{l_i} F_i \arrow{r}{p^{l_{i+1}} f_{i+1}} &
	J^{l_{i+1}} F_{i+1} \arrow{r} &
	\cdots ,
\end{tikzcd}
\end{equation}
with $l_i = l - k_i$ and $l_{i+1} = l - k_i - k_{i+1}$, for each
sufficiently large $l$. The differential complex is said to be
\emph{formally exact} if the above compositions are exact, as linear
bundle maps over $M$, for any values of $l$ and $i$ for which they are
defined. On the other hand, given $(x,p) \in T^*M$, the principal
symbols of the differential operators $f_i$ define a complex of linear
maps between the fibers of $F_i$ at $x$,
\begin{equation}
\begin{tikzcd}[column sep=large]
	\cdots \arrow{r} &
	F_{i-1,x} \arrow{r}{\sigma_{x,p} f_i} &
	F_{i,x} \arrow{r}{\sigma_{x,p} f_{i+1}} &
	F_{i+1,x} \arrow{r} &
	\cdots .
\end{tikzcd}
\end{equation}
The differential complex is said to be \emph{elliptic} if the above
complex is exact for every $(x,p)\in T^*M$, $p\ne 0$. These two weaker
notions are distinct~\cite{smith}. Formal exactness is a good hypothesis
for showing that differential operators factor in certain ways. On the
other hand, ellipticity is a condition that can be used to prove local
exactness, via the method of parametrices and fundamental solutions.
However, the general question of determining necessary and sufficient
conditions for local exactness for differential complexes is a difficult
and still open problem. The main conjecture is sometimes known as
\emph{Spencer's conjecture}: a formally exact, elliptic complex is
locally exact~\cite{spencer,smith,shl-tarkh}. On the other hand, some
supplementary sufficient conditions are known for an elliptic complex to
be locally exact. A prominent condition of this kind is known as the
\emph{$\delta$-estimate}~\cite[Sec.1.3.13]{tarkhanov}, which first
appeared in the works of Singer, Sweeney and MacKichan~\cite{spencer}.

\begin{prop}\label{prp:tw-dr-exact}
The twisted de~Rham complex associated to the flat bundle $(V,D)$
defined by a regular differential equation of finite type, defined in
Equation~\eqref{eq:tw-dr}, is formally exact, elliptic and locally
exact.
\end{prop}
\begin{proof}
As noted in Remark~\ref{rmk:tw-dr}, the twisted de~Rham complex is
locally (on sufficiently small contractible open sets) equivalent to
$\rk V$ copies of the ordinary de~Rham complex. To see the equivalence,
it suffices to locally choose a $D$-flat basis frame for $V$. Since all
of the desired properties, formal exactness, ellipticity and local
exactness are purely local, it suffices to check them for the ordinary
de~Rham complex. It is well known that each of these properties does
hold for the de~Rham complex, having served as a model example for each.
Formal exactness and ellipticity are discussed, for instance,
in~\cite{spencer,pommaret,tarkhanov} and~\cite[\textsection
XIX.4]{hoermander-III}. On the other hand, local exactness is
essentially the content of the Poincar\'e lemma~\cite{bott-tu}.

There is another way to establish local exactness that bypasses the
Poincar\'e lemma and does not require an explicit local choice of a
$D$-flat basis frame for $V$. In particular, as discussed for instance
in the given references, local exactness and ellipticity are independent
of such a choice. Then, local exactness follows provided the initial
operator of the complex, the connection operator $D\colon \Secs(V) \to
\Secs(T^*M\otimes V)$, satisfies the $\delta$-estimate.  According to
Example~1.3.58 of~\cite{tarkhanov}, any linear connection operator
satisfies the $\delta$-estimate. Hence, by Theorem~1.3.61
of~\cite{tarkhanov}, the twisted de~Rham complex is locally exact.
\end{proof}

As is well known in homological algebra, cochain maps and homotopies
between them are important concepts, the first because they descend to
cohomology, the second because equivalence up to homotopy descends to
isomorphism on cohomology. When dealing with differential complexes, it
becomes important to distinguish the case where the cochain maps and
homotopies are defined by differential operators. The most important
notion we will need is that of a \emph{formal homotopy equivalence}. Let
$(F_\bullet,f_\bullet)$ and $(G_\bullet,g_\bullet)$ be two differential
complexes. They are said to be \emph{formally homotopy equivalent}
provided there exist differential operators $e_i$, $h_i$, $u_i$ and
$v_i$ fitting into the diagram (we use the bundles to stand in for their
spaces of sections)
\begin{equation}
\begin{tikzcd}
	\cdots \ar{r} &
	F_{i-1} \ar[swap]{r}{f_i}
		\ar[shift left]{d}{u_{i-1}}
		\ar[<-,shift right,swap]{d}{v_{i-1}}
		\ar[bend right,dashed]{l} &
	F_i \ar[swap]{r}{f_{i+1}}
		\ar[shift left]{d}{u_i}
		\ar[<-,shift right,swap]{d}{v_i}
		\ar[bend right,dashed,pos=.57,swap]{l}{e_i} &
	F_{i+1} \ar{r}
		\ar[shift left]{d}{u_{i+1}}
		\ar[<-,shift right,swap]{d}{v_{i+1}}
		\ar[bend right,dashed,pos=.4,swap]{l}{e_{i+1}} &
	\cdots
		\ar[bend right,dashed]{l} \\
	\cdots \ar{r} &
	G_{i-1} \ar{r}{g_i}
		\ar[bend left,dashed]{l} &
	G_i \ar{r}{g_{i+1}}
		\ar[bend left,dashed,pos=.57]{l}{h_i} &
	G_{i+1} \ar{r}
		\ar[bend left,dashed,pos=.4]{l}{h_{i+1}} &
	\cdots
		\ar[bend left,dashed]{l}
\end{tikzcd} ,
\end{equation}
where the squares composed of solid arrows commute (cochain map
condition on $u_i$ and $v_i$) and the dashed arrows are homotopy
operators with respect to which $u_i$ and $v_i$ are quasi-inverses, $v_i
\circ u_i - \id = e_{i+1}\circ f_{i+1} + f_i\circ e_i$ and $u_i \circ
v_i - \id = h_{i+1}\circ g_{i+1} + g_i \circ h_i$.

\begin{lem}\label{lem:f-exact}
Consider two differential complexes $(F_\bullet,f_\bullet)$ and
$(G_\bullet,g_\bullet)$ that start in degree $0$, also denote the
corresponding complexes of sheaves of sections as $\F_i \to \F_{i+1}$
and $\GG_i \to \GG_{i+1}$. Suppose that both differential complexes are
formally exact, except in degree $0$. Further, suppose that the
equations $f_1[\phi] = 0$ and $g_1[\gamma] = 0$, with $\phi \in
\Secs(F_0)$ and $\gamma \in \Secs(G_0)$, are equivalent, or in other
words the degree-$0$ cohomology sheaves are isomorphic to some given
sheaf $\F \cong \H^0(\F_\bullet) \cong \H^0(\GG_\bullet)$.

(i) Then there there exists a formal homotopy equivalence between these
differential complexes and their cohomologies are isomorphic, both with
unrestricted and compact supports (or any other kind of restriction on
supports):
\begin{gather}
	H^i(F_\bullet,f_\bullet) \cong H^i(G_\bullet,g_\bullet)
	\quad \text{and} \quad
	H^i_c(F_\bullet,f_\bullet) \cong H^i_c(G_\bullet,g_\bullet) .
\end{gather}

(ii) If one of the differential complexes is locally exact, then both
are locally exact and their cohomologies both compute the sheaf
cohomology of $\F$:
\begin{equation}
	H^i(M,\F) \cong H^i(F_\bullet,f_\bullet) \cong H^i(G_\bullet,g_\bullet) .
\end{equation}
\end{lem}
\begin{proof}
(i) Equivalence of the equations $f_1[\phi]=0$ and $g_1[\gamma]=0$ means
(Appendix~\ref{app:jets}) that there exist differential operators, say
$u_0\colon \Secs(F_0) \to \Secs(G_0)$ and $v_0\colon \Secs(G_0)\to
\Secs(F_0)$, such that $v_0 \circ u_0[\phi] = 0$ whenever $f_1[\phi] =
0$ and such that $u_0 \circ v_0 [\gamma] = 0$ whenever $g_1[\phi] = 0$.
In other words, there exist differential operators $e_1\colon \Secs(F_1)
\to \Secs(F_0)$ and $h_1\colon \Secs(G_1) \to \Secs(G_0)$ such that $v_0
\circ u_0 = e_1 \circ f_1$ and $u_0 \circ v_0 = h_1 \circ g_1$. These
differential operators are the initial step in establishing the desired
formal homotopy equivalence.

We proceed by a standard induction argument from homological algebra (in
fact, a version of this argument proves the independence of sheaf
cohomology from the injective resolution used to compute it). Assume
that all the desired differential operators have been defined up to
$e_i$, $h_i$, $u_{i-1}$ and $v_{i-1}$, which also satisfy the desired
identities. We can easily verify the identities
\begin{align}
	(g_i \circ u_{i-1}) \circ f_{i-1}
		&= (g_i \circ g_{i-1}) \circ u_{i-2} = 0 , \\
	(f_i \circ v_{i-1}) \circ g_{i-1}
		&= (f_i \circ f_{i-1}) \circ v_{i-2} = 0 ,
\end{align}
which together with the formal exactness of the compositions $f_i
\circ f_{i-1} = 0$ and $g_i \circ g_{i-1} = 0$ imply the
factorizations $g_i \circ u_{i-1} = u_i \circ f_i$ and $f_i \circ
v_{i-1} = v_i \circ g_i$, for some differential operators $u_i \colon
\Secs(F_i) \to \Secs(G_i)$ and $v_i \colon \Secs(G_i) \to \Secs(F_i)$
(see Appendix~\ref{app:jets}). Further, we can also verify the
identities
\begin{align}
	(v_i \circ u_i - \id - f_i \circ e_i) \circ f_i
		&= (v_i \circ g_i) \circ u_{i-1} - f_i - f_i\circ e_i \circ f_i \\
		&= f_i \circ (v_{i-1} \circ u_{i-1}) - f_i - f_i\circ e_i \circ f_i
		= 0 , \\
	(u_i \circ v_i - \id - g_i \circ h_i) \circ g_i
		&= (u_i \circ f_i) \circ v_{i-1} - g_i - g_i\circ h_i \circ g_i \\
		&= g_i \circ (u_{i-1} \circ v_{i-1}) - g_i - g_i\circ h_i \circ g_i
		= 0 ,
\end{align}
which again together with formal exactness imply the factorizations $v_i
\circ u_i - \id - f_i \circ e_i = e_{i+1} \circ f_{i+1}$ and $u_i \circ
v_i - \id - g_i \circ h_i = h_{i+1} \circ g_{i+1}$, for some
differential operators $e_{i+1} \colon \Secs(F_{i+1}) \to \Secs(F_i)$
and $h_{i+1} \colon \Secs(G_{i+1}) \to \Secs(G_i)$. This concludes the
inductive step.

Now, let us consider the cohomology of these complexes,
$H^i(F_\bullet,f_\bullet) = H^i(\Secs(F_\bullet), f_\bullet)$ and
$H^i(G_\bullet,g_\bullet) = H^i(\Secs(G_\bullet), g_\bullet)$. As is
well known from homological algebra, a homotopy equivalence (of which a
formal homotopy equivalence is a special kind) induces an isomorphism in
cohomology: $H^i(F_\bullet, f_\bullet) \cong H^i(G_\bullet,g_\bullet)$.
However, if the operators implementing the homotopy equivalence are
differential operators, as in this case, we can replace unrestricted
sections $\Secs(-)$ by sections with compact supports $\Secs_c(-)$, so
that $H^i_c(F_\bullet,f_\bullet) = H^i(\Secs_c(F_\bullet), f_\bullet)$
and $H^i_c(G_\bullet,g_\bullet) = H^i(\Secs_c(G_\bullet), g_\bullet)$.
The homotopy equivalence of the resulting complexes still holds because
differential operators do not increase supports, and so we still have an
isomorphism in cohomology: $H^i_c(F_\bullet, f_\bullet) \cong
H^i_c(G_\bullet,g_\bullet)$. Incidentally, instead of compact supports,
any other family of supports would do as well.

(ii) By the local exactness hypothesis, both differential complexes
provide resolutions of the sheaf $\F$ (which happens to be isomorphic to
the solution sheaves $\S_{f_1} = \H^0(\F_\bullet)$ and $\S_{g_1} =
\H^0(\GG_\bullet)$). Then, by Proposition~\ref{prp:sh-res}, these
resolutions are acyclic and hence the corresponding cohomologies with
unrestricted supports compute the sheaf cohomology of $\F$. This
concludes the proof.
\end{proof}

\subsection{Generalized Poincar\'e duality}\label{sec:dual}
In Section~\ref{sec:sh-res}, we discussed how the cohomology
$H^i(F_\bullet,f_\bullet)$ of a differential complex can, under optimal
conditions, be equated with the cohomology $H^i(\F)$ of the sheaf
resolved by $(F_\bullet,f_\bullet)$. However, even under optimal
conditions, this connection breaks down if we consider cohomology
$H^i_c(F_\bullet,f_\bullet)$ with compact (or some other family of)
supports instead of unrestricted ones. What we discuss below is a way to
relate cohomology with compact supports to that with unrestricted
supports, a kind of Poincar\'e duality.

For the de~Rham complex on a manifold $M$, $\dim M = n$, a well known
formulation of Poincar\'e duality is the isomorphism $H^p(M) \cong
H^{n-p}_c(M)^*$~\cite[Rmk.5.7]{bott-tu} between the linear dual of
cohomology in degree-$p$ and compactly supported cohomology in
degree-$(n-p)$. This isomorphism is induced by the existence of a
non-degenerate natural pairing between $p$-forms and $(n-p)$-forms on
$M$ and its non-degenerate descent to cohomology. The goal of this
section is to leverage the properties of the Calabi complex and its
formal adjoint complex that were discussed in the preceding section to
demonstrate a generalized version of Poincar\'e duality, which
effectively computes the cohomology with compact supports in terms of
sheaf cohomology.

There are two ways to establish generalized Poincar\'e duality for a
differential complex $(F_\bullet,f_\bullet)$ that would be applicable to
the case of the Calabi complex and its formal adjoint. One of them,
discussed in Section~\ref{sec:tw-dr-dual}, relies on the fact that the
corresponding complex of sheaves resolves the sheaf of solutions of a
regular differential equation of finite type (a locally constant sheaf).
This method is somewhat more elementary. The other, discussed in
Section~\ref{sec:elliptic-dual}, works for any elliptic complex, but
requires some results from functional analysis and distribution theory.
Either of these results, as will be shown in
Section~\ref{sec:calabi-cohom}, can be applied to prove generalized
Poincar\'e duality for the Calabi complex and its formal adjoint
complex.

\subsubsection{Twisted de~Rham complex}\label{sec:tw-dr-dual}
First, we will discuss the twisted de~Rham complex, as introduced in
Section~\ref{sec:tw-dr}. The results will then apply to the Calabi
complex and its formal adjoint by virtue of Lemma~\ref{lem:f-exact}.
The strategy is straightforward and reproduces the logic of the proofs
of the ordinary Poincar\'e duality, cf.~\cite[\textsection 5]{bott-tu},
\cite[Ch.11]{spivak}, or~\cite[Sec.V.4]{ghv}. First, generalized
Poincar\'e duality is shown to hold on contractible open patches. Then,
given a ``good cover'' of the manifold consisting of such patches, we
use a version of the Mayer-Vietoris exact sequence as an inductive step
to conclude that generalized Poincar\'e duality also holds on the entire
manifold.

First, recall that we denote the fiber of the vector bundle $V\to M$ by
$\bar{V}$. Then, $\bar{V}^*$ is the fiber of the dual vector bundle
$V^*\to M$. We are interested in the relation between the cohomology of
the twisted de~Rham complex $H^i(\Lambda^\bullet M \otimes V,D)$ and the
compactly supported cohomology of the formal adjoint complex, which
happens to be $(\Lambda^\bullet M \otimes V^*,D)$, where the connection
$D$ has been extended to $V^*\to M$ by the rule $\d(\xi\cdot \psi) =
(D\xi) \cdot \psi + \xi \cdot (D\psi)$, with $\xi\in \Secs(V^*)$ and
$\psi\in \Secs(V)$. Presuming that $M$ is oriented, which is a
prerequisite for integrating top-degree forms, there is a duality
pairing between elements of $\Secs(\Lambda^p M \otimes V)$ and
$\Secs_c(\Lambda^{n-p} M \otimes V^*)$ given by the formula
\begin{equation}\label{eq:tw-dr-pair}
	\langle \xi, \psi \rangle = \int_M \langle \xi \wedge \psi \rangle ,
\end{equation}
where $\langle (\alpha \otimes \xi) \wedge (\beta \otimes \psi) \rangle
= (\alpha \wedge \beta) \otimes (\xi \cdot \psi)$.  The formal adjoint
relation is established (up to signs) for $\xi\in
\Secs(\Lambda^{n-p-1}M\otimes V^*)$ and $\psi \in \Secs(\Lambda^p M
\otimes V)$ by the identity
\begin{equation}
	\d\langle \xi \wedge \psi \rangle
	= \langle (D\xi) \wedge \psi \rangle
		- (-1)^{n-p} \langle \xi \wedge (D\psi) \rangle .
\end{equation}
\begin{lem}\label{lem:triv-dual}
Let $U\sse M$ be an oriented contractible open set. Then, generalized
Poincar\'e duality holds, $H^p(\Lambda^\bullet M\otimes V|_U,D) \cong
H^{n-p}_c(\Lambda^\bullet M\otimes V^*|_U,D)^*$, because all of the
cohomology spaces vanish except $H^0(\Lambda^\bullet M\otimes V,D) \cong
\bar{V}$ and $H^n_c(\Lambda^\bullet M\otimes V^*,D) \cong \bar{V}^*$.
\end{lem}
\begin{proof}
As we have already noted in the proof of
Proposition~\ref{prp:tw-dr-exact}, a choice of a locally $D$-flat basis
frame for $V$ over $U\sse M$ identifies the twisted de~Rham complex with
$\rk V$ copies of the usual de~Rham complex. Since $U$ is contractible,
such a choice is always possible. Moreover, the
pairing~\eqref{eq:tw-dr-pair} reduces to the usual pairing between forms
and compactly supported forms of complementary degrees on an oriented
manifold. Thus, we can easily conclude that
\begin{align}
	H^p(\Lambda^\bullet M\otimes V|_U, D) &= H^p(U)\otimes \bar{V} , \\
	H^{n-p}_c(\Lambda^\bullet M\otimes V^*|_U, D) &= H^{n-p}_c(U) \otimes \bar{V}^* .
\end{align}
Recalling that, for contractible $U$, $H^p(U) = 0$ except for $H^0(U) =
\R$ and $H^{n-p}_c(U) = 0$ except for $H^n_c(U) = \R$, concludes the
proof.
\end{proof}

\begin{lem}[Mayer-Vietoris]\label{lem:mv}
Consider two open subsets $U,W \sse M$. We have the following long exact
sequences in cohomology with unrestricted and compact supports, which we
shall for brevity denote as $H^i(-) = H^i(\Lambda^\bullet M \otimes
V|_{-}, D)$ and $H^i_c(-) = H^i_c(\Lambda^\bullet M \otimes V^*|_{-},
D)$:
\begin{equation}
\begin{tikzcd}[column sep=scriptsize,row sep=scriptsize]
	0 \arrow{r} &
	H^0(U\cup W) \arrow{r} &
	H^0(U)\oplus H^0(W) \arrow{r}\arrow[draw=none]{d}[name=Z,shape=coordinate]{}&
	H^0(U\cap W)
		\arrow[rounded corners,
			to path={ -- ([xshift=2ex]\tikztostart.east)
			|- (Z) [near end]\tikztonodes
			-| ([xshift=-2ex]\tikztotarget.west)
			-- (\tikztotarget)}]{dll} \\
	&
	H^1(U\cup W) \arrow{r} &
	H^1(U)\oplus H^1(W) \arrow{r} &
	H^1(U\cap W) \arrow{r} & \cdots
\end{tikzcd}
\end{equation}
\begin{equation}
\begin{tikzcd}[column sep=scriptsize,row sep=scriptsize]
	0 \arrow{r} &
	H^0_c(U\cap W) \arrow{r} &
	H^0_c(U)\oplus H^0_c(W) \arrow{r}\arrow[draw=none]{d}[name=Z,shape=coordinate]{}&
	H^0_c(U\cup W)
		\arrow[rounded corners,
			to path={ -- ([xshift=2ex]\tikztostart.east)
			|- (Z) [near end]\tikztonodes
			-| ([xshift=-2ex]\tikztotarget.west)
			-- (\tikztotarget)}]{dll} \\
	&
	H^1_c(U\cap W) \arrow{r} &
	H^1_c(U)\oplus H^1_c(W) \arrow{r} &
	H^1_c(U\cup W) \arrow{r} & \cdots
\end{tikzcd}
\end{equation}
\end{lem}
\begin{proof}
Both long exact sequences in cohomology follow from short exact
sequences of cochain complexes. These short exact sequences, where for
brevity we write $\Secs^i(-) = \Secs(\Lambda^i M\otimes V|_-)$ and
$\Secs^i_c = \Secs(\Lambda^i M \otimes V^*|_-)$ are
\begin{equation}
\begin{tikzcd}[column sep=scriptsize,row sep=scriptsize]
	0 \arrow{r} &
	\Secs^i(U\cup W) \arrow{r} &
	\Secs^i(U) \oplus \Secs^i(W) \arrow{r} &
	\Secs^i(U\cap W) \arrow{r} &
	0 ,
\end{tikzcd}
\end{equation}
\begin{equation}
\begin{tikzcd}[column sep=scriptsize,row sep=scriptsize]
	0 \arrow{r} &
	\Secs^i_c(U\cap W) \arrow{r} &
	\Secs^i_c(U) \oplus \Secs^i_c(W) \arrow{r} &
	\Secs^i_c(U\cup W) \arrow{r} &
	0 .
\end{tikzcd}
\end{equation}
In the first sequence, the maps are restrictions, $\alpha \mapsto
(\alpha|_U,\alpha|_W)$ and $(\alpha,\beta) \mapsto (\alpha|_{U\cap W} -
\beta|_{U\cap W})$. The exactness follows from the usual ability to
restrict and glue together smooth sections over open regions, also known
as their sheaf property. In the second sequence, the maps are
extensions by zero, $\alpha \mapsto (\alpha^U_0, \alpha^W_0)$ and
$(\alpha,\beta) \mapsto \alpha^{U\cup W}_0 - \beta^{U\cup W}_0$. The
exactness follows from the existence of a smooth partition of unity
adapted to the cover of $U\cup W$ by $U$ and $W$.

These maps are clearly compatible with the connection differential
operator $D$ and so are cochain maps. The general connection between
short exact sequences of cochain complexes and long exact sequences in
cohomology (Appendix~\ref{app:homalg}) gives the desired long exact
sequences and concludes the proof.
\end{proof}

\begin{prop}\label{prp:tw-dr-dual}
Given a flat vector bundle $(V,D)$ on an oriented $n$-dimensional
orientable manifold $M$, the unrestricted cohomology $H^p =
H^p(\Lambda^\bullet M \otimes V, D)$ of the associated twisted de~Rham
complex and the compactly supported cohomology $H^{n-p}_c =
H^{n-p}_c(\Lambda^\bullet M\otimes V^*, D)$ of its formal adjoint
complex satisfy generalized Poincar\'e duality:
\begin{equation}
	H^p \cong (H^{n-p}_c)^* .
\end{equation}
\end{prop}
Note the asymmetry of the isomorphism. The reverse identity $(H^p)^*
\cong H^{n-p}_c$ also holds when the cohomology vector spaces are finite
dimensional, but in general may not when they are infinite dimensional.
\begin{proof}
In this proof, we shall use induction over a special kind of open cover of
$M$. An open cover $(U_k)$ of $M$ is called \emph{good} if it is
locally finite, every nonempty finite intersection $U_{k_0} \cap \cdots
\cap U_{k_m}$ is diffeomorphic to $\R^n$, and it is closed under finite
intersections. In particular, each of the $U_k$ is itself diffeomorphic
to $\R^n$ and thus contractible. Good covers are known to exist for any
manifold~\cite[Thm.5.1]{bott-tu}. Inducing an orientation on each
element of the cover from the orientation on $M$,
Lemma~\ref{lem:triv-dual} establishes the desired duality relation for
any $U_k$ and thus the initial step of the inductive argument.

Next, we show, provided the desired duality relation holds on any finite
union $U_{k_0} \cup \cdots \cup U_{k_{m-1}}$ of $m$ sets, that it also
holds on any finite union $U_{k_0} \cup \cdots \cup U_{k_m}$ of $m+1$
sets as well. Of course, we take all such unions to be oriented in a way
compatible with the global orientation on $M$. Let $U = U_{k_m}$, $W =
U_{k_0} \cup \cdots \cup U_{k_{m-1}}$ and notice that both $W$ and
$W\cap U$ are finite unions of $m$ sets from the cover (recall that the
cover is closed under intersections). The fact that the
pairing~\eqref{eq:tw-dr-pair}, well defined on a given oriented, open
$U\sse M$, descends to cohomology means that we always have a mapping
$H^p(U) \to H_c^{n-p}(U)^*$, which may or may not be an isomorphism.
It is in fact an isomorphism on $U$ and, by the inductive hypothesis,
also on $W$ and $W\cap U$. Combining the long exact sequences of
Lemma~\ref{lem:mv} for $W$ and $U$ together with these maps and
isomorphisms, we obtain the following diagram (notice the arrow reversal
by linear duality in the second row):
\begin{equation*}
% Solution for scaling tikz-cd diagrams:
%   http://tex.stackexchange.com/a/139036
\begin{tikzpicture}[baseline=(diag).base]
\node[scale=.64] (diag) at (0,0){
\begin{tikzcd}[column sep=small]
	H^p(W\cup U) \ar{r} \ar[equal]{d} &
	H^p(W) \oplus H^p(U) \ar{r} \ar[equal]{d} &
	H^p(W\cap U) \ar{r} \ar{d} &
	H^{p+1}(W\cup U) \ar{r} \ar[equal]{d} &
	H^{p+1}(W) \oplus H^{p+1}(U) \ar[equal]{d}
	\\
	H^{n-p}_c(W\cup U)^* \ar{r} &
	H^{n-p}_c(W)^* \oplus H^{n-p}_c(U)^* \ar{r} &
	H^{n-p}_c(W\cap U)^* \ar{r} &
	H^{n-p-1}_c(W\cup U)^* \ar{r} &
	H^{n-p-1}_c(W)^* \oplus H^{n-p-1}_c(U)^*
\end{tikzcd}
};
\end{tikzpicture}
\end{equation*}
Thus, by the $5$-lemma (Appendix~\ref{app:homalg}), the map in the
center of the diagram is also an isomorphism and the inductive step is
established.

The only problem remaining is that a good cover is not always finite
(though it can be chosen to be finite for compact manifolds). There is a
way around that, however. Using a similar argument, one can show that
the desired duality holds also on disjoint countable unions of finite
unions of covering sets. It is at this stage that the asymmetry between
the cohomologies with unrestricted and compact supports appears. Then,
provided the manifold is second countable, one can choose a much
coarser, yet finite, cover $(U'_k)$. The key property of this cover is
that each of the non-empty finite intersections $U'_{k_0} \cap \cdots
\cap U'_{k_m}$ is itself either a finite union of sets from $(U_k)$ or a
disjoint countable union of those. The same $5$-lemma argument then
shows that the desired generalized Poincar\'e duality relation $H^p
\cong (H^{n-p}_c)^*$ holds on all of $M$. The technical details of this
argument can be found in~\cite[Sec.V.4]{ghv}.
\end{proof}

\subsubsection{Elliptic complexes and Serre duality}\label{sec:elliptic-dual}
Now we will discuss generic elliptic complexes, of which both the Calabi
and the twisted de~Rham complexes are special cases. The result is
essentially the same, though clearly more general. The arguments are
somewhat less elementary and rely on some background in functional
analysis and an result originally due to Serre~\cite{serre}. The Serre
duality method also gives some more information. Namely, that the
cohomology does not change if we replace smooth functions by
distributions with the same supports. Serre's original work was in the
context of the Dolbeault complex in the theory of several complex
variables. A good exposition of this result in the setting of general
elliptic complexes can be found in~\cite{tarkhanov}.

At this point it is convenient to recall some basic facts of
distribution theory~\cite{schwartz,treves,reed-simon}. Recall that, for
any vector bundle $F\to M$, we can interpret $\Secs(F)$ and $\Secs_c(F)$
as locally convex topological vector spaces, with the Whitney weak
Fr\'echet topology for the former and an inductive limit over supports
of similar Fr\'echet topologies for the latter, with the limit topology
still locally convex but no longer Fr\'echet (metrizable). These are the
usual topologies used in the theory of distributions. The spaces of
\emph{distributional sections} $\Do(F)$ and $\D(F)$ of $F$, with
respectively compact and unrestricted supports, are defined as
topological duals endowed with the strong topology (the usual
distributional topology), $\Do(F) = \Secs(\tilde{F}^*)^*$ and $\D(F) =
\Secs_c(\tilde{F}^*)^*$. Recall that $\tilde{F}^* = \Lambda^n M \otimes
F^*$ is the densitized dual bundle; the densitized dual of the
densitized dual is the original bundle. It so happens that, if we stick
with the strong topology for dual spaces, the topological dual of
$\Do(F)$ is again $\Secs(\tilde{F}^*)$ and that of $\D(F)$ is
$\Secs_c(\tilde{F}^*)$. So the spaces of smooth and distributional
sections are \emph{reflexive} (with respect to the strong topology).
Using the natural pairing
\begin{equation}
	\langle \psi , \alpha \rangle = \int_M \psi\cdot \alpha
\end{equation}
between $\psi\in \Secs(F)$ and $\alpha \in \Secs_c(\tilde{F}^*)$,
well-defined provided $M$ is oriented, we have the natural inclusions
$\Secs(F) \subset \D(F)$ and $\Secs_c(F) \subset \Do(F)$. By the
\emph{Schwartz kernel theorem}, the continuous maps $G\colon
\Secs_c(F_1) \to \D(F)$ are in bijection with \emph{bidistributions},
elements $G\in \D(F_2 \boxtimes \tilde{F}^*_2)$, where $F_2 \boxtimes
\tilde{F}^*_2 \to M\times M$ is the bundle with total space $F_2 \times
\tilde{F}^*_1$ and the obvious projection onto its base, by the formula
\begin{equation}
	(G\psi)(x) = \int_M G(x,y)\cdot \psi(y) .
\end{equation}
Let $\pi_1(x,y) = y$ and $\pi_2(x,y) = x$ denote the two projections
$M\times M \to M$. We say that a bidistribution $G\in \D(F_2\boxtimes
\tilde{F}^*_1)$ is \emph{properly supported} if $\pi_1\colon \supp G\to
M$ is a proper map (the preimage of a compact set is compact).
Differential operators define properly supported bidistributions,
because their support lies on the diagonal of $M\times M$ by the crucial
property that differential operators preserve supports. On the other
hand, properly supported bidistributions need not preserve supports,
though they still map compactly supported sections to compactly
supported distributions.  The amount by which the support of the image
grows depends on the size of the support of the bidistribution in
$M\times M$.

Once we have introduced distributional sections, we can extend to them
many operators that were previously defined only on smooth functions.
For instance, any linear differential operator $f\colon \Secs(F) \to
\Secs(E)$ between vector bundles $F\to M$ and $E\to M$ can be extended
to act on distributions, $f\colon \Secs'(F) \to \Secs'(E)$ or even
$f\colon \Secs_c'(F) \to \Secs_c'(E)$, according to the following
formula:
\begin{equation}
	\langle f[\alpha] , \psi \rangle = -\langle \alpha , f^*[\psi] \rangle ,
\end{equation}
for any $\psi \in \Secs_c(\tilde{F}^*)$ and $\alpha \in \D(F)$, where
$f^*$ is the formal adjoint of $f$ and $\langle -, - \rangle$ is the
natural dual pairing between sections and distributions. Since this
natural pairing is non-degenerate, it suffices to define $f$ on the
larger domain. Any other operator defined on smooth sections for which
the above formula applies can also be extended to distributions,
possibly with a restriction on their supports.

In particular, the operators of a differential complex
$(F_\bullet,f_\bullet)$ can be extended to distributional sections. Then
we can consider the cohomology of the complex in distributional
sections, $H^i(\D(F_\bullet), f_\bullet)$, which may a priori be
different from its cohomology in smooth sections $H^i(F_\bullet,
f_\bullet) = H^i(\Secs(F_\bullet), f_\bullet)$, and similarly with
compact supports. Below we shall see some sufficient conditions for the
cohomologies in smooth and distributional sections to coincide.

A crucial concept in the general theory of differential complexes is
that of a \emph{parametrix}~\cite[Ch.2]{tarkhanov}. Let the vector
bundles $F_i$ with differential operators $f_i\colon \Secs(F_{i-1}) \to
\Secs(F_i)$ constitute a differential complex $(F_\bullet,f_\bullet)$ on
$M$. Then, a \emph{parametrix} is a sequence of bidistributions $G_i \in
\D(F_{i-1} \boxtimes \tilde{F}^*_{i})$ such that
\begin{equation}\label{eq:parametrix}
	\id_i - Q_{i} = G_{i+1} \circ f_{i+1} + f_i \circ G_i ,
\end{equation}
where $\id_i \colon \Secs_c(F_i) \to \Secs_c(F_i)$ is the identity map
and $Q_i \in \Secs(F_{i+1} \boxtimes \tilde{F}^*_i) \subset \D(F_{i+1}
\boxtimes \tilde{F}^*_i)$ is a smooth bidistribution. We say that the
parametrix is \emph{properly supported} if each $G_i$ is a properly
supported bidistribution. Obviously, if each $G_i$ is properly
supported, then so is each $Q_i$.
\begin{prop}\label{prp:ell-param}
Let $(F_\bullet, f_\bullet)$ be an elliptic complex on an oriented
manifold $M$. (i) Then, for any open neighborhood $U \sse M\times M$ of
the diagonal $M\sso M\times M$, there exists a properly supported
parametrix $G_i\in \Secs'_c(F_{i-1} \boxtimes \tilde{F}^*_i)$ with
support $\supp G_i \sse U$. (ii) Then also, the cohomologies of smooth
and distributional sections are isomorphic:
\begin{equation}
	H^i(\Secs_c'(F_\bullet), f_\bullet) \cong H^i(F_\bullet,f_\bullet)
	\quad \text{and} \quad
	H^i(\Secs'(F_\bullet), f_\bullet) \cong H^i_c(F_\bullet,f_\bullet) .
\end{equation}
\end{prop}
\begin{proof}
(i) The existence of a parametrix for any elliptic complex follows from
Corollary~2.1.11 and Theorem~2.1.12 of~\cite{tarkhanov}. The support of
an existing pa\-ra\-met\-rix can be restricted arbitrarily close to the
diagonal since $G_i^\chi$, defined by $G_i^\chi[\psi] = \chi G_i[\psi]$,
is a parametrix as long as $G_i$ is a parametrix and $\chi\in
C^\oo(M\times M)$ is properly supported with $\chi \equiv 1$ on a
neighborhood of the diagonal.

(ii) By the defining Equation~\eqref{eq:parametrix}, they are cochain
homotopic to the identity operator, with respect to the cochain homotopy
$G_i$. Further, being by hypothesis smooth and by (i) properly
supported, they define smoothing operators, $Q_i\colon \Do(F_i) \to
\Secs_c(F_i)$ and $Q_i\colon \D(F_i) \to \Secs(F_i)$, when extended to
distributions. It is then straightforward to see that the $Q_i$ and the
inclusions of smooth sections in distributional ones (well defined
because $M$ is oriented) constitute a homotopy equivalence between the
complexes of smooth $(\Secs(F_\bullet), f_\bullet)$ and distributional
$(\D(F_\bullet), f_\bullet)$ sections, and similarly for compact
supports. Thus, as desired, these complexes have isomorphic
cohomologies.
\end{proof}

\begin{prop}[Serre, Tarkhanov]\label{prp:serre-dual}
Given a differential complex $(F_\bullet, f_\bullet)$, that is not
necessarily elliptic, on an oriented manifold $M$ that is countable at
infinity (there exists an exhaustion by a countable sequence of compact
sets), let $(\tilde{F}^*_\bullet, f_\bullet)$ be its formal adjoint
complex.  The following are algebraic (the topologies may not agree)
isomorphisms of vector spaces
\begin{align}
	H^i(F_\bullet, f_\bullet)^*
		&\cong H^i(\Do(\tilde{F}^*_\bullet), f^*_\bullet) , &
	H^i(F_\bullet, f_\bullet)
		&\cong H^i(\Do(\tilde{F}^*_\bullet), f^*_\bullet)^* , \\
	H_c^i(F_\bullet, f_\bullet)^*
		&\cong H^i(\D(\tilde{F}^*_\bullet), f^*_\bullet) , &
	H_c^i(F_\bullet, f_\bullet)
		&\cong H^i(\D(\tilde{F}^*_\bullet), f^*_\bullet)^* ,
\end{align}
where the cohomology vector spaces are endowed with the natural
Hausdorff locally convex topology of a quotient of a subspace of the
corresponding space of sections (be it smooth or distributional) and the
topological duals are taken with the strong topology.
\end{prop}
\begin{proof}
The original result of Serre~\cite{serre} appeared in the context of the
Dolbeault differential complex in the theory of several complex
variables. A detailed discussion and proof of the result for general
differential complexes can be found in Sections~5.1.1 and 5.1.2
of~\cite{tarkhanov}. In particular, the desired conclusion can be found
in Remark~5.1.9 thereof. Further conditions under which some of the
duality isomorphisms are also continuous, and not merely algebraic, can
be found there as well.
\end{proof}

Combining the two preceding propositions, it is easy to see that for any
elliptic complex (subject to a countability condition on $M$) we have
the Poincar\'e-Serre duality relation $H^i(F_\bullet, f_\bullet) =
H^i_c(\tilde{F}^*_\bullet, f_\bullet^*)^*$.

\subsection{The Calabi cohomology and homology}\label{sec:calabi-cohom}
Below, we finally make use of the background information summarized in
Sections~\ref{sec:sh-lc}, \ref{sec:sh-res}, and~\ref{sec:dual} and its
consequences for the Calabi and its formal adjoint complexes,
$(C_\bullet,B_\bullet)$ and $(C_\bullet,B_\bullet^*)$, which were
introduced in~\ref{sec:calabi}. Namely, we make precise the
identification between their cohomologies and the sheaf cohomologies of
the Killing and Killing-Yano sheaves, $\K$ and $\KY$, introduced in
Section~\ref{sec:sh-lc}. The hope created by this identification is that
the difficult problem of solving systems of differential equations,
which appear in these complexes, can be replaced by the equivalent and
potentially easier problem of computing sheaf cohomologies. The latter
problem is potentially easier because of the many available methods of
computing sheaf cohomology. Some of which will be discussed in
Section~\ref{sec:killing}.

First, we introduce the basic definitions of Calabi cohomology and
homology. Let us denote the cohomology of the Calabi complex (\emph{Calabi
cohomology}) on a pseudo-Riemannian manifold $(M,g)$ of constant
curvature as
\begin{equation}
	HC^i(M,g) = H^i(C_\bullet, B_\bullet) = H^i(\Secs(C_\bullet),B_\bullet) .
\end{equation}
Let us also denote the cohomology of the formal adjoint Calabi complex
with compact supports (\emph{Calabi homology})
\begin{equation}
	HC_i(M,g) = H^i_c(C_\bullet,B^*_\bullet)
		= H^i(\Secs_c(C_\bullet),B^*_\bullet) .
\end{equation}
The naming convention will be justified later by the generalized
Poincar\'e duality relation in Corollary~\ref{cor:calabi-dual}. Similarly,
we define the cohomology of the Calabi complex with compact supports
(\emph{Calabi cohomology with compact supports}) as
\begin{equation}
	HC_c^i(M,g) = H^i_c(C_\bullet, B_\bullet)
		= H^i(\Secs_c(C_\bullet),B_\bullet)
\end{equation}
and the cohomology of the formal adjoint Calabi complex (\emph{locally
finite Calabi homology}) as
\begin{equation}
	HC^\lf_i(M,g) = H^i(C_\bullet,B^*_\bullet)
		= H^i(\Secs(C_\bullet),B^*_\bullet) .
\end{equation}

The following proposition is the main technical tool that we use to
establish all other results in this section.
\begin{prop}\label{prp:calabi-exact}
Consider a pseudo-Riemannian manifold $(M,g)$ of constant curvature and
dimension $n$. The corresponding Calabi complex $(C_\bullet, B_\bullet)$
is elliptic, formally exact and locally exact (except in degree $0$).
The same is true for its formal adjoint complex
$(C_\bullet,B_\bullet^*)$ (except in degree $n$).
\end{prop}
\begin{proof}
In principle, we would need quite a bit of machinery for a full proof.
Instead, we give a sketch of the main ideas and refer to the literature
for technical details. The Calabi complex is actually an instance of a
\emph{second Spencer sequence}
construction~\cite{quillen,goldschmidt-lin,spencer,pommaret} applied to
the Killing operator $B_1 = K$. This fact is the demonstrated in the
papers~\cite{gasqui-goldschmidt-fr,gasqui-goldschmidt,goldschmidt-calabi}.
These papers make use of the general construction and properties of the
differential complex constituting a second Spencer sequence demonstrated
in~\cite{quillen,goldschmidt-lin}. In fact, the resulting differential
complex gives a formally exact compatibility complex for the Killing
operator, which is also an elliptic complex. This holds since the
Killing operator $K$ is itself elliptic (has injective symbol, which
follows from the property of being of finite type,
cf.~Section~\ref{sec:tw-dr}) and formally integrable (contains all of
its integrability conditions) on a constant curvature background.

A more elementary argument for ellipticity can be made on representation
theoretic grounds (Appendix~\ref{app:yt-bkg}). The fibers of the tensor
bundles $C_iM$ carry irreducible representations of $\GL(n)$. Further,
as mentioned in Remark~\ref{rmk:elliptic}, the principal symbols of the
differential operators $B_i$ are all $\GL(n)$-equivariant maps $\sigma
B_i\colon \Y^{(k_i)}T^* \otimes C_{i-1} \to C_i$ or equivalently
$\sigma_p B_i \colon C_{i-1} \to C_i$, for $p\in T^*$. By Schur's lemma,
the symbol map $\sigma B_i$ is then an isomorphism when restricted to an
irreducible summand of the tensor product representation. The well-known
Littlewood-Richardson rules~\cite{fulton,lrr} for tensor products of
$\GL(n)$ representations then show that the $C_i$ irreps have been
chosen precisely such that the symbol sequence $\sigma_p B_i$ is exact
for $p\ne 0$. This representation-theoretic line of argument is a
special case of the construction of what are known as BGG
resolutions~\cite{bgg}.

Finally, local exactness (except in degree $0$) can be established by
checking, for the Killing operator, a sufficient condition known as the
\emph{$\delta$-estimate}~\cite[Sec.1.3.13]{tarkhanov}. Equivalently, we
can simply invoke Proposition~\ref{prp:tw-dr-exact}, since, being of
finite type, the Killing operator is equivalent to a flat covariant
operator (Section~\ref{sec:tw-dr}).

A more elementary proof of local exactness was given in the original
article by Calabi~\cite{calabi}. He relied on the well known local
exactness of the de~Rham complex and its relation to the simplified form
of the complex in the flat (zero curvature) case. The non-zero curvature
case was handled by embedding it in a flat space and then restricting
and extending the relevant sheaves with respect to this embedding.
Unfortunately, unlike the more sophisticated argument above, this
simpler argument is unlikely to generalize, when the Calabi complex is
replaced by a more general one.

To finish the proof, we note that the properties of formal exactness and
ellipticity are obviously preserved by taking formal adjoints, so that
they apply equally well to the formal adjoint Calabi complex
$(C_\bullet,B_\bullet^*)$. The formal adjoint complex then serves as the
formally exact compatibility complex for the Killing-Yano operator
$B_n^* = KY$, which is also regular and of finite type on constant
curvature backgrounds, as discussed in Section~\ref{sec:tw-dr}. Thus,
repeating the same arguments as above establishes local exactness
(except this time in degree $n$) for the adjoint complex as well.
\end{proof}

\begin{cor}\label{cor:calabi-equiv}
There is a formal homotopy equivalence between the Calabi complex
$(C_\bullet,B_\bullet)$ and the twisted de~Rham complex
$(\Lambda^\bullet M \otimes V_K,D_K)$ resolving the Killing sheaf,
$\H^0(C_\bullet,B_\bullet) = \K$. The same is true (up to a trivial
renumbering) of the formal adjoint complex and the twisted de~Rham
complex $(\Lambda^\bullet M \otimes V_{KY},D_{KY})$ resolving the
Killing-Yano sheaf, $\H^n(C_\bullet,B^*_\bullet) = \KY$.
\end{cor}
\begin{proof}
We already know that both the Calabi and twisted de~Rham complex
associated to the Killing operator are formally exact, locally exact
(Propositions~\ref{prp:calabi-exact} and~\ref{prp:tw-dr-exact}) and both
resolve the Killing sheaf, since the operators $K$ and $D_K$ are
equivalent (Section~\ref{sec:tw-dr}).  Thus, by Lemma~\ref{lem:f-exact},
there exists a formal homotopy equivalence (realized by differential
operators) between the two complexes. Noting that the exact same
argument (with trivial changes) applies to the formal adjoint Calabi
complex and the Killing-Yano sheaf concludes the proof.
\end{proof}

\begin{cor}\label{cor:calabi-dual}
Provided the manifold $M$ is countable at infinity (there is an
exhaustion by a countable sequence of compact sets) or is of finite type
(has a finite ``good cover''), we have the following generalized
Poincar\'e duality isomorphisms
\begin{align}
	HC^i(M,g) &\cong HC_i(M,g)^* , &
	HC^i_c(M,g)^* &\cong HC_i^\lf(M,g) , \\
	HC^i(M,g)^* &\cong HC_i(M,g) , &
	HC^i_c(M,g) &\cong HC_i^\lf(M,g)^* ,
\end{align}
where isomorphisms are taken in the algebraic sense and duality is meant
in the topological sense, as described in
Proposition~\ref{prp:serre-dual}.
\end{cor}
Note that in the case when all cohomology vector spaces are finite
dimensional, the distinction between algebraic or topological
isomorphisms and duals is irrelevant.
\begin{proof}
There are two ways to establish the desired duality isomorphisms, each
relying on slightly different conditions on $M$, reflected in the
hypotheses. We should note that both require an orientation on $M$.  The
existence off a non-degenerate metric on $M$ implies that it is
orientable. We then simply fix an orientation arbitrarily.

The Mayer-Vietoris argument (Proposition~\ref{prp:tw-dr-dual})
establishes the duality isomorphisms
\begin{align}
	H^i(\Lambda^\bullet M \otimes V_K,D_K)
		&\cong H^i_c(\Lambda^\bullet M \otimes V_K^*,D_K)^* \\
	\text{and} \quad 
	H^i(\Lambda^\bullet M \otimes V_{KY},D_{KY})
		&\cong H^i_c(\Lambda^\bullet M \otimes V_{KY}^*,D_{KY})^* .
\end{align}
Under the finite
type condition on $M$, an easy modification to the Mayer-Vietoris
argument (Propositions~5.3.1 and~5.3.2 of~\cite{bott-tu}) also shows
that each of these cohomology groups is finite dimensional, so the
reverse duality isomorphisms hold as well. Finally, the formal homotopy
equivalence of Corollary~\ref{cor:calabi-equiv} translates these
isomorphisms into the desired duality relations for Calabi cohomology
and homology.

The Poincar\'e-Serre argument, applies by virtue of the ellipticity of
the Calabi and its formal adjoint complexes
(Proposition~\ref{prp:calabi-exact}) and the hypothesis of countability
at infinity. Combining the results of Propositions~\ref{prp:ell-param}
and~\ref{prp:serre-dual}, easily establishes the desired duality
isomorphisms directly.
\end{proof}

\begin{cor}\label{cor:calabi-sheaf}
Assume the same hypotheses on $M$ as in Corollary~\ref{cor:calabi-dual}.
The Calabi cohomology and homology, assuming they are finite
dimensional, the following identities hold (with respect to algebraic
duals):
\begin{align}
	HC^i(M,g) &\cong H^i(\K) , & HC^i_c(M,g) &\cong H^{n-i}(\KY)^* , \\
	HC_i(M,g) &\cong H^i(\K)^* & HC^\lf_i(M,g) &\cong H^{n-i}(\KY) .
\end{align}
\end{cor}
Note that we do expect the relevant cohomology and homology spaces to be
finite dimensional in most applications. If the cohomology vector spaces
happen to be infinite dimensional, then the correct (topological and
algebraic) isomorphisms can be deduced from Corollary~\ref{cor:calabi-dual}
and Proposition~\ref{prp:serre-dual}.
\begin{proof}
By Proposition~\ref{prp:calabi-exact} and
Corollary~\ref{cor:calabi-equiv} we already know that the Calabi and its
formal adjoint complexes are locally exact differential complexes that
respectively resolve the Killing and Killing-Yano sheaves, $\K$ and
$\KY$. Then, Lemma~\ref{lem:f-exact} establishes the isomorphisms
$HC^i(M,g) \cong H^i(\K)$ and $HC^\lf_i(M,g) \cong H^{n-i}(\KY)$. Finally,
the duality isomorphisms of Corollary~\ref{cor:calabi-dual} establish
the rest of the desired identities. Note that we have added the finite
dimensionality hypothesis only to avoid explicitly specifying a topology
on the relevant cohomology vector spaces, so that the topological and
algebraic duals coincide.
\end{proof}

\section{The Killing sheaf and its cohomology}\label{sec:killing}
In this section we concentrate on possible effective ways of computing
the Killing sheaf cohomology (or rather the cohomology of any locally
constant sheaf) of a pseudo-Riemannian manifold $(M,g)$ of constant
curvature. For us, \emph{effective} is used somewhat loosely and we
take it to mean roughly to either consist of finitely many steps
involving only finite-dimensional linear algebra or to reduce to
calculation that has already been done in the literature. In particular,
any such method would be more effective than the brute force approach of
trying to solve the systems of differential equations appearing in the
Calabi complex. Since the interest in the cohomology of the Killing
sheaf may extend beyond the constant curvature context, we always
discuss the more general situation, specializing to the constant
curvature case when necessary.

There are two main possibilities, either the manifold $M$ is simply
connected or it is not. They are discussed respectively in
Sections~\ref{sec:simp-conn} and~\ref{sec:nonsimp-conn}. In the simply
connected case, the sheaf cohomology can be expressed completely in
terms of the de~Rham cohomology. The non-simply connected case is more
complicated, where several complementary but potentially overlapping
methods may be used.  None of them, unfortunately, gives a complete
solution.

Crucial to the discussion that follows (see
Appendix~\ref{app:bndl-deform} for relevant notation and concepts
related to $G$-bundles) is the notion of the monodromy representation of
the fundamental group $\pi = \pi_1(M)$ of a manifold with respect to a
flat connection $D$ on a vector bundle $V\to M$
(cf.~Section~\ref{sec:tw-dr}). Let us identify $\pi_1(M) = \pi_1(M,x)$
for some $x\in M$. The connection $D$ gives rise to a notion of parallel
transport on $V$.  Since the connection is flat, the parallel transport
along a curve connecting $x,y\in M$ depends only on the homotopy class
of the path with its endpoints fixed. Therefore, since parallel
transport acts linearly, parallel transport along loops based at $x\in
M$ induces a representation $\rho_V\colon \pi \to \GL(\bar{V})$, where
$\bar{V}\cong V_x$ is the typical fiber of $V\to M$, called the
\emph{monodromy representation}. Another common term is the
\emph{holonomy representation}. However, we reserve the term
\emph{holonomy} for the same concept associated specifically to the
$g$-compatible Levi-Civita connection on $M$. If $V\to M$ is a vector
$G$-bundle, then there necessarily is an associated representation of
the structure group on $\bar{V}$, $\sigma_V \colon G \to \GL(\bar{V})$.
When the connection $D$ preserves the $G$-bundle structure, parallel
transport and hence monodromy factors through the associated
representation. Hence $\rho_V = \sigma_V \circ \rho$, where $\rho\colon
\pi \to G$ is the monodromy representation of $\pi$ in the structure
group.

Recall also that for any manifold $M$ there exists a unique (up to
diffeomorphism) connected, simply connected \emph{universal cover}
$\tilde{M} \to M$, where the projection map is a surjective local
diffeomorphism. In fact, $\tilde{M}\to M$ is a $\pi$-principal bundle
over $M$. The principal bundle action of $\pi$ on $\tilde{M}$ by is
called action by \emph{deck transformations}. Note that $M\cong
\tilde{M}/\pi$. Deck transformations, being diffeomorphisms, commute
with the de~Rham differential. Hence the action by deck transformations
descends to de~Rham cohomology. We call it the \emph{deck
representation} $\Delta^i\colon \pi \to \GL(H^i(M))$. The projection to
$M$ pulls the bundle $V\to M$ back to $\tilde{V}\to \tilde{M}$ and the
connection $D$ to $\tilde{D}$. Since the universal cover is simply
connected, the pulled back bundle trivializes, $\tilde{V} \cong \bar{V}
\times \tilde{M}$. Therefore, we have the isomorphism
$H^i(\Lambda^\bullet \tilde{M} \otimes \tilde{V},D) \cong
H^i(\tilde{M})\otimes \bar{V}$. It is not hard to see that the two sides
are isomorphic not only as vector spaces but also as representations of
the fundamental group $\pi$, with the right side transforming as the
tensor product of the deck and monodromy representations
$\Delta^i\otimes \rho_V$.

Let us fix the assumptions that $(M,g)$ is connected and that its
Killing sheaf $\K_g$ is locally constant, then concretize the above
ideas to this case.  Recall from Section~\ref{sec:tw-dr} that $\K_g$ is
then resolved by the twisted de~Rham complex associated to the flat
vector bundle $(V_K,D_K)$. The typical fiber $\bar{V}_K$ of $V_K \to M$
consists of the germs of local Killing vector fields. Each local Killing
vector field extends to a global one, and hence to an infinitesimal
isometry, on the universal cover $(\tilde{M},\tilde{g})$. Thus, we can
identify $\bar{V}$ with the Lie algebra $\g$ of the Lie group $G =
\Isom(\tilde{M},\tilde{g})$ of isometries of $(\tilde{M},\tilde{g})$.

Infinitesimal isometries act on each other by the formula $\Lie_u v =
[u,v]$, which corresponds to the infinitesimal adjoint representation
$\ad\colon \g \to \End(\g)$. This representation integrates to the
adjoint representation $\Ad \colon G \to \GL(\g)$, which is how finite
isometries act on Killing vector fields. Also, it is clear by
construction that deck transformations act on $(\tilde{M},\tilde{g})$ by
isometries. Let us denote this representation of the fundamental group
$\pi = \pi_1(M)$ by isometries as $\rho\colon \pi \to G$. As described in
Section~\ref{sec:monodr-def}, this information is equivalent to
specifying a flat principal $G$-bundle $P\to M$ with monodromy
representation $\rho$ of $\pi$ in $G$. Further, it is clear that $V_K
\cong \g_P$ is the vector $G$-bundle over $M$ associated to $P$ with
respect to the adjoint action $\Ad$ of $G$ on $\g$ and that $D_K$ is the
connection associated to the flat principal connection on $P$. The
monodromy representation of $\pi$ on $\bar{V}_K$ is then the
\emph{composite adjoint monodromy representation} $\rho_V = \Ad_\rho =
\Ad\circ \rho$.

\subsection{Simply connected case}\label{sec:simp-conn}
The simplest case is when the manifold $M$ is simply connected, that is,
its fundamental group $\pi = \pi_1(M)$ is trivial. Let the locally
constant sheaf $\F$ have stalk $\bar{F}$ so that it defines a flat
vector bundle $(F,D)$, with $\bar{F}$ the typical fiber of $F\to M$
(Sections~\ref{sec:sh-lc} and~\ref{sec:tw-dr}). We know that the twisted
de~Rham differential complex $(\Lambda^\bullet M\otimes F,D)$ is an
acyclic resolution of $\F$. Hence their cohomologies agree. On the other
hand, since $M$ is simply connected, we can choose a global $D$-flat
basis frame for $F$ and identify the twisted de~Rham complex with $\rk F
= \dim \bar{F}$ copies of the standard de~Rham complex. This argument
proves
\begin{thm}\label{thm:simp-conn}
Let $(M,g)$ be a connected, simply connected pseudo-Riemannian manifold
with locally constant Killing sheaf $\K_g$, resolved by the twisted
de~Rham complex $(\Lambda^\bullet M \otimes V_K, D_K)$. Let $\g \cong
\bar{V}_K$ be the Lie algebra of isometries of $(M,g)$. Then the
following isomorphisms hold:
\begin{equation}
	H^i(\K_g)
	\cong H^i(\Lambda^\bullet M \otimes V_K, D_K)
	\cong H^i(M)\otimes \g .
\end{equation}
In particular $H^0(\K_g) \cong \g$ and $H^1(\K_g) = 0$.
\end{thm}

\subsection{Non-simply connected case}\label{sec:nonsimp-conn}
The non-simply connected case is of course more complicated and we can
offer only partial results, which we summarize in this paragraph. The
simplest sub-case is when the fundamental group $\pi = \pi_1(M)$ of the
pseudo-Riemannian manifold $(M,g)$ is finite (Section~\ref{sec:fin-fg}).
The Killing sheaf cohomology is then the $\pi$-invariant subspace of the
de~Rham cohomology of the universal covering space. If the fundamental
group is not necessarily finite, we still have the following general
result for the degree-$1$ cohomology of constant curvature spaces. We
can equate $\dim H^1(\K_g)$ to the dimension of the space of possible
infinitesimal deformations of the metric that preserve the constant
curvature condition as well as the value of the scalar curvature itself.
That observation was already made in the original work of
Calabi~\cite{calabi} and in fact prompted his interest in a resolution
of the Killing sheaf $\K_g$. This space of infinitesimal deformations
can also be computed as the degree-$1$ group cohomology of $\pi$ with
coefficients in a certain representation on the Lie algebra of
isometries of the universal cover of $M$ (Section~\ref{sec:inf-fg-1}).
Another result helps compute higher degree cohomology groups. The
Killing sheaf, being locally constant, defines a \emph{local system} or
a system of \emph{local coefficients} on $M$, a concept well known in
algebraic topology. A general result from the theory of local systems is
that the aforementioned group cohomology computes higher Killing sheaf
cohomology groups up to the degree of asphericity of $M$
(Section~\ref{sec:loc-coef}). Finally, there is a general method for
completely computing the Killing sheaf cohomology based on a
presentation of the manifold $M$ as a finite simplicial set
(Section~\ref{sec:simp-set}).

\subsubsection{Finite fundamental group}\label{sec:fin-fg}
The basic idea here is to take advantage of the complete decomposability
of representations of a finite group and then apply Schur's lemma. As
will be clear from the proof, it is the complete decomposability that is
important not the finiteness of $\pi$. So the same result actually holds
under suitably weaker hypotheses.
\begin{thm}
Let $(M,g)$ be a connected pseudo-Riemannian manifold with fundamental
group $\pi = \pi_1(M)$ and Killing sheaf $\K_g$, resolved by the twisted
de~Rham complex $(\Lambda^\bullet M \otimes V_K, D_K)$. Let $\g \cong
\bar{V}_K$ be the Lie algebra of isometries of the universal cover
$(\tilde{M},\tilde{g})$. If $\pi$ is finite, we have the following
isomorphisms:
\begin{equation}
	H^i(\K_g) \cong (H^i(\tilde{M})\otimes \g)^\pi ,
\end{equation}
where the superscript $\pi$ denotes the $\pi$-invariant subspace with
respect to the representation $\Delta^i\otimes\Ad_\rho$, the tensor
product of the deck and composite adjoint monodromy representations. In
particular $H^0(\K_g) \cong \g^\pi$.
\end{thm}
\begin{proof}
Consider the spaces of sections $\Omega_i = \Secs(\Lambda^i
\tilde{M}\otimes \tilde{V}_K)$, where $\tilde{V}_K \to \tilde{M}$ is the
pullback of $V_K\to M$ along the universal covering projection
$\tilde{M} \to M$. Let $\tilde{D}_K$ denote the pullback of the $D_K$.
As we have already discussed at the top of Section~\ref{sec:killing},
this pulled back bundle is trivial, $\tilde{V}_K \cong \g \times M$.
Moreover, by simple connectedness of $\tilde{M}$ and
Theorem~\ref{thm:simp-conn}, we have the isomorphism
$H^i = H^i(\Omega_i,\tilde{D}_K) \cong H^i(\tilde{M})\otimes \g$. 

As also discussed at the top of Section~\ref{sec:killing}, the spaces
$\Omega_i$ carry representations of the fundamental group $\pi$, which
also descends to the cohomologies $H^i$. Since $\pi$ is finite, it is
well known that any representation thereof is completely
decomposable~\cite{maschke}, that is, any subrepresentation has a direct
sum complement subrepresentation. So, the subspace $\Omega_i^\pi \sso
\Omega_i$ invariant under the action $\pi$ (every element of $\pi$ acts
as the identity operator) has a direct sum complement
$\Omega_i^{\hat\pi}$, so that $\Omega_i \cong \Omega_i^\pi \oplus
\Omega_i^{\hat\pi}$. This direct sum induces the short exact sequence
\begin{equation}\label{eq:pi-short-exact}
\begin{tikzcd}
	0 \ar{r} &
	\Omega_i^\pi \ar{r} &
	\Omega_i \ar{r} &
	\Omega_i^{\hat\pi} \ar{r} &
	0 .
\end{tikzcd}
\end{equation}
It is straightforward to note that, by construction of the universal
cover $\tilde{M}\to M$, the $\pi$-invariant subcomplex
$(\Omega_i^\pi,\tilde{D}_K)$ on $\tilde{M}$ is in fact cochain
isomorphic to the complex $(\Secs(\Lambda^\bullet M \otimes V_K), D_K)$
on $M$. Therefore the desired cohomology groups are $H^i(\Lambda^\bullet
M \otimes V_K, D_K) \cong H^i(\Omega_\bullet^\pi, \tilde{D}_K)$.

The complement $\Omega_i^{\hat\pi}$ naturally does not contain any
non-zero vectors invariant under the action of $\pi$. In representation
theoretic terminology, these two complementary subspaces are
\emph{disjoint}. By Schur's lemma~\cite{schur}, the only equivariant map
(\emph{intertwiner}) between any two disjoint representations is the
zero map. Note that the differentials $\tilde{D}_K$ and the maps in the
short exact sequence~\eqref{eq:pi-short-exact} are in fact
$\pi$-equivariant. By the general machinery of homological algebra
(Appendix~\ref{app:homalg}) the short exact
sequence~\eqref{eq:pi-short-exact} induces the long exact sequence
\begin{equation}\label{eq:pi-long-exact}
\begin{tikzcd}
	0 \arrow{r} &
	H^0_\pi \arrow{r} &
	H^0 \arrow{r}\arrow[draw=none]{d}[name=Z,shape=coordinate]{}&
	H^0_{\hat\pi}
		\arrow[rounded corners,
			to path={ -- ([xshift=2ex]\tikztostart.east)
			|- (Z) [near end]\tikztonodes
			-| ([xshift=-2ex]\tikztotarget.west)
			-- (\tikztotarget)}]{dll} \\
	&
	H^1_\pi \arrow{r} &
	H^1 \arrow{r} &
	H^1_{\hat\pi} \arrow{r} & \cdots
\end{tikzcd}
\end{equation}
where $H^i_\pi = H^i(\Omega_\bullet^\pi, \tilde{D}_K)$,
$H^i_{\hat\pi} = H^i(\Omega_\bullet^{\hat\pi}, \tilde{D}_K)$ and all
the maps are also $\pi$-equivariant. It is clear that the
representations carried by $H^i_\pi$ and $H^i_{\hat\pi}$ are also
disjoint. Therefore, the maps connecting the rows of
diagram~\eqref{eq:pi-long-exact} are all zero. In other words, each of
the rows becomes a short exact sequence on its own. Invoking again
complete decomposability of representations of $\pi$, we can write $H^i
\cong H^i_\pi \oplus H^i_{\hat\pi}$ and hence identify $H^i_\pi \cong
(H^i)^\pi$ with the subspace of $H^i$ on which $\pi$ acts trivially.

Collecting the above arguments together, while recalling the sheaf
cohomology identity $H^i(\K_g) \cong H^i(\Lambda^\bullet M \otimes V_K,
D_K)$, we obtain the isomorphism $H^i(\K_g) \cong (H^i(\tilde{M})\otimes
\g)^\pi$. Noting the special cases $H^0(\tilde{M}) = \R$ and
$H^1(\tilde{M}) = 0$, as in Theorem~\ref{thm:simp-conn}, concludes the
proof.
\end{proof}

\subsubsection{Degree-$1$ cohomology}\label{sec:inf-fg-1}
Consider a $1$-parameter family of $n$-dimensional pseudo-Riemannian
manifolds $(M,g(t))$ where each $g(t)$, for $t$ in some neighborhood of
zero, has constant curvature, with scalar curvature independent of $t$:
Riemann tensor equal to $\frac{k}{n(n-1)} g(t)\odot g(t)$. Let $g(0) =
g$ and $\dot{g}(0) = h$. Then the linearization of the identity $R[g(t)]
- \frac{k}{n(n-1)} g(t)\odot g(t) = 0$ at $t=0$ will give
(cf.~Section~\ref{sec:calabi-ops})
\begin{equation}
	\dot{R}[h] - k\frac{2}{n(n-1)} (g\odot h) = -\frac{1}{2} C_2[h] = 0 .
\end{equation}
In other words, $h$ is a Calabi $1$-cocycle. It is possible that not
every Calabi $1$-cocycle gives rise to an actual $1$-parameter family of
deformations, since there may be algebraic obstructions%
	\footnote{The study of these obstructions follows the general ideas
	outlined by Kodaira and
	Spencer~\cite{spencer-deform1,spencer-deform2}. See also the related
	phenomenon of linearization instabilities~\cite{kh-linstab}.} %
to solving for higher order terms in the expansion parameter $t$.
However, at the infinitesimal level, there are no other conditions and
we can identify infinitesimal deformations with Calabi $1$-cocycles. If
the deformation family $g(t)$ is trivial, induced by a $1$-parameter
family of diffeomorphisms of the manifold $M$, then it is well known
that $h = K[v]$ for some $1$-form $v$ (vector field generating the
diffeomorphism family, with index lowered by the metric $g$), in other
words a Calabi $1$-coboundary. It is easy to see that Calabi
$1$-coboundaries can be identified with infinitesimal trivial
deformations. Therefore, the Calabi cohomology vector space $HC^1(M,g)$,
and hence the Killing cohomology vector space $H^1(\K_g)$ isomorphic to
it, is in bijective correspondence with the space of infinitesimal
deformations of the metric $g$ within the space of constant curvature
metrics of scalar curvature $k$, modulo infinitesimal diffeomorphisms.

There is another way to look at this infinitesimal deformation space. It
is well known that the only geodesically complete, simply connected,
constant curvature spaces are the pseudo-Euclidean ($k=0$),
pseudo-spherical ($k>0$) and pseudo-hyperbolic ($k<0$)
spaces~\cite[Sec.2.4]{wolf-cc}. In Riemannian signature, these are
respectively the ordinary Euclidean, spherical and hyperbolic spaces. In
Lorentzian signature, these are respectively the Minkowski, de~Sitter
and anti-de~Sitter spaces. Thus, the elements of a family $(M,g(t))$ of
geodesically complete, constant curvature, pseudo-Riemannian manifolds
of fixed scalar curvature $k$ all have isometric universal covers
$(\tilde{M},\tilde{g})$. Moreover, since the action of the fundamental
group $\pi = \pi_1(M)$ on its universal cover via deck transformations
is by isometries, there is an injective group homomorphism $\pi \to G =
\Isom(\tilde{M},\tilde{g})$, so that we have a subgroup $\rho(\pi) \sse
G$ that acts on $\tilde{M}$ properly and
discontinuously~\cite[Sec.1.8]{wolf-cc}. Conversely, for any subgroup of
$\pi'\sse G$ that acts on $\tilde{M}$ properly and discontinuously the
quotient $(M',g') = (\tilde{M},\tilde{g})/\pi'$ will be a manifold of
the same constant curvature, but with fundamental group $\pi' =
\pi_1(M')$. So, we have already noticed that all $(M,g)$ with constant
curvature arise in this way. Of course, any two subgroups $\pi',\pi''
\sse G$ that are conjugate, $\pi'' = a \pi' a^{-1}$ for some $a\in G$,
give rise to isometric quotients. In fact, we have just argued that the
infinitesimal deformations of the representation $\rho\colon \pi \to G$,
up to conjugation by $G$, are in bijection with infinitesimal constant
curvature deformations of the metric $(M,g)$. It is well known that the
deformations of the representation $\rho$ are in bijection with certain
degree-$1$ \emph{group cohomology} of the fundamental group $\pi$. On
the other hand, we have already seen that deformations of the constant
curvature spaces are parametrized by the Killing sheaf cohomology
$H^1(\K_g)$. Thus, computing the group cohomology may be an effective
way of computing the Killing sheaf cohomology, at least in degree-$1$.
The details of the definition of the representation $\rho$ are described
at the top of Section~\ref{sec:killing} and are also subsumed by the
more general discussion below.

This connection between the degree-$1$ Killing sheaf cohomology,
deformations of the geometry and group cohomology of the fundamental
group $\pi$ extends far beyond the case of manifolds of constant
curvature. We base what follows on the remark at the top of
Section~\ref{sec:killing} and the contents of
Appendix~\ref{app:bndl-deform}. If $(\tilde{M},\tilde{g})$ is the
universal cover of $(M,g)$ and $G = \Isom(\tilde{M},\tilde{g})$ with Lie
algebra $\g$, then there is a naturally defined flat principal
$G$-bundle $P\to M$. Then, the infinitesimal deformations of this flat
principal $G$-bundle are in bijections with $H^1(\K_g)$, the degree-$1$
Killing sheaf cohomology group. That is because the flat vector bundle
$(V_K,D_K)$, whose twisted de~Rham complex resolves the Killing sheaf,
is isomorphic to the associated bundle $\g_P \to M$ with connection $D$
induced by the flat principal connection on $P$. Recall that the fibers
of $\g_P$ transform under the adjoint representation $\Ad\colon G \to
\GL(\g)$ and that parallel transport with respect to the flat connection
on $P$ defines a representation $\rho\colon \pi \to G$ of the
fundamental group $\pi = \pi_1(M)$. Their composition $\Ad_\rho = \Ad
\circ \rho$, as already mentioned at the top of
Section~\ref{sec:killing}, is known as the \emph{composite adjoint
monodromy representation}. In the case of spaces of constant curvature,
the infinitesimal deformations of the flat principal bundle $P\to M$ are
the same thing as the infinitesimal deformations of the given constant
curvature metric, fixing the value of the curvature.
\begin{thm}
Given the notations and hypotheses of the above paragraph, the following
isomorphisms between the Killing sheaf cohomology and group cohomology
of $\pi$ with coefficients in $\Ad_\rho$ hold:
\begin{equation}
	H^0(\K_g) \cong H^0(\pi,\Ad_\rho) \cong \g^\pi , \quad
	H^1(\K_g) \cong H^1(\pi,\Ad_\rho) .
\end{equation}
\end{thm}
This result is a direct consequence of Proposition~\ref{prp:cohom-equiv}
of Appendix~\ref{app:bndl-deform}. Unfortunately, we cannot use the same
methods to establish isomorphisms between the group and sheaf
cohomologies in higher degrees. See, however,
Section~\ref{sec:loc-coef}. The connection between group cohomology of
$\pi$ and deformations of a flat principal bundle is well known,
cf.~\cite{goldman}. The connection between, specifically, the cohomology
of the Killing sheaf, infinitesimal deformations of the corresponding
principal bundle, and group cohomology seems to be less well known, but
is mentioned explicitly in~\cite[App.A.2]{acm}.

\subsubsection{Cohomology with local coefficients}\label{sec:loc-coef}
We have just noted, in Section~\ref{sec:inf-fg-1}, a geometric relation
between degree-$1$ locally constant sheaf cohomology and cohomology of
the fundamental group. A more general connection between the cohomology
of a locally constant sheaf, or equivalently cohomology with coefficients
in a local system~\cite[Ch.VI]{whitehead}, and group cohomology of the
fundamental group has also been noticed in pure algebraic topology. In
fact, that is how the notion of group cohomology first arose.

The original goal was to calculate the cohomology of a space (with or
without coefficients in a non-trivial local system) in terms of data
specifying its homotopy type.  Following some early work by Hurewicz,
Hopf and Eilenberg, Eilenberg and Maclane~\cite{eilenberg-maclane}
introduced what are now known as $K(1,\pi)$ spaces (topological spaces
with $\pi_1 = \pi$ and $\pi_i = 0$ for all $i>0$) and computed all of
their cohomology groups by introducing an algebraic construction based
on the knowledge of the group $\pi$. We now call this construction
\emph{group cohomology}~\cite{group-cohom}. They further showed that the
same construction works also for any topological space $M$, not just a
$K(1,\pi)$, for the cohomologies in degree-$i$, with $0 < i \le p$, as
long as the space $M$ is $p$-aspherical, $\pi_i = 0$ for $0 < i \le p$.
This result, applied to the Killing sheaf gives the following
\begin{prop}
Let $(M,g)$ be a connected pseudo-Riemannian manifold with locally
constant Killing sheaf $\K_g$ and universal cover
$(\tilde{M},\tilde{g})$. Denote $G = \Isom(\tilde{M},\tilde{g})$ the
group of isometries of the universal cover and let $\g$ be its Lie
algebra. The fundamental group $\pi = \pi_1(M)$ acts on $\g$ via the
composite adjoint monodromy representation $\Ad_\rho \colon \pi \to
\GL(\g)$. If the manifold $M$ is $p$-aspherical, meaning $\pi_i(M) = 0$
for $0 < i \le p$, then we have the following isomorphisms:
\begin{equation}
	H^i(\K_g) = H^i(\pi,\Ad_\rho) \quad
	\text{for $0 \le i \le p$.}
\end{equation}
\end{prop}
For higher degree cohomology there are other contributions to the
homology groups. There is still a homomorphism $H^i(\pi,\Ad_\rho) \to
H^i(\K_g)$, but it need no longer be an
isomorphism~\cite[Sec.1.4.2]{morita}.

Later, Postnikov~\cite{postnikov-dok,postnikov-ru} proposed a full
solution for algebraically determining all the cohomology groups of a
space based on its homotopy type. Postnikov's method encodes the full
homotopy type of a space in terms of their homotopy groups and certain
additional algebraic data known as a \emph{Postnikov system} or
\emph{tower}. This construction is currently more commonly known in its
topological form~\cite[Ch.IX]{whitehead}. For $p$-aspherical spaces and
cohomology in degree $i$, with $0< i \le p$, Postnikov's construction
coincides with group cohomology. In general, the two constructions do
differ in degrees higher than the degree of asphericity.

Unfortunately, both Postnikov's encoding of the homotopy type and his
algebraic reconstruction of the cohomology are rather complicated, do
not appear to have gained much popularity. They seem to be fully
described only in his original monograph~\cite{postnikov-ru} or its
translation~\cite{postnikov-en}, both being rather obscure references.
At the moment, it is not clear to us what is the modern state of the art
in terms of reconstructing the cohomology of a space with coefficients
in a local system in terms of the space's homotopy type.

\subsubsection{Simplicial set cohomology}\label{sec:simp-set}
The last mathematical tool, which we will discuss, that can aid in the
computation of the cohomologies of a locally constant sheaf is
simplicial cohomology with local coefficients. The idea is to substitute
the underlying manifold $M$ with a combinatorial structure like a
\emph{simplicial complex} or a \emph{simplicial set}. Then, provided the
combinatorial model is finite, the corresponding cohomology theory
reduces to the computation of the cohomology of a finite dimensional
cochain complex, and thus to finite dimensional linear algebra. We defer
to the discussion in \cite[Sec.I.4.7--10]{gelfand-manin} for technical
details.

A disadvantage of this method is that finite combinatorial models only
cover the case of compact manifolds. Non-compact manifolds require
either an infinite combinatorial model or a non-trivial extension of the
formalism. Another inconvenience, besides the need for an explicit
decomposition of $M$ into simplices, is the need to define a discrete
analog of parallel transport on the simplicial model to reproduce the
composite adjoint monodromy representation $\Ad_\rho$. That is usually
done by associating a copy of $\g$ to each vertex of the simplicial
model for $M$ and explicitly assigning a coherent set of linear
isomorphisms between these copies to the edges connecting them, such
that the composition of the isomorphisms of the edges along a closed
loop is equal to the $\Ad_\rho$ action of the corresponding element of
$\pi$. These choices may be simplified if all vertices could be
collapsed into a single one, which is allowed for simplicial sets. Such
a construction is always possible when $M$ is compact and results in a
so-called \emph{reduced simplicial set}~\cite{rsset}.

\section{Application to linearized gravity}\label{sec:appl}
Recently, the symplectic and Poisson structure of linear classical field
theories has been studied by the author within a very general
framework~\cite{kh-peierls,kh-big} (see also~\cite{forger-romero,bfr,hs}
for related work), which admits in particular any linear field theory
whose gauge fixed equations of motion can be formulated as a hyperbolic
PDE system with possible constraints and residual gauge freedom. Certain
sufficient geometric conditions need to be satisfied for a field theory
to fit into that framework. The framework can then precisely
characterize the degeneracies of the presymplectic and Poisson tensors
on the solution space of the theory. These sufficient conditions require
the gauge generator and the constraint operator to fit into differential
complexes and the degeneracies of the presymplectic and Poisson tensors
are then characterized using the cohomology of these complexes. Once
known, these presymplectic and Poisson degeneracies are known to be of
importance in classifying the charges, locality, superselection sectors
and quantization of the corresponding classical theory.

The well known examples of Maxwell electromagnetism and Maxwell
$p$-forms~\cite{sdh,bds,benini} fit into this
framework~\cite[Sec.4.2]{kh-peierls}, invoking the well known de~Rham
complex. Linearized gravity on a constant curvature Lorentzian manifold
also fits into this framework, with the role of the de~Rham complex
replaced by the Calabi complex or, as appropriate, the formal adjoint
Calabi complex. For linearized gravity on an arbitrary background, we
would need to make use of different differential complexes. The Calabi
complex would be replaced by complexes defined by the property that they
(at least formally) resolve the sheaf of Killing vectors on the given
background (cf.~Section~\ref{sec:sh-res}). The corresponding formal
adjoint complexes would play a role as well. This connection to the
Killing sheaf, even without explicitly knowing the needed differential
complexes, shows that the Killing sheaf cohomology plays a similar role
both in the constant curvature context and more generally.  Thus the
ability to compute the Killing sheaf cohomology in as many circumstances
as possible (as discussed in Section~\ref{sec:killing}) should take us a
large part of the way towards understanding the presymplectic and
Poisson degeneracies of linearized gravity on general backgrounds.
Unfortunately, about half the desired information would still be
missing, since it is not clear which sheaf cohomology theory would
control the cohomology of the formal adjoint differential complex. In
the case of constant curvature, we were able to identify it as the
cohomology of the sheaf of rank-$(n-2)$ Killing-Yano tensors, which is
resolved by the formal adjoint Calabi complex
(Section~\ref{sec:calabi-adj}). It is currently not clear how to identify
its analog in the case of a general background, without knowing the full
differential complex that (formally) resolves the sheaf of Killing
vectors.

Now, specialized to the case of linearized gravity on a constant
curvature background $(M,g)$, the analysis of~\cite{kh-peierls,kh-big}
concludes that the presymplectic and Poisson tensors are actually
non-degenerate (with spacelike compact support for solutions and compact
support for smeared observables) if and only if the following two
conditions are satisfied: \emph{(global recognizability)} a certain
bilinear pairing between degree-$1$ Calabi cohomology with spacelike
compact supports and degree-$1$ Calabi homology, \emph{(global
parametrizability)} a certain bilinear pairing between on-shell
degree-$1$ Calabi cohomology with spacelike compact supports and
on-shell degree-$1$ timelike finite Calabi homology.

The descriptions of off-shell or on-shell Calabi cohomologies with
\emph{spacelike compact} supports, $HC^i_\sc$ or $HC^i_{\square,\sc}$,
and of off-shell or on-shell \emph{timelike finite} Calabi homology,
$HC_i^\tf$ or $HC_i^{\square,\tf}$ go beyond the scope of the current
work. However, they are defined and studied in detail in~\cite{kh-cohom}
(similar ideas appear also in~\cite{benini}). In fact, the results
of~\cite{kh-cohom} show how to express these non-standard cohomologies
in terms of the standard ones with unrestricted or compact supports, and
similarly for homology. Recall also (Section~\ref{sec:calabi-cohom})
that the latter are isomorphic to appropriate cohomologies (or their
linear duals) of the Killing or Killing-Yano sheaves, $\K_g$ or $\KY_g$.
Using all of these results we are able to translate the non-degeneracy
requirements as follows: \emph{(global recognizability)} a certain
bilinear pairing between
\begin{align}
	HC^1_\sc(M,g) %\cong HC^2_c(M,g)
		&\cong H^{n-2}(M,\KY_g)^* \\
	\text{and} \quad 
	HC_1(M,g)
		&\cong H^1(M,\K_g)^*
\end{align}
is non-degenerate, \emph{(global parametrizability)} a certain bilinear
pairing between
\begin{align}
\notag
	HC^1_{\square,\sc}(M,g)
		%&\cong HC_c^1(M,g)\oplus HC_c^2(M,g) \\
		&\cong H^{n-1}(M,\KY_g)^* \oplus H^{n-2}(M,\KY_g)^* \\
\notag
	\text{and} \quad
	HC^{\square,\tf}_1(M,g)
		%&\cong HC_1(M,g)\oplus HC_0(M,g) \\
		&\cong H^1(M,\K_g)^* \oplus H^0(M,\K_g)^*
\end{align}
is non-degenerate. Notice that we have succeeded in expressing the vector
spaces on which these pairings are defined purely in terms of Killing
and Killing-Yano sheaf cohomologies.

Checking non-degeneracy of course requires an explicit expression for
the required bilinear pairings. Such expressions can be obtained from
the general framework of~\cite{kh-peierls,kh-big}. However, there are
two cases were we do not need such detailed information, and these are
the ones we shall content ourselves here. For instance, if all the
relevant cohomology vector spaces are trivial, then the only possible,
trivial bilinear pairing is automatically non-degenerate. On the other
hand, if the paired vector spaces have different dimensions, then every
possible pairing between them must be degenerate.

We conclude this section by listing several well known Lorentzian
backgrounds for which the methods of Section~\ref{sec:killing} allow us
to determine all or a few of the cohomologies of the Killing sheaf. For
the reasons discussed above, we make note of the Killing-Yano sheaf
cohomologies only for constant curvature backgrounds.

\begin{table}
\begin{center}%
\begin{tabular}{lllll}
	spacetime & topology & $b^i = \dim H^i(\K)$ & $c^i = \dim H^i(\KY)$ \\
	\hline\rule[0.5ex]{0pt}{2.5ex}%
	Minkowski & $\R^n$ & $b^0 = \frac{n(n+1)}{2}$ & $c^0 = \frac{n(n+1)}{2}$ \\
	open FLRW & $\R^n$ & $b^0 = \frac{(n-1)n}{2}$ \\
	de~Sitter & $\R\times S^{n-1}$ & $b^0 = b^{n-1} = \frac{n(n+1)}{2}$
		& $c^0 = c^{n-1} = \frac{n(n+1)}{2}$ \\
	closed FLRW & $\R\times S^{n-1}$ & $b^0 = b^{n-1} = \frac{(n-1)n}{2}$ \\
	Schwarzschild & $\R^2\times S^2$ & $b^0 = b^2 = 4$ \\
	Tangherlini & $\R^2\times S^{n-2}$ & $b^0 = b^{n-2} =
		\frac{(n-2)(n-1)}{2}$ \\
	Kerr & $\R^2\times S^2$ & $b^0 = b^2 = 2$ \\
	Meyers-Perry & $\R^2\times S^{n-2}$ & $b^0 = b^{n-2} = 1+N$
\end{tabular}%
\caption{A list of some well known, simply connected solutions of
	(cosmological) vacuum Einstein equations, together with their topology
	and non-vanishing dimensions of Killing or Killing-Yano sheaf
	cohomologies. Note that $b^0$ always counts the number of independent
	global Killing vectors, and similarly for $c^0$. The Tangherlini
	solutions generalize the Schwarzschild one to higher dimensions and
	the Meyers-Perry solutions do the same for Kerr~\cite{emparan-reall}.
	For the latter, $N$ counts the number of rotational symmetries, which
	varies depending on the variant of the solution. We only consider the
	exterior regions for black hole solutions.\label{tab:simp-conn}}
\end{center}
\end{table}

The easiest case is that of simply connected spacetimes. Then, the
Killing sheaf cohomology is just the de~Rham cohomology tensored with
the Lie algebra of global isometries (Section~\ref{sec:simp-conn}), with
an analogous result for any other locally constant sheaf. Many of the
well known exact solutions are in fact defined on simply connected
underlying manifolds, including Minkowski space, black hole solutions
and cosmological solutions. A few explicit examples are listed in
Table~\ref{tab:simp-conn}. Note that only the Minkowski and de~Sitter
spaces are of constant curvature, so that the Calabi complex could be
defined on them. For these backgrounds, it makes sense to also compute
the Killing-Yano sheaf cohomologies $H^i(\KY)$. However, since we know
that the number of linearly independent rank-$(n-2)$ Killing-Yano
tensors on these spaces is the same as the number of linearly
independent Killing vectors (Section~\ref{sec:calabi-adj}), the
cohomology vector spaces are isomorphic, $H^i(\KY) \cong H^i(\K)$.

In the non-simply connected case, we can rely on the results of
Sections~\ref{sec:fin-fg}, \ref{sec:inf-fg-1} and~\ref{sec:loc-coef},
according to which we can equate the Killing sheaf cohomologies with the
group cohomology of the fundamental group with coefficients in the
composite adjoint monodromy representation, at least up to the degree of
asphericity of the underlying spacetime manifold. Unfortunately, there
does not seem to exist a comprehensive list of exact solutions of
Einstein's equations indexed by spacetime topology. So it takes some
effort to find explicit examples of exact solutions on non-simply
connected spacetimes. A rich source of examples comes from quotients of
simply connected spacetimes (such as those mentioned in the preceding
paragraph) by a discrete, freely acting subgroup $\pi$ of the isometry
group. The quotient is a manifold because the action of $\pi$ is free
and the metric descends to the quotient because the action of $\pi$ on
the original spacetime is by isometries. The group $\pi$ then becomes
the fundamental group of the quotient.

\begin{table}
\begin{center}%
\begin{tabular}{llllll}
	$M$ & $\pi_1(M)$ & Bianchi sym. & additional sym. & $b^0$ & $b^1$ \\
	\hline\rule[0.5ex]{0pt}{2.5ex}%
	$\R\times T^3$ & $\mathbb{Z}^3$ & $\R^3$ & $1$
		& $3$ & $6$ \\
	$\R\times T^3$ & $\mathbb{Z}^3$ & $\mathrm{VII}(0)$ & $1$
		& 2 & $4$ \\
	$\R\times T^3$ & $\mathbb{Z}^3$ & $\R^3$ & $SO(2)$
		& $3$ & $6$ \\
	$\R\times T^3$ & $\mathbb{Z}^3$ & $\R^3$ & $SO(3)$
		& $3$ & $5$
\end{tabular}
\caption{Known values of $b^i = \dim H^i(\K_g)$ for a generic spatially
	homogeneous spacetime $(M,g)$ with given topology and symmetry
	properties. See text for more details.\label{tab:fg-z3}}
\end{center}%
\end{table}

An nearly exhaustive study of possible quotients of $4$-dimensional
cosmological solutions (meaning spatially homogeneous ones) has been
carried out in~\cite{kth,kth2,kodama}. A complete presentation of the
results is rather complicated and is relegated to the original
references. A particular cosmological solution $(M,g)$ is identified by
(a) the topology of the spacetime $M$, (b) the topology and isometry
group of the universal cover $(\tilde{M},\tilde{g})$, (c) a number of
continuous \emph{metric parameters} specifying $\tilde{g}$, and (d) a
number of continuous \emph{moduli} (or \emph{Teichm\"uller parameters})
specifying the quotient class. There may also be additional discrete
parameters, but we ignore them here, since they do not affect the number
of continuous parameters. According to the discussion of
Section~\ref{sec:inf-fg-1}, the number of moduli (denoted by
$N_\mathrm{m}$ in~\cite{kodama}) is in fact equal%
	\footnote{The space of moduli may not always be a smooth manifold, but
	may have algebraic singularities. Still, the number of moduli is the
	dimension of the generic subset of the moduli space, which is a smooth
	manifold. This dimension is also equal to $b^1 = \dim H^1(\K)$. At singular
	points of the moduli space, $b^1$ may actually exceed the number of
	moduli, so at those points a more careful analysis is needed.} %
to $b^1 = \dim H^1(\K_g)$. We shall not specify any metric parameters,
since, as long as they take generic values they do not affect the number
of moduli. For simplicity, we only consider the examples with toroidal
spatial topology $M = \R \times T^3$, where $T^3 = S^1 \times S^1 \times
S^1$. Hence, the fundamental group is $\pi_1(M) = \mathbb{Z}^3$ and the
universal cover is $\mathbb{R}^4$. The (identity component) of the
isometry group of $(\tilde{M},\tilde{g})$ is then a semidirect product
of a $3$-dimensional transitive Bianchi group and an additional
connected Lie group. Let us concentrate on the cases of either Bianchi
type $I \cong \R^3$ or $\mathrm{VII}(0)$. Under these conditions, we can
read off all the remaining possibilities and information form Table~IV
of~\cite{kodama}. They are summarized in Table~\ref{tab:fg-z3}. Note
that $b^0 = \dim H^0(\K_g)$ counts the number of independent global
Killing vectors on $(M,g)$. The number of independent Killing vectors on
$(\tilde{M},\tilde{g})$ counts the dimension of the Bianchi group
(always $3$) and the dimension of the additional symmetry group. The
number of independent Killing vectors not broken by compactification to
$T^3$ can be deduced from the explicit presentation of the isometry
groups $\Isom(\tilde{M},\tilde{g})$ and the discrete subgroups effecting
the compactification, which for the examples given in
Table~\ref{tab:fg-z3} in~\cite[Sec.3]{kodama}. Many more examples cam be
read off from Tables~IV, VII and Section~5.3 of~\cite{kodama}.

It appears difficult to locate literature on explicit calculations that
are equivalent to computing higher Killing sheaf cohomologies for other
non-simply connected spacetimes.

\section{Discussion and generalizations}\label{sec:discuss}
We have reviewed in detail the algebraic, geometric and analytical
properties of the Calabi differential complex~\cite{calabi}.

In Section~\ref{sec:calabi} we have defined the nodes of the complex in
terms of Young symmetrized tensor bundles and given explicit formulas
for the differential operators between them, verifying through explicit
calculations that they in fact constitute a complex
(Appendix~\ref{app:young}). Such explicit formulas are otherwise
difficult to extract from the existing literature, especially in terms
of tensor variables, as opposed to moving coframe variables used in
Calabi's original work. Further, our formulas work for pseudo-Riemannian
backgrounds of any signature, generalizing from the standard purely
Riemannian context. We have also identified a differential operator
cochain homotopy (Equations~\eqref{eq:calabi-diag},
\eqref{eq:calabi-homot-start}--\eqref{eq:calabi-homot-end}) that
generates a cochain map from the complex to itself with a Laplacian-like
principal symbol. This cochain homotopy map may be new. However, its
lower order terms coincide with well known geometric operators known
from the theory of linearized gravity (General Relativity). Another
interesting and likely novel observation involved the formal adjoint
complex (Section~\ref{sec:calabi-adj}), whose initial differential
operator turned out to be equivalent to the rank-$(n-2)$ Killing-Yano
operator, in analogy with the Killing operator in the original complex.

In Sections~\ref{sec:sheaves} and~\ref{sec:killing} we showed that the
Calabi complex is elliptic and locally exact. Hence, it resolves the
sheaf of Killing vectors on the given constant curvature
pseudo-Riemannian manifold. The same is true for the formal adjoint
complex and the sheaf of rank-$(n-2)$ Killing-Yano tensors. Thus the
cohomology of the Calabi complex could be expressed in terms the Killing
sheaf cohomology, while that of its formal adjoint in terms of the
Killing-Yano sheaf cohomology. When a sheaf is locally constant
(covering the relevant cases on constant curvature pseudo-Riemannian
manifolds), its cohomology can be effectively computed in many
circumstances using tools from algebraic topology, thus enabling
effective computation of the Calabi cohomology. These methods were
reviewed in Section~\ref{sec:killing}, specialized to the Killing sheaf.

Finally, in Section~\ref{sec:appl}, we discussed a physical application
that motivated this work. Jointly, the results collected in this work,
together with those of~\cite{kh-cohom,kh-peierls,kh-big} imply that
knowledge of Killing and Killing-Yano sheaf cohomologies allows some
degree of control over the degeneracy subspaces of the presymplectic and
Poisson structures within the classical field theory of linearized
gravity on constant curvature backgrounds.

Unfortunately, the above results do not apply directly to linearized
gravity on arbitrary Lorentzian manifolds, only those that have constant
curvature and where the Calabi complex is defined. However, the Calabi
complex serves as a case study for the more general situation and the
same results partially generalize to general backgrounds. In particular,
we can already make the following conclusions. In general, the Calabi
complex will have to be replaced by a different differential complex,
which will likely depend on some of the algebraic characteristics of the
Lorentzian manifold (such as its isometries and the algebraic type of
the curvature tensor and its derivatives). This complex would be
identified, as was the Calabi
complex~\cite{gasqui-goldschmidt-fr,gasqui-goldschmidt}, by the property
of being a formally exact compatibility complex of the Killing operator.
Such a complex is known to exist under general conditions and also have
the property of being elliptic, since the Killing operator is itself
elliptic~\cite{quillen,goldschmidt-lin}. Further, under a generic
condition, it can be shown to be locally exact
(Section~\ref{sec:sh-res}). The local exactness property connects the
cohomology of this complex to that of the Killing sheaf, which can be
effectively computed, at least in many circumstances, when the sheaf is
locally constant. Unfortunately, one piece of the puzzle remains
incomplete. The connection between the cohomology of the formal adjoint
complex and sheaf cohomology depends on the knowledge of the initial
operator in that differential complex, which is the adjoint of the final
operator of the differential complex resolving the Killing sheaf.  In
the Calabi case it is equivalent to the Killing-Yano operator.  However,
since the differential complex is expected to change depending on the
Lorentzian manifold, so is this initial operator. Thus, it is not clear
which sheaf cohomology will replace the Killing-Yano sheaf in the
general case.

Hence, in future work, it would be very interesting to investigate these
differential complex resolutions of the Killing sheaf, especially
computing their differential operators explicitly. Besides the general
existence results~\cite{quillen,goldschmidt-lin}, such a complex has
already been constructed for locally symmetric spaces ($\nabla_a
R_{bcde} = 0$)~\cite{gasqui-goldschmidt-fr,gasqui-goldschmidt}. Also,
heuristic arguments suggest that they could be partially constructed by
linearizing the so-called `ideal' characterizations of certain exact
families of solutions of Einstein's equations. These include
Schwarzschild~\cite{saez-schw}, Kerr~\cite{saez-kerr} and some perfect
fluid~\cite{ferr-fluid} solutions. An `ideal' characterization consists
of a number of tensor fields, locally and covariantly defined using the
metric and its derivatives, which vanish iff the given metric is locally
isometric to a particular geometry from the desired family. For
instance, the vanishing of the Riemann tensor $R$ is an ideal
characterization of the flat geometry, while the vanishing of the
corrected Riemann tensor $R-\bar{R}$ (Section~\ref{sec:calabi-ops}) does
the same for a constant curvature geometry. It should be clear from
these examples, that the linearization of the tensors that constitute
such an ideal characterization gives an operator whose composition with
the Killing operator is formally exact. At the moment it is not
completely clear what geometric interpretation can be given to
subsequent differential operators in the desired formally exact
differential complex.

Finally, one can easily imagine situations where the number of
independent solutions to the Killing equations changes over the
background pseudo-Riemannian manifold. The Killing sheaf is then no
longer locally constant and many of the techniques described in this
work are no longer applicable. In those cases, perhaps some insight can
be gained from the theory of constructible sheaves~\cite[Ch.4]{dimca},
\cite[Ch.VIII]{ks}, which are allowed to deviate from being locally
constant in a controlled way.

\section*{Acknowledgments}
The author would like to thank Mauro Carfora and Wouter von Limbek for
helpful discussions on deformations of isometry and flat principal
bundle structures. Thanks to Micha\l{} Wrochna for feedback on an earlier
version of the manuscript. Thanks also to Benjamin Lang and Alexander
Schenkel for useful references on the geometry of principal bundles. The
kind hospitality of the Erwin Schr\"odinger Institute during the
workshop ``Algebraic Quantum Field Theory: Its Status and Its Future,''
where this work was first publicly presented, is also acknowledged.

\appendix

\section{Young tableaux and irreducible $\GL(n)$ representations}
\label{app:young}

\subsection{Basic background}\label{app:yt-bkg}
A Yong diagram of type $(r_1,r_2,\ldots)$ with $k$ cells consists of a
number of rows of cells of non-increasing lengths $r_i$, $r_{i+1} \le
r_i$, such that $\sum_i r_i = k$. Example:
\begin{equation*}
	\text{type}~(3,3,1)~\text{or}~(3^2,1),
	\quad \text{diagram}~~\yd{3,3,1} ~ .
\end{equation*}
Given a Young diagram with $k$ cells, an instance of the corresponding
$\GL(n)$ irrep can be realized as the image of the space of covariant
$k$-tensors after two projections: assign an independent tensor index to
each cell of the diagram, symmetrize over each row, anti-symmetrize over
each column. The composition of these operations is called a \emph{Young
symmetrizer}, which we will denote by $\Y^d$, where $d=(r_1,r_2,\ldots)$
is the type of the Young diagram. It will be convenient for us to group
the indices of a symmetrized tensor by the columns of the corresponding
diagram, separating them by a colon. For instance, we write $b_{abc:de}$
corresponding to the filling
\begin{equation*}
	\yt{{a}{d},{b}{e},c} ~ .
\end{equation*}
Here's an example of a simple Young symmetrizer:
\begin{equation}
	\Y^{(2,1)}[t]_{ab:c}
	= \frac{1}{4} (t_{abc} + t_{cba} - t_{bac} - t_{abc}) .
\end{equation}

Different permutations of tensor indices filling a Young diagram create
distinct Young symmetrizers, unless the permutation preserves the
columns. The images of the Young symmetrizer for given diagram type with
$k$-cells are all isomorphic as $\GL(n)$ representations, but are not
necessarily all identical as subspaces of the space of covariant
$k$-tensors. The reason for this observation is that the space of
covariant $k$-tensors is a reducible $\GL(n)$ representation that
decomposes into a sum of irreps corresponding to all possible diagram
types with $k$-cells, but with, in general, non-trivial multiplicities.
Both the dimension and the multiplicity of each occurring irrep can be
computed with the so-called \emph{hook formulas}. The \emph{hook length}
for a given cell is the number of cells constituting a hook with vertex
at the given cell, extending to the right and down. \emph{Multiplicity:}
$k!$ divided by the product of the hook lengths for each cell.
\emph{Dimension:} the product of shifted dimensions for each cell,
divided by the product of hook lengths for each cell; the shifted
dimensions of the cells are obtained by placing $n$ in the top left
cell, then always increasing by $1$ to the right and decreasing by $1$
down. Example:
\begin{gather*}
	\text{hook lengths} ~~ \yt{{5}{3}{2},{4}{2}{1},{1}} ~ , \quad
	\text{shifted dimensions} ~ (n=4) ~~ \yt{{4}{5}{6},{3}{4}{5},{2}} ~ , \\
	\text{multiplicity:} ~ \frac{7!}{(5\cdot 3\cdot 2) (4\cdot 2\cdot 1)(1)} = 21,
	~
	\text{dimension:} ~ \frac{(4\cdot 5\cdot 6) (3 \cdot 4 \cdot 5) (2)}
		{(5\cdot 3\cdot 2) (4\cdot 2\cdot 1)(1)} = 60 .
\end{gather*}

Note that when the number of rows exceeds $n$, the corresponding
representation becomes zero-dimensional. This clearly follows from the
dimension formula and from the more elementary observation that there do
not exist non-trivial fully antisymmetric tensors of rank greater than
$n$, the dimension of the fundamental representation of $\GL(n)$.

By construction, it is clear that every Young symmetrized subspace of
covariant $k$-tensors is fully antisymmetric in the indices
corresponding to each column of its Young diagram. However, this
subspace will actually be even smaller and thus satisfy more identities.
A complete set of identities selecting an irreducible $\GL(n)$
sub-representation of the space of covariant $k$-tensors identified by a
diagram of type $(r_1, \ldots, r_l)$ filled with indices $a^i_k$ ($k$
being the row number and $i$ the column number) consists of (i)
\emph{intracolumn} exchange identities and (ii) \emph{intercolumn}
exchange identities. The exchange of any two indices within a column
changes the tensor by a sign. All such exchanges constitute the
\emph{intracolumn} identities. Let us define a two-column exchange as
follows. Fix two columns $i < j$ and select the top $k$ indices of
column $j$. A two column-exchange consists of a swap between a set of
$k$ indices from column $i$ and the top $k$ indices of column $j$,
without altering the internal order the substituted set of indices. For
a fixed choice of such $i, j, k$ an intercolumn identity consists of the
equality of the tensor with unpermuted indices with the sum over all
corresponding two-column exchanges. All such exchange identities with
consistent choices of $i, j, k$ constitute the \emph{intercolumn}
identities.

There already exists a special notation for antisymmetrization of a
group of indices: inclusion in square brackets, $[a^i_1 a^i_2 \cdots ]$.
Let us introduce a special notation for the sum over all two column
exchanges: fixing integers $i < j$ and $k$, we shall enclose the indices
of column $i$ in curly braces, $\{a^i_1 a^i_2 \cdots \}$, as well as the
top $k$ indices of column $j$, $\{a^j_1 \cdots a^j_k\} a^j_{k+1}
\cdots$. We give explicit examples of intracolumn and intercolumn
identities for Young diagrams of type $(2,2)$ and $(2,2,1)$:
\begin{align}
\label{eq:r-asym1}
	r_{ab:cd}
		&= r_{[ab]:cd} = \frac{1}{2} (r_{ab:cd} - r_{ba:cd}) , \\
\label{eq:r-asym2}
	r_{ab:cd}
		&= r_{ab:[cd]} = \frac{1}{2} (r_{ab:cd} - r_{ab:dc}) , \\
\label{eq:r-bianchi}
	r_{ab:cd}
		&= r_{\{ab\}:\{c\}d} = r_{cb:ad} + r_{ac:bd} , \\
\label{eq:r-sym}
	r_{ab:cd}
		&= r_{\{ab\}:\{cd\}} = r_{cd:ab} , \\
	b_{abc:de}
		&= b_{[abc]:de}
		= \frac{1}{3} (b_{a[bc]:de} + b_{b[ca]:de} + b_{c[ab]:de}) , \\
	b_{abc:de}
		&= b_{abc:[de]} , \\
	b_{abc:de}
		&= b_{\{abc\}:\{d\}e}
		= b_{dbc:ae} + b_{adc:be} + b_{abd:ce} , \\
	b_{abc:de}
		&= b_{\{abc\}:\{de\}}
		= b_{dec:ab} + b_{dbe:ac} + b_{ade:bc} .
\end{align}
It is remarkable, upon noticing the identity $r_{ab:cd} -
r_{\{ab\}:\{c\}d} = 3 r_{[ab:c]d}$, that according to
Equations~\eqref{eq:r-asym1}--\eqref{eq:r-sym} a tensor $r_{ab:cd}$ with
Young symmetry type $(2,2)$ has the same algebraic symmetries as a
Riemann curvature tensor (antisymmetry in $ab$ and $cd$, interchange of
$ab$ with $cd$, and the algebraic Bianchi identity). This fact is
well-known~\cite{fkwc}, but not often mentioned in textbooks on
relativity.

\subsection{Special algebraic and differential operators}\label{app:yt-ops}
Now, suppose that we are working on an $n$-dimensional pseudo-Riemannian
manifold $(M,g)$, with Levi-Civita connection $\nabla$. As in
Section~\ref{sec:calabi-tens}, the Young symmetrizes introduced above
define vector bundles of Young symmetrized covariant tensors $\Y^d T^*M
\to M$, where $d$ stands for a Young diagram. We define special linear
algebraic and differential operators, already briefly discussed in
Section~\ref{sec:calabi-ops}, between these Young symmetrized tensor
bundles occurring in the Calabi complex. Each of the corresponding Young
diagrams has at most two columns, where the first column usually has at
most $n$ cells and the second column has at most two cells. The operator
\emph{trace} ($\tr$) removes one row of cells, \emph{metric exterior
product} ($g \odot -$) adds one row of cells, \emph{left} or \emph{right
exterior derivative} ($\d_L$ and $\d_R$) adds one cell to the left or
right column respectively, and \emph{left} or \emph{right divergence}
($\delta_L$ and $\delta_R$) removes one cell from the left or right
column respectively. The name of each of these operators should be
suggestive of their form, with the main complication being to maintain
appropriate Young symmetry.

In principle, the Littlewood-Richardson decomposition rules uniquely fix
the principal symbols of each of these operators up to a scalar
multiple, with the Levi-Civita operator canonically converting a first
order principal symbol into a first order operator. In practice, it
takes a bit of work to find explicit formulas for them, given that a
naive application a Young symmetrizer produces unmanageably large
expressions. Moreover, the existence of the intracolumn and intercolumn
symmetrization identities introduces non-uniqueness into possible
explicit expressions. Below, we give explicit formulas for these
operators. In case of ambiguity, the choice was dictated by practical
convenience. Then, in Secs.~\ref{app:yt-pres} and~\ref{app:yt-comp} we
show by explicit calculation that they satisfy the required
symmetrization identities and thus carry the correct Young type.

\begin{align}
\label{eq:yop-tr}
	\tr[b]_{a_1\cdots a_l : b}
		&= b_{a_1\cdots a_l c : b}{}^c , \\
\label{eq:yop-g}
	(g\odot t)_{a_1 \cdots a_l:bc}
		&= l (g_{b[a_1} t_{a_2\cdots a_l]:c} - g_{c[a_1} t_{a_2\cdots a_l]:b}) , \\
\label{eq:yop-dl}
	\d_L [b]_{a_1\cdots a_l : bc}
		&= l \nabla_{[a_1} b_{a_2 \cdots a_l] : bc} ,  \\
\label{eq:yop-Dl}
	\delta_L[b]_{a_1 \cdots a_l : bc}
		&= \nabla^a b_{a a_1\cdots a_l:bc}
			+ 2l^{-1}\nabla^a b_{[b|a_1\cdots a_l:|c]a} , \\
\label{eq:yop-dr}
	\d_R[t]_{a_1\cdots a_l:bc}
		&= 2\nabla_{[b} t_{|a_1 \cdots a_l:|c]}
				+ 2(l-1)^{-1} \nabla_{[\{b\}|} t_{\{a_1\cdots a_l\}:|c]} , \\
\label{eq:yop-Dr}
	\delta_R[b]_{a_1\cdots a_l:b}
		&= \nabla^c b_{a_1\cdots a_l:bc} .
\end{align}
Let us give an explicit example of~\eqref{eq:yop-g} for $l=2$, which
appears in the formulas for the constant curvature Riemann
tensor~\eqref{eq:Rb-def} and for the linearized Riemann curvature
operator~\eqref{eq:linR-def}:
\begin{align}
	(g\odot h)_{a_1a_2:bc}
		&= g_{a_1b} h_{a_2c} - g_{a_2b} h_{a_1c} - g_{a_1c} h_{a_2b}
			+ g_{a_2c} g_{a_1b} , \\
	(g\odot g)_{a_1a_2:bc}
		&= 2 (g_{a_1b} g_{a_2c} - g_{a_2b} g_{a_1c}) , \\
	(\nabla\nabla\odot h)_{a_1a_2:bc}
		&= \nabla\nabla_{a_1b} h_{a_2c}
			-\nabla\nabla_{a_2b} h_{a_1c}
			-\nabla\nabla_{a_1c} h_{a_2b}
			+\nabla\nabla_{a_2c} g_{a_1b} , \\
\notag
		& \qquad \text{where} ~~
			\nabla\nabla_{ab}
				= \nabla_{(a}\nabla_{b)} 
				= \frac{1}{2}(\nabla_{a}\nabla_{b}+\nabla_{b}\nabla_{a}) .
\end{align}
In the last equation we used the $\odot$ operation to define another
differential operator of definite Young type. This property follows
directly from that of~\eqref{eq:yop-g}.

\subsection{Preservation of Young type}\label{app:yt-pres}
Each of the operators~\eqref{eq:yop-tr}--\eqref{eq:yop-Dr} maps tensors
of one Young type into another one, as is indicated by the index
notation described in Sec.~\ref{app:yt-bkg}. Below we explicitly
demonstrate that, by showing that the result of applying one of these
operators to a tensor of a given Young type always satisfies the
required intracolumn and intercolumn identities.

First, we list some key identities satisfied by the idempotent
antisymmetrization and column exchange operations. They follow from
straightforward, though possibly lengthy, application of their
definitions. Here, a tensor $t_{a_1\cdots a_l}$ is assumed to be fully
antisymmetric. Also, to simplify the notation for nested operations, we
use the notation $\{\cdots\}^k_l$, where the braces necessarily enclose
the indices $a_k$, $a_{k+1}$, \ldots, $a_{l}$, though perhaps also
others, to mean that we apply the appropriate column exchange operation
to these $a_i$ indices as if they appeared in the order $\{a_k\cdots
a_l\}$.

% Uncomment to see detailed derivations.
\begin{align}
%\MoveEqLeft
%\notag
	(l+1) p_{[a} t_{a_1\cdots a_l]}
%		= p_a t_{a_1\cdots a_l} - \sum_i (-)^{i-1} p_{a_i}
%			t_{a a_1\cdots \hat{a}_i\cdots a_l} \\
\label{eq:id-asym}
		&= p_a t_{a_1\cdots a_l} - p_{\{a\}} t_{\{a_1\cdots a_l\}} , \\
%\MoveEqLeft
%\notag
	p_{\{b\}} t_{\{a_1\cdots a_l\}}
%		= \sum_i (-)^{i-1} p_{a_i} t_{b a_1\cdots \hat{a}_i\cdots a_l} \\
%\notag
%		&= p_{a_1} t_{b a_2\cdots a_l}
%			- \sum_{i>1} (-)^{i-1} p_{a_i} t_{a_1 b a_2\cdots \hat{a}_i\cdots a_l} \\
\label{eq:id-c1}
		&= p_{a_1} t_{b a_2\cdots a_l} + p_{\{b\}} t_{a_1\{a_2\cdots a_l\}} , \\
%\MoveEqLeft
%\notag
	p_{\{a\}} t_{b\{a_1\cdots a_l\}}
%		= \sum_i (-)^{i-1} p_{a_i} t_{ba a_1\cdots \hat{a}_i\cdots a_l} \\
%\notag
%		&= -\sum_i(-)^{i-1} p_{a_i} t_{ab a_1\cdots \hat{a}_i\cdots a_l} \\
\label{eq:id-c1sw}
		&= - p_{\{b\}} t_{a\{a_1\cdots a_l\}} , \\
%\MoveEqLeft
%\notag
	p_{\{b} q_{c\}} t_{\{a_1\cdots a_l\}}
%		= \sum_{i<j} (-)^{i+j-1} p_{a_i} q_{a_j}
%			t_{bc a_1\cdots \hat{a}_i \cdots \hat{a}_j \cdots a_l} \\
%\notag
%		&= \sum_{1<j} (-)^{(j-1)-1} p_{a_1} q_{a_j}
%				t_{bc a_2\cdots \hat{a}_j \cdots a_l} \\
%\notag
%		& \quad {}
%			+ \sum_{1<i<j} (-)^{(i-1)+(j-1)-1} p_{a_i} q_{a_j}
%			t_{bc a_1 a_2\cdots \hat{a}_i \cdots \hat{a}_j \cdots a_l} \\
\label{eq:id-c2}
		&= p_{a_1} q_{\{c\}} t_{b\{a_2\cdots a_l\}}
			+ p_{\{b} q_{c\}} t_{a_1\{a_2\cdots a_l\}} , \\
\notag
\MoveEqLeft[7]
	\left(p_{b'} t_{\{a_1\cdots a_l\}:c'} - p_{c'} t_{\{a_1\cdots a_l\}:b'}\right)
\delta^{b'}_{\{b} \delta^{c'}_{c\}} \\
%\notag
%		&= \sum_{i<j} (-)^{i+j-1} p_{a_i} t_{bc a_1\cdots \hat{a}_i \cdots
%\hat{a}_j \cdots a_l:a_j} \\
%\notag
%		& \quad {}
%			-\sum_{i<j} (-)^{i+j-1} p_{a_j} t_{bc a_1\cdots \hat{a}_i \cdots \hat{a}_j \cdots a_l:a_i}
%			~~ \text{[let $i\leftrightarrow j$]} \\
%\notag
%		&= \sum_{j<i} (-)^{(i-1)+j-1} p_{a_i} t_{bc a_1\cdots \hat{a}_i
%\cdots \hat{a}_j \cdots a_l:a_j}
%			+\sum_{i<j} (-)^{(i-1)+j} p_{a_i} t_{bc a_1\cdots \hat{a}_i \cdots \hat{a}_j \cdots a_l:a_j} \\
%\notag
%		& \quad {}
%			+ \sum_i (-)^{i-1} p_{a_i} t_{c a_1\cdots \hat{a}_i \cdots a_l:b}
%			- \sum_i (-)^{i-1} p_{a_i} t_{c a_1\cdots \hat{a}_i \cdots a_l:b} \\
%\notag
%		&= \sum_i (-)^{i-1} p_{a_i} t_{\{b a_1\cdots \hat{a}_i \cdots a_l\}:c}
%			-\sum_i (-)^{i-1} p_{a_i} t_{c a_1\cdots \hat{a}_i \cdots a_l:b} \\
\label{eq:id-c2m}
		&= p_{\{b\}} t_{\{a_1\cdots a_l\}:c}
			-p_{\{c\}} t_{\{a_1\cdots a_l\}:b} , \\
\MoveEqLeft[7]
\notag
	\left(t_{b'\{a_1\cdots a_l\}:c'a}-t_{b'\{a_1\cdots a_l\}:c'a}\right)
			\delta^{b'}_{\{b} \delta^{c'}_{c\}} \\
%\notag
%		&= \sum_{i<j} (-)^j t_{bc a_1\cdots \hat{a}_j\cdots a_l:a_j a}
%			-\sum_{i<j} (-)^i t_{cb a_1\cdots \hat{a}_i\cdots a_l:a_i a}
%				~~ \text{[let $i\leftrightarrow j$]} \\
%\notag
%		&= \sum_{i\ne j} (-)^j t_{bc a_1\cdots \hat{a}_j\cdots a_l:a_j a}
%		= -(l-1) t_{b\{a_1\cdots a_l\}:\{c\}a}
%			~~ \eqref{eq:id-c1} \\
%\notag
%		&= -(l-1) t_{\{b a_1\cdots a_l\}:\{c\}a}+(l-1) t_{c a_1\cdots a_l:ba} \\
\label{eq:id-c2a}
		&= -2(l-1) t_{[b|a_1\cdots a_l:|c]} , \\
%\MoveEqLeft
%\notag
	p_{\{\{b\}\}} t_{\{\{a_1\cdots a_l\}\}}
%		= \sum_i (-)^{i-1} \delta^{b'}_{\{b\}} p_{\{a_i}
%			t_{b' a_1\cdots \hat{a}_i\cdots a_l\}^1_l} \\
%\notag
%		&= -\sum_{j\ne i} (-)^{i-1} p_{a_i}
%				t_{b a_1\cdots \hat{a}_i\cdots a_l}
%			+ \sum_{i} (-)^{i-1} p_{b} t_{a_1\cdots a_l} \\
\label{eq:id-c1i}
		&= l p_b t_{a_1\cdots a_l} - (l-1) p_{\{b\}} t_{\{a_1\cdots a_l\}} , \\
%\MoveEqLeft
%\notag
	p_{\{b\}} q_{\{\{c\}} t_{\{a_1\cdots a_l\}\}^1_l}
%		= \sum_i (-)^{i-1} p_{\{b\}} q_{\{a_i} t_{c a_1\cdots \hat{a}_i \cdots
%a_l\}^1_l} \\
%\notag
%		&= \sum_i (-)^{i-1} p_{a_i} q_b t_{c a_1\cdots \hat{a}_i\cdots a_l}
%			+\sum_{j<i} (-)^{i+j} p_{a_j} q_{a_i} t_{cb a_1\cdots \hat{a}_j \cdots \hat{a}_i \cdots a_l}
%			~~ \text{[let $i\leftrightarrow j$]} \\
%\notag
%		& \quad {}
%			+\sum_{i<j} (-)^{i+j-1} p_{a_j} q_{a_i} t_{cb a_1\cdots \hat{a}_i \cdots \hat{a}_j \cdots a_l} \\
%\notag
%		&= p_{\{c\}} q_b t_{\{a_1\cdots a_l\}}
%			+ \sum_{i<j} (-)^{i+j-1} p_{a_i} q_{a_j} t_{bc a_1\cdots \hat{a}_i
%\cdots \hat{a}_j \cdots a_l} \\
%\notag
%		& \quad {}
%			- \sum_{i<j} (-)^{i+j-1} q_{a_i} p_{a_j} t_{bc a_1\cdots \hat{a}_i \cdots \hat{a}_j \cdots a_l} \\
\label{eq:id-c1j}
		&= p_{\{c\}} q_b t_{\{a_1\cdots a_l\}}
			+ (p_{\{b} q_{c\}} - q_{\{b} p_{c\}}) t_{\{a_1\cdots a_l\}} , \\
\notag
\MoveEqLeft[8]
	\left( p_{\{\{b'\}} t_{\{a_1\cdots a_l\}\}^1_l:c'}
	- p_{\{\{c'\}} t_{\{a_1\cdots a_l\}\}^1_l:b'} \right)
			\delta^{b'}_{\{b} \delta^{c'}_{c\}} \\
\label{eq:id-c1i2}
		&= 2(l-1) p_{[b|} t_{a_1\cdots a_l:|c]}
			- 2(l-2) p_{[\{b\}|} t_{\{a_1\cdots a_l\}:|c]} , \\
%\MoveEqLeft
%\notag
	p_{\{\{a\}} t_{\{a_1\cdots a_l\}\}^1_l:\{bc\}}
%		\\
%\notag
%		&= \sum_i (-)^{i-1} p_{\{a_i} t_{a a_1\cdots \hat{a}_i\cdots a_l\}^1_l:\{bc\}} \\
%\notag
%		&= \sum_i (-)^{i-1} p_{a_i} t_{a\{a_1\cdots \hat{a}_i\cdots a_l\}:\{bc\}}
%			+\sum_i (-)^{i-1} p_{b} t_{a a_1\cdots \hat{a}_i\{\cdots a_l\}:a_i\{c\}}
%				~~ \eqref{eq:id-c2} \\
%\notag
%		& \quad {}
%			+\sum_i (-)^{i-1} p_{c} t_{a \{a_1\cdots\} \hat{a}_i\cdots a_l:\{b\}a_i}
%				~~ \text{[let $i\leftrightarrow j$]} \\
%\notag
%		&= \sum_i (-)^{i-1} p_{a_i} t_{\{a a_1\cdots \hat{a}_i\cdots a_l\}:\{bc\}}
%			+\sum_i (-)^{i-1} p_{a_i} t_{b\{a_1\cdots \hat{a}_i\cdots a_l\}:a\{c\}}
%				~~ \eqref{eq:id-c1} \\
%\notag
%		& \quad {}
%			+\sum_{i<j} (-)^{i+j-1}
%				(p_b t_{ac a_1\cdots \hat{a}_i\cdots \hat{a}_j\cdots a_l:a_ia_j}
%				-p_c t_{ab a_1\cdots \hat{a}_i\cdots \hat{a}_j\cdots a_l:a_ia_j}) \\
%\notag
%		&= p_{\{a\}} t_{\{a_1\cdots a_l\}:bc}
%			+ \sum_i (-)^{i-1}
%				(p_{a_i} t_{\{b a_1\cdots \hat{a}_i\cdots a_l\}:a\{c\}}
%				-p_{a_i} t_{c a_1\cdots \hat{a}_i\cdots a_l:ab}) \\
%\notag
%		& \quad {}
%			+ p_b t_{\{a_1\cdots a_l\}:\{ac\}}
%			- p_c t_{\{a_1\cdots a_l\}:\{ab\}} \\
\notag
		&= p_{\{a\}} t_{\{a_1\cdots a_l\}:bc} \\
\label{eq:id-c1i2m}
		& \quad {}
			+ 2 p_{[\{b\}|} t_{\{a_1\cdots a_l\}:|c]a}
			- 2 p_{[b|} t_{a_1\cdots a_l:|c]a} .
\end{align}

Next, we show how the above key identities can be used to explicitly
demonstrate that the required symmetrization identities are satisfied.
We try to indicate which of the key identities are used and where, while
also silently making use of the symmetrization properties of the Young
type tensors on which the operations are being performed.

For the \emph{trace}~\eqref{eq:yop-tr}, the intracolumn identities are
obvious, so there is only one intercolumn identity to check:
\begin{align*}
	\tr[b]_{\{a_1\cdots a_l\}:\{b\}}
		&= b_{\{a_1\cdots a_l\}c:\{b\}}{}^c
			~~ \eqref{eq:id-c1} \\
		&= b_{\{a_1\cdots a_l c\}:\{b\}}{}^c - b_{\{a_1\cdots a_l b\}:c}{}^c
		= b_{a_1 \cdots a_l c:b}{}^c \\
		&= \tr[b]_{a_1 \cdots a_l:b} .
\end{align*}

For the \emph{metric exterior product}~\eqref{eq:yop-g}, the intracolumn
identities are obvious, so there are two intercolumn identities to
check:
\begin{align*}
	(g\odot t)_{\{a_1 \cdots a_l\}:\{b\}c}
		&= l(g_{\{b\}\{[a_1} t_{a_2\cdots a_l]\}:c}
			- g_{c\{[a_1} t_{a_2\cdots a_l]\}:\{b\}})
			~~ \eqref{eq:id-asym} \\
		&= -l(l+1) (g_{[b[a_1} t_{a_2\cdots a_l]]:c}
			- g_{c[[a_1} t_{a_2\cdots a_l]:b]}) \\
		& \quad {}
			+ l (g_{b[a_1} t_{a_2\cdots a_l]:c}
				- g_{c[a_1} t_{a_2\cdots a_l]:b}) \\
		&= (g\odot t)_{a_1\cdots a_l:bc} , \\
	(g\odot t)_{\{a_1\cdots a_l\}:\{bc\}}
		&= l(g_{b'\{[a_1} t_{a_2\cdots a_l]\}:c'}
			- g_{c'\{[a_1} t_{a_2\cdots a_l]\}:b'})
			\delta^{b'}_{\{b} \delta^{c'}_{c\}}
			~~ \eqref{eq:id-c2m} \\
		&= l (g_{\{b\}\{[a_1} t_{a_2\cdots a_l]\}:c}
			- g_{\{c\}\{[a_1} t_{a_2\cdots a_l]\}:b})
			~~ \eqref{eq:id-asym} \\
		&= -l(l+1) (g_{[b[a_1} t_{a_2\cdots a_l]]:c}
			- g_{[c[a_1} t_{a_2\cdots a_l]]:b}) \\
		& \quad {}
			+ l (g_{b[a_1} t_{a_2\cdots a_l]:c}
				- g_{c[a_1} t_{a_2\cdots a_l]:b}) \\
		&= (g\odot t)_{a_1\cdots a_l:bc} .
\end{align*}
The double anti-symmetrizations vanished because of the identities
$g_{[ab]} = 0$ and $t_{[a_2\cdots a_l:a_1]} = 0$, with the latter
following from a combination of~\eqref{eq:id-asym} and an intercolumn
identity. Also, we have used the fact that $p_{[a_1} t_{a_2\cdots
a_l]:b}$ is a tensor of the corresponding Young type, which follows from
the identities in the paragraph below.

For the \emph{left exterior derivative}~\eqref{eq:yop-dl}, the
intracolumn identities are obvious, so there are two intercolumn
identities to check:
\begin{align*}
	\d_L[b]_{\{a_1\cdots a_l\}:\{b\}c}
		&= l\nabla_{\{[a_1} b_{a_2\cdots a_l]\}:\{b\}c}
			~~ \eqref{eq:id-asym} \\
		&= -l(l+1) \nabla_{[[a_1} b_{a_2\cdots a_l]:b]c}
			+ l \nabla_{[a_1} b_{a_2\cdots a_l]:bc}
		= \d_L[b]_{a_1\cdots a_l:bc} , \\
	\d_L[b]_{\{a_1\cdots a_l\}:\{bc\}}
		&= l\nabla_{\{[a_1} b_{a_2\cdots a_l]\}:\{bc\}}
			~~ \eqref{eq:id-asym} \\
		&= \nabla_{\{a_1} b_{a_2\cdots a_l\}:\{bc\}}
			- \nabla_{\{\{a_1\}} b_{\{a_2\cdots a_l\}\}^1_l:\{bc\}}
			~~ \eqref{eq:id-c2} \\
		&= \nabla_{a_1} b_{\{a_2\cdots a_l\}:\{bc\}}
			+ \nabla_{b} b_{\{a_2\cdots a_l\}:a_1\{c\}} \\
		& \quad {}
			- \nabla_{\{\{a_1\}} b_{\{a_2\cdots a_l\}\}^2_l:\{bc\}}
			- \nabla_{\{\{b\}} b_{\{a_2\cdots a_l\}\}^2_l:a_1\{c\}}
			~~ \eqref{eq:id-c1i2m}, \eqref{eq:id-c1j} \\
		&= \nabla_{a_1} b_{a_2\cdots a_l:bc}
			+ \nabla_b b_{a_2\cdots a_l:a_1 c}
			- \nabla_{c} b_{\{a_2\cdots a_l\}:a_1\{b\}} \\
		& \quad {}
			- \nabla_{\{a_1\}} b_{\{a_2\cdots a_l\}:bc}
			- 2\nabla_{[\{b\}|} b_{\{a_2\cdots a_l\}:|c]a_1}
			+ 2\nabla_{[b|} b_{a_2\cdots a_l:|c] a_1} \\
		& \quad {}
			- (\nabla_{b'} b_{\{a_2\cdots a_l\}:a_1 c'}
				-\nabla_{c'} b_{\{a_2\cdots a_l\}:a_1 b'})
				\delta^{c'}_{\{c} \delta^{b'}_{b\}}
			~~ \eqref{eq:id-c2m} \\
		&= l\nabla_{[a_1} b_{a_2\cdots a_l]:bc}
			- 2\nabla_{[\{b\}|} b_{\{a_2\cdots a_l\}:|c]a_1} \\
		& \quad {}
			+ 2\nabla_{[\{c\}|} b_{\{a_2\cdots a_l\}:a_1|b]} \\
		&= l\nabla_{[a_1} b_{a_2\cdots a_l]:bc}
		= \d_L[b]_{a_1\cdots a_l:bc} .
\end{align*}

For the \emph{left divergence}~\eqref{eq:yop-Dl}, the intracolumn
identities are obvious. It remains to check the two intercolumn identities:
\begin{align*}
	\delta_L[b]_{\{a_1\cdots a_l\}:\{b\}c}
		&= \nabla^a (b_{a\{a_1\cdots a_l\}:\{b\}c}
			+ l^{-1} b_{\{b\}\{a_1\cdots a_l\}:ca}
			- l^{-1} b_{c\{a_1\cdots a_l\}:\{b\}a}) \\
		&= \nabla^a (b_{\{a a_1\cdots a_l\}:\{b\}c}
			+ b_{b a_1\cdots a_l:ca}
				~~ \eqref{eq:id-c1} \\
		& \quad {}
			- l^{-1}(l+1)b_{[b a_1\cdots a_l]:ca} + l^{-1} b_{b a_1\cdots a_l:ca}
				~~ \eqref{eq:id-asym} \\
		& \quad {}
			- l^{-1} b_{\{c a_1\cdots a_l\}:\{b\}a}
			+ l^{-1} b_{b a_1\cdots a_l:ca})
				~~\eqref{eq:id-c1} \\
		&= \nabla^a (b_{a a_1\cdots a_l:bc}
			+ 2l^{-1} b_{[b| a_1\cdots a_l:|c]a})
		= \delta_L[b]_{a_1\cdots a_l:bc} , \\
	\delta_L[b]_{\{a_1\cdots a_l\}:\{bc\}}
		&= \nabla^a (b_{a\{a_1\cdots a_l\}:\{bc\}}
			~~ \eqref{eq:id-c2} \\
		& \quad {}
			+ l^{-1} (b_{b'\{a_1\cdots a_l\}:c'a} - b_{c'\{a_1\cdots a_l\}:b'a})
				\delta_{\{b}^{b'}\delta_{c\}}^{c'})
				~~ \eqref{eq:id-c2a} \\
		&= \nabla^a (b_{\{a a_1\cdots a_l\}:\{bc\}}
			+ b_{b\{a_1\cdots a_l\}:\{c\}a}
				~~ \eqref{eq:id-c1} \\
		& \quad {}
			- l^{-1}(l-1)
					(b_{b a_1\cdots a_l:ca}
					-b_{c a_1\cdots a_l:ba})) \\
		&= \nabla^a (b_{a a_1\cdots a_l:bc}
			+ b_{\{b a_1\cdots a_l\}:\{c\}a} - b_{c a_1\cdots a_l:ba} \\
		& \quad {}
			- (1-l^{-1})
					(b_{b a_1\cdots a_l:ca}
					-b_{c a_1\cdots a_l:ba})) \\
		&= \nabla^a (b_{a a_1\cdots a_l:bc}
				+ l^{-1} b_{b a_1\cdots a_l:ca} - l^{-1} b_{c a_1\cdots a_l:ba}) \\
		&= \delta_L[b]_{a_1\cdots a_l:bc}
\end{align*}

For the \emph{right exterior derivative}~\eqref{eq:yop-dr}, the
following rewriting makes the intracolumn identities obvious:
\begin{align*}
	\d_R[b]_{a_1\cdots a_l:bc}
		&= \nabla_{b} b_{a_1\cdots a_l:c} - \nabla_{c} b_{a_1\cdots a_l:b} \\
		& \quad {}
			+(l-1)^{-1} \left(
				\nabla_{\{b\}} b_{\{a_1\cdots a_l\}:c}
				-\nabla_{\{c\}} b_{\{a_1\cdots a_l\}:b} \right)
			~~ \eqref{eq:id-asym} \\
		&= \nabla_{b} b_{a_1\cdots a_l:c} - \nabla_{c} b_{a_1\cdots a_l:b} \\
		& \quad {}
			- (l-1)^{-1} (l+1)
				(\nabla_{[b} b_{a_1\cdots a_l]:c}
				-\nabla_{[c} b_{a_1\cdots a_l]:b}) \\
		& \quad {}
			+ (l-1)^{-1}
				(\nabla_b b_{a_1\cdots a_l:c} - \nabla_c b_{a_1\cdots a_l:b}) .
\end{align*}
There are also two intercolumn identities to check:
\begin{align*}
	\d_R[b]_{\{a_1\cdots a_l\}:\{b\}c}
		&= \nabla_{\{b\}} b_{\{a_1\cdots a_l\}:c}
				- \nabla_{c} b_{\{a_1\cdots a_l\}:\{b\}} \\
		& \quad {}
			+ (l-1)^{-1}
				(\nabla_{\{\{b\}\}} b_{\{\{a_1\cdots a_l\}\}:c}
				-\nabla_{\{\{c\}} b_{\{a_1\cdots a_l\}\}^1_l:\{b\}}) \\
		&= \nabla_{\{b\}} b_{\{a_1\cdots a_l\}:c}
				- \nabla_{c} b_{a_1\cdots a_l:b} \\
		& \quad {}
			+ (l-1)^{-1} l \nabla_b b_{a_1\cdots a_l:c}
			- \nabla_{\{b\}} b_{\{a_1\cdots a_l\}:c}
			~~ \eqref{eq:id-c1i} \\
		& \quad {}
			- (l-1)^{-1} \nabla_b b_{\{b_1\cdots a_l\}:\{c\}}
			~~ \eqref{eq:id-c1j} \\
		& \quad {}
			- (l-1)^{-1} (\nabla_{c'} b_{a_1\cdots a_l:b'} - \nabla_{b'}
			  b_{a_1\cdots a_l:c'}) \delta^{b'}_{\{b} \delta^{c'}_{c\}}
			~~ \eqref{eq:id-c2m} \\
		&= \nabla_b b_{a_1\cdots a_l:c} - \nabla_c b_{a_1\cdots a_l:b} \\
		& \quad {}
			+ (l-1)^{-1} (\nabla_{\{b\}} b_{\{a_1\cdots a_l\}:c}
				- \nabla_{\{c\}} b_{\{a_1\cdots a_l\}:b}) \\
		&= \d_R[b]_{a_1\cdots a_l:bc} , \\
	\d_R[b]_{\{a_1\cdots a_l\}:\{bc\}}
		&= (\nabla_{b'} b_{\{a_1\cdots a_l\}:c'}
				- \nabla_{c'} b_{\{a_1\cdots a_l\}:b'})
			\delta^{b'}_{\{b} \delta^{c'}_{c\}}
			~~ \eqref{eq:id-c2m}, \eqref{eq:id-c1i2} \\
		& \quad {}
			+ (l-1)^{-1}
				(\nabla_{\{\{b'\}} b_{\{a_1\cdots a_l\}\}^a_l:c'}
				-\nabla_{\{\{c'\}} b_{\{a_1\cdots a_l\}\}^a_l:b'})
				\delta^{b'}_{\{b} \delta^{c'}_{c\}} \\
		&= \nabla_{\{b\}} b_{\{a_1\cdots a_l\}:c}
			-\nabla_{\{c\}} b_{\{a_1\cdots a_l\}:b} \\
		& \quad {}
			+ (\nabla_{b} b_{a_1\cdots a_l:c} - \nabla_{c} b_{a_1\cdots a_l:b}) \\
		& \quad {}
			- (l-1)^{-1}(l-2)
					(\nabla_{\{b\}} b_{\{a_1\cdots a_l\}:c}
					-\nabla_{\{c\}} b_{\{a_1\cdots a_l\}:b}) \\
		&= \nabla_{b} b_{a_1\cdots a_l:c} - \nabla_{c} b_{a_1\cdots a_l:b} \\
		& \quad {}
			+ (l-1)^{-1}
				(\nabla_{\{b\}} b_{\{a_1\cdots a_l\}:c}
				-\nabla_{\{c\}} b_{\{a_1\cdots a_l\}:b}) \\
		&= \d_R[b]_{a_1\cdots a_l:bc} .
\end{align*}

For the \emph{right divergence}~\eqref{eq:yop-Dr}, the intracolumn
identities are obvious, so there is only one intercolumn identity to
check:
\begin{align*}
	\delta_R[b]_{\{a_1\cdots a_l\}:\{b\}}
		&= \nabla^c b_{\{a_1 \cdots a_l\}:\{b\}c}
		= \nabla^c b_{a_1 \cdots a_l:bc} = \delta_R[b]_{a_1\cdots a_l:b} .
\end{align*}

\subsection{Composition identities}\label{app:yt-comp}
Below, we list identities between some possible compositions of the
operators~\eqref{eq:yop-tr}--\eqref{eq:yop-Dr}. These will be
instrumental in the following Sec.~\ref{app:yt-calabi}, where they will
be used to explicitly define the operators involved in the Calabi
complex~\eqref{eq:calabi-diag} and the necessary identities between
them. We do not show the necessary explicit calculations, as they are
lengthy but straightforward. It suffices to make use of the key
identities~\eqref{eq:id-asym}--\eqref{eq:id-c1i2m}, as explicitly
illustrated in Sec.~\ref{app:yt-pres}.

Recall that $\nabla$ denotes the Levi-Civita connection on a
pseudo-Riemannian space of constant curvature with metric $g$ and
dimension $n$. The Riemann tensor on this space is defined by the
convention $2\nabla_{[a}\nabla_{b]} \omega_c = \bar{R}_{ab:c}{}^d
\omega_d$ and is explicitly equal to
\begin{equation}\label{eq:Rb-def}
	\bar{R}_{ab:cd} = \frac{\lambda}{2} (g\odot g)_{ab:cd}
		= \lambda (g_{ac} g_{bd} - g_{ad} g_{bc}) ,
	\quad \text{with} ~~
	\lambda = \frac{k}{n(n-1)} ,
\end{equation}
such that $k = g^{ac} \bar{R}_{ab:c}{}^b$ is the curvature constant.

The simplest composition identity is of two left exterior derivatives
($l\ge 4$):
\begin{equation}\label{eq:dldl}
	\d_L\circ\d_L[b]_{a_1\cdots a_l:bc} = 0 .
\end{equation}
The principal symbols of these operators augment the left index column
of the argument and antisymmetrize over it, thus composing to zero, as
in the case of the de~Rham differential. This means that, at worst, the
result of the composition of the operators is of order zero and
proportional to the background curvature~$\bar{R}$ given
in~\eqref{eq:Rb-def}. Now, note that the background curvature is
$\GL(n)$-equivariantly composed only out of the metric and the
composition $\d_L\circ\d_L$ is also an equivariant operator (taking into
account the transformation properties of the covariant derivative and
the metric). Then the result of the composition (Young type
$(2,2,1^{l-2})$) must be equivariantly composed only out of the metric
$g$ (Young type $(2)$) and the argument $b$ (Young type
$(2,2,1^{l-4})$). However, according the Littlewood-Richardson
rules~\cite{fulton,lrr}, there is no non-trivial combination of that
kind. Therefore, the composition of these operator must vanish.

Next, we show the relation between the compositions $\delta_L\circ\d_L$
and $\d_L\circ\delta_L$, along with some auxiliary identities involving
the curvature. These formulas hold when the length of the left index
column of the output is $l>2$.
\begin{align}
\label{eq:ddRb2}
	2\nabla_{[a}\nabla_{b]} b_{a_1\cdots a_l:cd}
		&= (\bar{R}\cdot b)_{ab~a_1\cdots a_l:cd} , \\
\notag
	(\bar{R}\cdot b)_{ab~a_1\cdots a_l:cd}
		&= \bar{R}_{ab:\{e\}}{}^e b_{\{a_1\cdots a_l\}:cd}
			+ \bar{R}_{ab:c}{}^e b_{a_1\cdots a_l:ed}
			+ \bar{R}_{ab:d}{}^e b_{a_1\cdots a_l:ce} \\
\notag
		&= \lambda (g_{a\{b\}} - g_{b\{a\}}) b_{\{a_1\cdots a_l\}:cd} \\
\notag
		& \quad {}
			+ \lambda (g_{ac} b_{a_1\cdots a_l:bd} - g_{bc} b_{a_1\cdots a_l:ad}) \\
\label{eq:Rb2}
		& \quad {}
			- \lambda (g_{ad} b_{a_1\cdots a_l:bc} - g_{bd} b_{a_1\cdots a_l:ac}) ;
\end{align}

\begin{align*}
\MoveEqLeft
	(l+1)(\bar{R}\cdot b)_{a[b~a_1\cdots a_l]:cd} \\
%		&= (l+1) \lambda (g_{a[\{b\}} b_{\{a_1\cdots a_l\}]:cd}
%				- g_{\{a\}[b} b_{\{a_1\cdots a_l\}]:cd}) \\
%		& \quad {}
%			+ (l+1)\lambda (g_{ac} b_{[a_1\cdots a_l:b]d}
%				- g_{c[b} b_{a_1\cdots a_l]:ad}) \\
%		& \quad {}
%			- (l+1)\lambda (g_{ad} b_{[a_1\cdots a_l:b]c}
%				- g_{d[b} b_{a_1\cdots a_l]:ac}) \\
%		&= \lambda (g_{a\{b\}} b_{\{a_1\cdots a_l\}:cd}
%				- g_{a\{\{b\}\}} b_{\{\{a_1\cdots a_l\}\}:cd}) \\
%		& \quad {}
%			- (l+1)\lambda (g_{c[b} b_{a_1\cdots a_l]:ad}
%				- g_{d[b} b_{a_1\cdots a_l]:ac}) \\
%		&= \lambda (g_{a\{b\}} b_{\{a_1\cdots a_l\}:cd}
%				- l g_{ab} b_{a_1\cdots a_l:cd}
%				+ (l-1) g_{a\{b\}} b_{\{a_1\cdots a_l\}:cd}) \\
%		& \quad {}
%			- (l+1)\lambda (g_{c[b} b_{a_1\cdots a_l]:ad}
%				- g_{d[b} b_{a_1\cdots a_l]:ac}) \\
		&= -l(l+1)\lambda g_{a[b} b_{a_1\cdots a_l]:cd} \\
		& \quad {}
			- (l+1)\lambda (g_{c[b} b_{a_1\cdots a_l]:ad}
				- g_{d[b} b_{a_1\cdots a_l]:ac}) , \\
\MoveEqLeft
	(l+1) (\bar{R}\cdot b)^a{}_{[a~a_1\cdots a_l]:bc} \\
%		&= -l\lambda \delta^a_{[a} b_{a_1\cdots a_l]:bc} \\
%		& \quad {}
%			- (l+1)\lambda (g_{b[a} b_{a_1\cdots a_l]:}{}^a{}_c
%				- g_{c[a} b_{a_1\cdots a_l]:}{}^a{}_b) \\
%		&= -l \lambda
%			(\delta^a_a b_{a_1\cdots a_l:bc}
%			-\delta^a_{\{a\}} b_{\{a_1\cdots a_l\}:bc}) \\
%		& \quad {}
%			+ \lambda (\delta^a_b b_{a_1\cdots a_l:ca}
%				- \delta^a_c b_{a_1\cdots a_l:ba}) \\
%		& \quad {}
%			- \lambda (g_{b\{a\}} b_{\{a_1\cdots a_l\}:c}{}^a
%				- g_{c\{a\}} b_{\{a_1\cdots a_l\}:b}{}^a) \\
%		&= -(l(n-l)+2) \lambda b_{a_1\cdots a_l:bc} \\
%		& \quad {}
%			- (-)^{l-1} l \lambda
%				(g_{b[a_1} b_{a_2\cdots a_l] a:c}{}^a
%				-g_{c[a_1} b_{a_2\cdots a_l] a:b}{}^a) \\
%		&=  -(l(n-l)+2) \lambda b_{a_1\cdots a_l:bc} \\
%		& \quad {}
%			+ (-)^l l \lambda
%				(g_{b[a_1} \tr[b]_{a_2\cdots a_l]:c}
%				-g_{c[a_1} \tr[b]_{a_2\cdots a_l]:b}) \\
		&= -(l(n-l)+2) \lambda b_{a_1\cdots a_l:bc}
			+ (-)^l \lambda (g\odot \tr[b])_{a_1\cdots a_l:bc} , \\
\MoveEqLeft
	(l+1)(\bar{R}\cdot b)^a{}_{[b~a_1\cdots a_l]:ca}
	-(l+1)(\bar{R}\cdot b)^a{}_{[c~a_1\cdots a_l]:ba} \\
%		&= -(l+1)\lambda (lb_{[a_1\cdots a_l|:c|b]}
%				+ g_{c[b} b_{a_1\cdots a_l]:a}{}^a
%				- b_{[a_1\cdots a_l:b]c}) \\
%		& \quad {}
%			+(l+1)\lambda (lb_{[a_1\cdots a_l|:b|c]}
%				+ g_{b[c} b_{a_1\cdots a_l]:a}{}^a
%				- b_{[a_1\cdots a_l:c]b}) \\
		&= 0 ;
\end{align*}

\begin{align}
%\notag
%\MoveEqLeft
	\delta_L \circ \d_L[b]_{a_1\cdots a_l:bc}
\notag
		&= \square b_{a_1\cdots a_l:bc}
			- \d_L\circ \delta_L[b]_{a_1\cdots a_l:bc}
			+ l^{-1} \d_R\circ \delta_R[b]_{a_1\cdots a_l:bc} \\
\label{eq:dlDl}
		& \quad {}
			- (l(n-l)+2) \lambda b_{a_1\cdots a_l:bc}
			+ (-)^l \lambda (g\odot \tr[b])_{a_1\cdots a_l:bc} .
\end{align}

The main identity that will be useful in Section~\ref{app:yt-calabi} is
the following:
\begin{multline} \label{eq:dlDl-Dldl}
	(\delta_L\circ \d_L + \d_L \circ \delta_L)[b]_{a_1\cdots a_l:bc}
		= \square b_{a_1\cdots a_l:bc}
			+ l^{-1} \d_R\circ \delta_R[b]_{a_1\cdots a_l:bc} \\
			- (l(n-l)+2) \lambda b_{a_1\cdots a_l:bc}
			+ (-)^l \lambda (g\odot \tr[b])_{a_1\cdots a_l:bc} .
\end{multline}

The composition $\delta_L\circ \d_L$ has a special form when the length
of the left index column of the output is $l=2$:
\begin{align}
%\notag
	\nabla_{\{a\}} r_{\{a_1a_2\}:bc}
%		&= \nabla_{a_1} r_{a a_2:bc} + \nabla_{a_2} r_{a_1 a:bc} \\
%\notag
%		&= -\nabla_{a_1} r_{\{bc\}:\{a_2\}a}
%			+ \nabla_{a_2} r_{\{bc\}:\{a_1\}a} \\
%\notag
%		&= -\nabla_{a_1} r_{a_2 c:ba}
%			- \nabla_{a_1} r_{b a_2:ca}
%			+ \nabla_{a_2} r_{a_1 c:ba}
%			+ \nabla_{a_2} r_{b a_1:ca} \\
\label{eq:id-y22}
		&= -(\nabla_{\{c\}} r_{\{a_1a_2\}:ba} + \nabla_{\{b\}} r_{\{a_1a_2\}:ca}) ;
\end{align}

\begin{align*}
	(\bar{R}\cdot r)^a{}_{b~a_1a_2:ca}
%		&= \lambda
%				(\delta^a_{\{b\}} r_{\{a_1a_2\}:ca}
%				-g_{b\{a\}} r_{\{a_1a_2\}:c}{}^a) \\
%		& \quad {}
%			+ \lambda (\delta^a_c r_{a_1a_2:ba} - g_{bc} r_{a_1a_2:a}{}^a) \\
%		& \quad {}
%			- \lambda (\delta^a_a r_{a_1a_2:bc} - \delta^a_b r_{a_1a_2:ac}) \\
%		&= \lambda (-r_{\{a_1a_2\}:\{b\}c}
%				+ g_{b a_1} \tr[r]_{a_2:c} - g_{b a_2} \tr[r]_{a_1:c}) \\
%		& \quad {}
%			- (n-2)\lambda r_{a_1a_2:bc} \\
		&= -(n-1)\lambda r_{a_1a_2:bc}
				+ \lambda(g_{b a_1} \tr[r]_{a_2:c} - g_{b a_2} \tr[r]_{a_1:c}) , \\
\MoveEqLeft[7]
	(\bar{R}\cdot r)^a{}_{b~a_1a_2:ca}
	-(\bar{R}\cdot r)^a{}_{c~a_1a_2:ba} \\
		&= -2(n-1) \lambda r_{a_1a_2:bc} + \lambda (g\odot \tr[r])_{a_1a_2:bc} , \\
\MoveEqLeft[7]
	(\bar{R}\cdot r)^a{}_{[b~a_1a_2]:ca}
	-(\bar{R}\cdot r)^a{}_{[c~a_1a_2]:ba} \\
		&= 0 ;
\end{align*}

\begin{align}
%\notag
	\delta_L\circ\d_L[r]_{a_1a_2:bc}
%		&= 3\nabla^a \nabla_{[a} r_{a_1a_2]:bc}
%				~~ \eqref{eq:id-asym} \\
%\notag
%		& \quad {}
%			+ \frac{1}{2} (3\nabla^a \nabla_{[b} r_{a_1a_2]:ca}
%				- 3\nabla^a \nabla_{[c} r_{a_1a_2]:ba})
%				~~ \eqref{eq:id-asym} \\
%\notag
%		&= \nabla^a \nabla_a r_{a_1a_2:bc}
%			- \nabla^a \nabla_{\{a\}} r_{\{a_1a_2\}:bc}
%			~~ \eqref{eq:id-y22} \\
%\notag
%		& \quad {}
%			+ \frac{1}{2} (\nabla^a \nabla_b r_{a_1a_2:ca}
%				- \nabla^a \nabla_c r_{a_1a_2:ba}) \\
%\notag
%		& \quad {}
%			- \frac{1}{2} (\nabla^a \nabla_{\{b\}} r_{\{a_1a_2\}:ca}
%				- \nabla^a \nabla_{\{c\}} r_{\{a_1a_2\}:ba}) \\
%\notag
%		&= \square r_{a_1a_2:bc} \\
%\notag
%		& \quad {}
%			+ \frac{1}{2} (\nabla^a \nabla_b r_{a_1a_2:ca}
%				- \nabla^a \nabla_c r_{a_1a_2:ba})
%				~~ \eqref{eq:ddRb2} \\
%\notag
%		& \quad {}
%			+ \frac{1}{2} (\nabla^a \nabla_{\{b\}} r_{\{a_1a_2\}:ca}
%				- \nabla^a \nabla_{\{c\}} r_{\{a_1a_2\}:ba}) \\
%\notag
%		&= \square r_{a_1a_2:bc}
%			+ \frac{1}{2}\d_R\circ \delta_R[r]_{a_1a_2:bc} \\
%\notag
%		& \quad {}
%			+ \frac{1}{2}
%				((\bar{R}\cdot r)^a{}_{b~a_1a_2:ca}
%				-(\bar{R}\cdot r)^a{}_{c~a_1a_2:ba})
%				~~ \eqref{eq:id-asym} \\
%\notag
%		& \quad {}
%			+ \frac{1}{2}
%				((\bar{R}\cdot r)^a{}_{\{b\}~\{a_1a_2\}:ca}
%				-(\bar{R}\cdot r)^a{}_{\{c\}~\{a_1a_2\}:ba}) \\
%\notag
%		&= \square r_{a_1a_2:bc}
%			+ \frac{1}{2}\d_R\circ \delta_R[r]_{a_1a_2:bc} \\
%\notag
%		& \quad {}
%			+ ((\bar{R}\cdot r)^a{}_{b~a_1a_2:ca}
%				-(\bar{R}\cdot r)^a{}_{c~a_1a_2:ba}) \\
%\notag
%		& \quad {}
%			- \frac{3}{2}
%				((\bar{R}\cdot r)^a{}_{[b~a_1a_2]:ca}
%				-(\bar{R}\cdot r)^a{}_{[c~a_1a_2]:ba}) \\
\notag
		&= \square r_{a_1a_2:bc}
			+ \frac{1}{2}\d_R\circ \delta_R[r]_{a_1a_2:bc} \\
\label{eq:D2d2}
		& \quad {}
			-2(n-1)\lambda r_{a_1a_2:bc} + \lambda (g\odot \tr[r])_{a_1a_2:bc} .
\end{align}

Next, we show the relation between the compositions $\tr{}\circ\d_L$ and
$\d_L\circ\tr$:
\begin{align}
%\notag
	\tr{} \circ \d_L[b]_{a_1\cdots a_l:b}
%		&= (l+1)\nabla_{[a_1} b_{a_2\cdots a_l d]:b}{}^d \\
%\notag
%		&= \nabla_{a_1} b_{a_2\cdots a_l d:b}{}^d
%			- \nabla_{\{a_1\}} b_{\{a_2\cdots a_l d\}:b}{}^d \\
%\notag
%		&= \nabla_{a_1} b_{a_2\cdots a_l d:b}{}^d
%			- \nabla_{\{a_1\}} b_{\{a_2\cdots a_l\} d:b}{}^d
%			- (-)^{l-1}\nabla_{d} b_{a_1\cdots a_l:b}{}^d \\
%\notag
%		&= l\nabla_{[a_1} b_{a_2\cdots a_l]d:b}{}^d
%			+ (-)^l \nabla_{d} b_{a_1\cdots a_l:b}{}^d \\
\label{eq:trdl}
		&= \d_L\circ \tr[b]_{a_1\cdots a_l:b}
			+ (-)^l \delta_R[b]_{a_1\cdots a_l:b} .
\end{align}

Next, we show the relations between the compositions $\d_L\circ\d_R$,
$\d_R\circ\d_L$ and the operator $(\nabla\nabla\odot-)$, along with some
auxiliary identities involving the curvature. Note that below we make
use of the notation $[\cdots]^1_l$ which denotes idempotent
antisymmetrization of the indices $a_1 \cdots a_l$, as if they were
given in that position and ignoring any other indices appearing within
the same brackets.
\begin{align}
\label{eq:ddRb1}
	2\nabla_{[a}\nabla_{b]} b_{a_1\cdots a_l:c}
		&= (\bar{R}\cdot b)_{ab~a_1\cdots a_l:c} , \\
\notag
	(\bar{R}\cdot b)_{ab~a_1\cdots a_l:c}
		&= \bar{R}_{ab:\{d\}}{}^d b_{\{a_1\cdots a_l\}:c}
			+ \bar{R}_{ab:c}{}^d b_{a_1\cdots a_l:d} \\
%\notag
%		&= \lambda (g_{a\{d\}} \delta^d_b - \delta_a^d g_{b\{d\}})
%				b_{\{a_1\cdots a_l\}:c}
%			+ \lambda (g_{ac} \delta_b^d - \delta_a^d g_{bc}) b_{a_1\cdots a_l:d} \\
\notag
		&=  \lambda (g_{a\{b\}} - g_{b\{a\}}) b_{\{a_1\cdots a_l\}:c} \\
\label{eq:Rb1}
		& \quad {}
			+ \lambda (g_{ac} b_{a_1\cdots a_l:b} - g_{bc} b_{a_1\cdots a_l:a}) ;
\end{align}

\begin{align*}
\MoveEqLeft
	l(\bar{R}\cdot b)_{b[a_1~a_2\cdots a_l]:c} \\
%		&= l\lambda (g_{b[\{a_1\}} b_{\{a_2\cdots a_l\}]:c}
%				- g_{[\{b\}a_1} b_{\{a_2\cdots a_l\}]^1_l:c}) \\
%		& \quad {}
%			+ l\lambda (g_{bc} b_{[a_2\cdots a_l:a_1]}
%				- g_{c[a_1} b_{a_2\cdots a_l]:b}) \\
%		&= l\lambda (g_{b[a_1} b_{a_2\cdots a_l]:c}
%				- l g_{b[[a_1} b_{a_2\cdots a_l]]:c})
%			- l\lambda g_{c[a_1} b_{a_2\cdots a_l]:b} \\
%		&= -l(l-1)\lambda g_{b[a_1} b_{a_2\cdots a_l]:c}
%			- l\lambda g_{c[a_1} b_{a_2\cdots a_l]:b} \\
		&= -l^2\lambda g_{b[a_1} b_{a_2\cdots a_l]:c}
			+ \lambda (g\odot b)_{a_1\cdots a_l:bc} , \\
\MoveEqLeft
	2l(\bar{R}\cdot b)_{[b|[a_1~a_2\cdots a_l]:|c]} \\
%		&= -l\lambda (g\odot b)_{a_1\cdots a_l:bc}
%			+ 2\lambda (g\odot b)_{a_1\cdots a_l:bc} \\
		&= -(l-2) \lambda (g\odot b)_{a_1\cdots a_l:bc} , \\
\MoveEqLeft
	l(\bar{R}\cdot b)_{\{b\}\{[a_1~a_2\cdots a_l]\}:c} \\
%		&= -l(l+1)(\bar{R}\cdot b)_{[b[a_1~a_2\cdots a_l]]:c}
%			+ l(\bar{R}\cdot b)_{b[a_1~a_2\cdots a_l]:c} \\
%		&= (l+1)l^2\lambda g_{[b[a_1} b_{a_2\cdots a_l]]:c}
%			-(l+1)\lambda (g\odot b)_{[a_1\cdots a_l:b]c} \\
%		& \quad {}
%			- l^2\lambda g_{b[a_1} b_{a_2\cdots a_l]:c}
%			+ \lambda (g\odot b)_{a_1\cdots a_l:bc} \\
		&= - l^2\lambda g_{b[a_1} b_{a_2\cdots a_l]:c}
				+ \lambda (g\odot b)_{a_1\cdots a_l:bc} , \\
\MoveEqLeft
	2l(\bar{R}\cdot b)_{[\{b\}|\{[a_1~a_2\cdots a_l]\}:|c]} \\
%		&= -l\lambda (g\odot b)_{a_1\cdots a_l:bc}
%				+ 2\lambda (g\odot b)_{a_1\cdots a_l:bc} \\
		&= -(l-2)\lambda (g\odot b)_{a_1\cdots a_l:bc} , \\
\MoveEqLeft
	Q_{b a_1~a_2\cdots a_l:c}
		= (\bar{R}\cdot b)_{a_1\{b\}~\{a_2\cdots a_l\}:c} \\
%		&= -(\bar{R}\cdot b)_{b a_1~a_2\cdots a_l:c}
%			- l(\bar{R}\cdot b)_{a_1[b~a_2\cdots a_l]:c} \\
		&= -(\bar{R}\cdot b)_{b a_1~a_2\cdots a_l:c}
%		\\
%		& \quad {}
			+ l^2\lambda g_{a_1[b} b_{a_2\cdots a_l]:c}
			- \lambda (g\odot b)_{b a_2\cdots a_l:a_1 c} , \\
\MoveEqLeft
	l(\bar{R}\cdot b)_{[a_1\{b\}~\{a_2\cdots a_l\}]^1_l:c}
	-l(\bar{R}\cdot b)_{[a_1\{c\}~\{a_2\cdots a_l\}]^1_l:b} \\
		&= lQ_{b[a_1~a_2\cdots a_l]:c} - lQ_{c[a_1~a_2\cdots a_l]:b} \\
%		&= -l(\bar{R}\cdot b)_{b[a_1~a_2\cdots a_l]:c}
%			+l(\bar{R}\cdot b)_{c[a_1~a_2\cdots a_l]:b} \\
%		& \quad {}
%			-l\lambda (g\odot b)_{b[a_2\cdots a_l:a_1]c}
%			+l\lambda (g\odot b)_{c[a_2\cdots a_l:a_1]b} \\
%		& \quad {}
%			+ l^2\lambda g_{b [a_1} b_{a_2\cdots a_l]:c}
%			- l^2\lambda g_{c [a_1} b_{a_2\cdots a_l]:b} \\
%		& \quad {}
%			- l^2\lambda g_{[\{b\}a_1} b_{\{a_2\cdots a_l\}]^1_l:c}
%			+ l^2\lambda g_{[\{c\}a_1} b_{\{a_2\cdots a_l\}]^1_l:b} \\
%		&= l^2\lambda g_{b[a_1} b_{a_2\cdots a_l]:c}
%			- l^2\lambda g_{c[a_1} b_{a_2\cdots a_l]:b} \\
%		& \quad {}
%			- \lambda(g\odot b)_{a_1\cdots a_l:bc}
%			+ \lambda(g\odot b)_{a_1\cdots a_l:cb} \\
%		& \quad {}
%			-(l+1)\lambda (g\odot b)_{[b a_2\cdots a_l:a_1]c}
%			+(l+1)\lambda (g\odot b)_{[c a_2\cdots a_l:a_1]b} \\
%		& \quad {}
%			- \lambda (g\odot b)_{a_1\cdots a_l:bc}
%			+ \lambda (g\odot b)_{a_1\cdots a_l:cb} \\
%		& \quad {}
%			+ l\lambda (g\odot b)_{a_1\cdots a_l:bc} \\
		&= 2(l-2)\lambda (g\odot b)_{a_1\cdots a_l:bc} ;
\end{align*}

\begin{align}
%\notag
%\MoveEqLeft
	\d_R\circ\d_L[b]_{a_1\cdots a_l:bc}
% 		= 2l \nabla_{[b|}\nabla_{[a_1} b_{a_2\cdots a_l]:|c]}
% 			+ 2l(l-1)^{-1} \nabla_{[\{b\}|} \nabla_{\{[a_1} b_{a_2\cdots a_l]\}:|c]} \\
%\notag
%		&= l(\nabla\nabla_{b [a_1} b_{a_2\cdots a_l]:c}
%				-\nabla\nabla_{c [a_1} b_{a_2\cdots a_l]:b}) \\
%\notag
%		& \quad {}
%			- (l+1)(l-1)^{-1}l
%				(\nabla\nabla_{[b[a_1} b_{a_2\cdots a_l]]:c}
%				-\nabla\nabla_{[c[a_1} b_{a_2\cdots a_l]]:b})
%			~~ \eqref{eq:id-asym} \\
%\notag
%		& \quad {}
%			+ (l-1)^{-1}l
%				(\nabla\nabla_{b[a_1} b_{a_2\cdots a_l]:c}
%				-\nabla\nabla_{c[a_1} b_{a_2\cdots a_l]:b}) \\
%\notag
%		& \quad {}
%			+ \frac{1}{2} 2l (\bar{R}\cdot b)_{[b|[a_1~a_2\cdots a_l]:|c]} \\
%\notag
%		& \quad {}
%			+ \frac{(l-1)^{-1}}{2} 2l
%				\bar{R}\cdot b)_{[\{b\}|\{[a_1~a_2\cdots a_l]\}:|c]} \\
%\notag
%		&= (l-1)^{-1} l (\nabla\nabla\odot b)_{a_1\cdots a_l:bc} \\
%\notag
%		& \quad {}
%			- \frac{(l-2)}{2} \lambda (g\odot b)_{a_1\cdots a_l:bc}
%			- \frac{(l-2)}{2(l-1)} \lambda (g\odot b)_{a_1\cdots a_l:bc} \\
\notag
		&= (l-1)^{-1} l (\nabla\nabla\odot b)_{a_1\cdots a_l:bc} \\
\label{eq:drdl}
		& \qquad {}
			- \frac{l(l-2)}{2(l-1)} \lambda (g\odot b)_{a_1\cdots a_l:bc} ,
\end{align}

\begin{align}
%\notag
%\MoveEqLeft
	\d_L\circ\d_R[b]_{a_1\cdots a_l:bc}
%		\\
%\notag
%		&= l(\nabla_{[a_1|}\nabla_{b} b_{|a_2\cdots a_l]:c}
%				-\nabla_{[a_1|}\nabla_{c} b_{|a_2\cdots a_l]:b}) \\
%\notag
%		& \quad {}
%			+ (l-2)^{-1} l (\nabla_{[a_1|} \nabla_{\{b\}} b_{|\{a_2\cdots a_l\}]:c}
%				- \nabla_{[a_1|} \nabla_{\{c\}} b_{|\{a_2\cdots a_l\}]:b}) \\
%\notag
%		&= l(\nabla\nabla_{b[a_1} b_{a_2\cdots a_l]:c}
%				-\nabla\nabla_{c[a_1} b_{a_2\cdots a_l]:b}) \\
%\notag
%		& \quad {}
%			+ (l-2)^{-1} l
%					(\nabla\nabla_{[a_1\{b\}} b_{\{a_2\cdots a_l\}]^1_l:c}
%					-\nabla\nabla_{[a_1\{c\}} b_{\{a_2\cdots a_l\}]^1_l:b}) \\
%\notag
%		& \quad {}
%			- \frac{1}{2} 2l (\bar{R}\cdot b)_{[b|[a_1~a_2\cdots a_l]:|c]} \\
%\notag
%		& \quad {}
%			+ \frac{(l-2)^{-1}}{2} l
%				((\bar{R}\cdot b)_{[a_1\{b\}~\{a_2\cdots a_l\}]^1_l:c}
%				-(\bar{R}\cdot b)_{[a_1\{b\}~\{a_2\cdots a_l\}]^1_l:c}) \\
%\notag
%		&= (\nabla\nabla \odot b)_{a_1\cdots a_l:bc} \\
%\notag
%		& \quad {}
%			+ \frac{(l-2)}{2} \lambda (g\odot b)_{a_1\cdots a_l:bc}
%			+ \lambda (g\odot b)_{a_1\cdots a_l:bc} \\
\label{eq:dldr}
		&= (\nabla\nabla \odot b)_{a_1\cdots a_l:bc}
			+ \frac{l}{2} \lambda (g\odot b)_{a_1\cdots a_l:bc} .
\end{align}

The main identity that will be useful in Section~\ref{app:yt-calabi} is
the following:
\begin{align}
%\MoveEqLeft
%\notag
	l^{-1} \d_R\circ \d_L - (l-1)^{-1} \d_L\circ \d_R
%	\\
%		&= (l-1)^{-1} (\nabla\nabla \odot b)_{a_1\cdots a_l:bc}
%			- \frac{(l-2)}{2(l-1)} \lambda (g\odot b)_{a_1\cdots a_l:bc} \\
%\notag
%		& \quad {}
%			- (l-1)^{-1} (\nabla\nabla \odot b)_{a_1\cdots a_l:bc}
%			- \frac{l}{2(l-1)} \lambda (g\odot b)_{a_1\cdots a_l:bc} \\
\label{eq:drdl-dldr}
		&= -\lambda (g\odot b)_{a_1\cdots a_l:bc} .
\end{align}

\subsection{Calabi complex and its homotopy formulas}\label{app:yt-calabi}
Below, we use the special differential operators introduced earlier in
Section~\ref{app:yt-ops} to explicitly define the differential operators
$B_l$, $E_l$ and $P_l$ that make up the Calabi complex and its homotopy
formulas, as discussed in more detail in Section~\ref{sec:calabi-ops}:
\begin{align*}
	B_1[v]_{a:b}
		&= \nabla_a v_b + \nabla_b v_a , \\
	B_2[h]_{a_1a_2:bc}
		&= (\nabla\nabla\odot h)_{a_1a_2:bc}
			+ \lambda (g\odot h)_{a_1a_2:bc} , \\
	B_l[b]_{a_1\cdots a_l:bc}
		&= \d_L[b]_{a_1\cdots a_l:bc} \quad (l\ge 3) , \\
	E_1[h]_a
		&= \nabla^b h_{a:b} - \frac{1}{2} \nabla_a \tr[h] , \\
	E_2[b]_{a:b}
		&= \tr[b]_{a:b} , \\
	E_{l+1}[b]_{a_1\cdots a_l:bc}
		&= \left(\delta_L - (-)^l l^{-1} \d_R\circ\tr\right)[b]_{a_1\cdots a_l:bc}
			\quad (l\ge 2) .
\end{align*}
Explicit formulas for $B_l$ and $E_l$ with low $l$ have been given in
Section~\ref{sec:calabi-ops}.

Further, we make use of the identities given in
Section~\ref{app:yt-comp} to show that these operators satisfy the
required identities, namely $B_{l+1}\circ B_l = 0$. The identities
$B_2\circ B_1 = 0$ and $B_3\circ B_2 = 0$ have already been shown to
follow in Section~\ref{sec:calabi-ops} from the usual transformation
properties of the Riemann curvature tensor under diffeomorphisms and
from its Bianchi identities. The identities $B_{l+1} \circ B_l = 0$ for
$l>2$ then follow directly from the composition identity $\d_L\circ\d_L
= 0$ in Equation~\eqref{eq:dldl}.

Again, appealing to the identities of Section~\ref{app:yt-comp}, we give
the homotopy formulas $P_l = E_{l+1}\circ B_{l+1} + B_l\circ E_l$ for
$l\le 2$:
\begin{align*}
	E_1\circ B_1[v]_a
		&= \nabla^b (\nabla_a v_b + \nabla_b v_a)
			- \frac{1}{2} \nabla_a (2\nabla^b v_b) \\
%		&= \nabla^b\nabla_b v_a + \nabla^b \nabla_a v_b - \nabla_a \nabla^b v_b \\
%		&= \square v_a + \bar{R}^b{}_{ab}{}^c v_c \\
%		&= \square v_a + \lambda (\delta^b_b \delta^c_a - g^{bc} g_{ab}) v_c \\
	P_0 &= \square v_a + \lambda (n-1) v_a ,
\end{align*}

\begin{align*}
\MoveEqLeft
	(E_2 \circ B_2 + B_1\circ E_1)[h]_{a:b} \\
		&= (\nabla\nabla \odot h)_{ac:b}{}^c
			+ \lambda (g\odot h)_{ac:b}{}^c \\
		& \quad {}
			+ \nabla_a\left(\nabla^c h_{b:c} - \frac{1}{2} \nabla_b \tr[h]\right)
			+ \nabla_b\left(\nabla^c h_{a:c} - \frac{1}{2} \nabla_a \tr[h]\right) \\
%		&= \square h_{ab} + \nabla\nabla_{ab} \tr[h]
%			- \nabla\nabla_a{}^c h_{cb} - \nabla\nabla_b{}^c h_{ca} \\
%		& \quad {}
%			+ (n-2)\lambda h_{ab} + \lambda g_{ab} \tr[h] \\
%		& \quad {}
%			+ \nabla\nabla_a{}^c h_{cb} + \nabla\nabla_b{}^c h_{ca}
%			- \nabla\nabla_{ab} \tr[h] \\
%		& \quad {}
%			+ \frac{1}{2}(\bar{R}\cdot h)_a{}^c{}_{~c:b}
%			+ \frac{1}{2}(\bar{R}\cdot h)_b{}^c{}_{~c:a} \\
%		&= \square h_{ab}
%			+ (n-2)\lambda h_{ab} + \lambda g_{ab} \tr[h] \\
%		& \quad {}
%			+ \frac{1}{2}\lambda
%				(g^{cd} g_{a\{c\}} h_{\{d\}:b}
%				-\delta^c_{\{a\}} h_{\{c\}:b})
%			+ \frac{1}{2}\lambda
%				(g_{ab} h_{c}{}^c - \delta^c_b h_{c:a}) \\
%		& \quad {}
%			+ \frac{1}{2}\lambda
%				(g^{cd} g_{b\{c\}} h_{\{d\}:a}
%				-\delta^c_{\{b\}} h_{\{c\}:a})
%			+ \frac{1}{2}\lambda
%				(g_{ba} h_{c}{}^c - \delta^c_a h_{c:b}) \\
%		&= \square h_{ab}
%			+ (n-2)\lambda h_{ab} + \lambda g_{ab} \tr[h] \\
%		& \quad {}
%			+ \lambda (h_{ab} - n h_{ab} + g_{ab}\tr[h]-h_{ab}) \\
	P_1 &= \square h_{ab}
			- 2\lambda h_{ab} + 2\lambda g_{ab} \tr[h] ,
\end{align*}

\begin{align*}
\MoveEqLeft
	(E_3\circ B_3 + B_2\circ E_2)[r]_{a_1a_2:bc} \\
		&= \delta_L\circ\d_L[r]_{a_1a_2:bc}
			- \frac{1}{2}\d_R\circ\tr{}\circ \d_L[r]_{a_1a_2:bc}
			%~~ \eqref{eq:D2d2}, \eqref{eq:trdl}
			\\
		& \quad {}
			+ (\nabla\nabla\odot \tr[r])_{a_1a_2:bc}
			+ \lambda (g\odot \tr[r])_{a_1a_2:bc} \\
%		&= \square r_{a_1a_2:bc}
%			+ \frac{1}{2}\d_R\circ \delta_R[r]_{a_1a_2:bc} \\
%		& \quad {}
%			-2(n-1)\lambda r_{a_1a_2:bc} + \lambda (g\odot \tr[r])_{a_1a_2:bc} \\
%		& \quad {}
%			- \frac{1}{2}\d_R\circ \d_L\circ\tr[r]_{a_1a_2:bc}
%			- \frac{1}{2}\d_R\circ\delta_R[r]_{a_1a_2:bc}
%				~~ \eqref{eq:drdl} \\
%		& \quad {}
%			+ (\nabla\nabla\odot \tr[r])_{a_1a_2:bc}
%			+ \lambda (g\odot \tr[r])_{a_1a_2:bc} \\
%		&= \square r_{a_1a_2:bc}
%			-2(n-1)\lambda r_{a_1a_2:bc} + 2\lambda (g\odot \tr[r])_{a_1a_2:bc} \\
%		& \quad {}
%			- (\nabla\nabla\odot \tr[r])_{a_1a_2:bc}
%			+ (\nabla\nabla\odot \tr[r])_{a_1a_2:bc} \\
	P_2 &= \square r_{a_1a_2:bc}
			-2(n-1)\lambda r_{a_1a_2:bc} + 2\lambda (g\odot \tr[r])_{a_1a_2:bc} .
\end{align*}

Finally, the same set of identities also implies the following formulas
for $P_l$ with $l>2$:
\begin{align*}
\MoveEqLeft
	(E_{l+1}\circ B_{l+1} + B_l\circ E_l)[b]_{a_1\cdots a_l:bc} \\
		&= (\delta_L\circ \d_L + \d_L \circ \delta_L)[b]_{a_1\cdots a_1:bc} \\
		& \quad {}
			- (-)^l
				\left(l^{-1} \d_R\circ \tr{}\circ \d_L
				- (l-1)^{-1}\d_L\circ\d_R\circ\tr{}\right)[b]_{a_1\cdots a_l:bc} \\
%		&= (\delta_L\circ \d_L + \d_L \circ \delta_L)[b]_{a_1\cdots a_1:bc} \\
%		& \quad {}
%			- (-)^l
%				\left(l^{-1} \d_R\circ \d_L
%				- (l-1)^{-1}\d_L\circ\d_R\right)[\tr[b]]_{a_1\cdots a_l:bc} \\
%		& \quad {}
%			- l^{-1} \d_R\circ \delta_R[b]_{a_1\cdots a_l:bc} \\
%		&= \square b_{a_1\cdots a_l:bc}
%			+ l^{-1} \d_R\circ \delta_R[b]_{a_1\cdots a_l:bc} \\
%		& \quad {}
%			- (l(n-l)+2) \lambda b_{a_1\cdots a_l:bc}
%			- (-)^l \lambda (g\odot \tr[b])_{a_1\cdots a_l:bc} \\
%		& \quad {}
%			+ (-)^l \lambda (g\odot b)_{a_1\cdots a_l:bc}
%			- l^{-1} \d_R\circ \delta_R[b]_{a_1\cdots a_l:bc} \\
	P_l &= \square b_{a_1\cdots a_l:bc}
			- (l(n-l)+2)\lambda b_{a_1\cdots a_l:bc}
			+ (-)^l 2\lambda (g\odot \tr[b])_{a_1\cdots a_l:bc} .
\end{align*}
Explicit formulas for $P_l$ with low $l$ have also been given in
Section~\ref{sec:calabi-ops}. Recall that, as in
Equation~\eqref{eq:Rb-def}, we have used the notation $\lambda =
\frac{k}{n(n-1)}$.

\subsection{An adjoint operator}\label{app:yt-adj}
Here we derive Equation~\eqref{eq:dlDl-adj}, which according to the
general formula~\eqref{eq:gen-adj} implies that $-n^{-1} \delta_L$ is
the formal adjoint of $\d_L$ when acting on tensors of Young type
$c_{a_2\cdots a_n:bc}$.

\begin{align}
\notag
\MoveEqLeft
	\nabla_a (c^{a a_2\cdots a_n : bc} b_{a_2\cdots a_n : bc}) \\
\notag
		&= c^{a a_2\cdots a_n : bc} \nabla_a b_{a_2\cdots a_n : bc}
			+ (\nabla_a c^{a a_2\cdots a_n : bc}) b_{a_2\cdots a_n : bc} \\
\notag
		&= c^{a a_2\cdots a_n : bc} \nabla_{[a} b_{a_2\cdots a_n] : bc}
			+ \delta_L[c]^{a_2\cdots a_n:bc} b_{a_2\cdots a_n:bc} \\
\notag
		& \quad {}
			- \frac{1}{n-1} (\nabla_a c^{b a_2\cdots a_n:ca}
				-\nabla_a c^{c a_2\cdots a_n:ba}) b_{a_2\cdots a_n:bc} \\
\notag
		&= \frac{1}{n} c^{a a_2\cdots a_n : bc} \d_L[b]_{a a_2\cdots a_n : bc}
			+ \delta_L[c]^{a_2\cdots a_n:bc} b_{a_2\cdots a_n:bc} \\
\notag
		& \quad {}
			- \frac{1}{n-1}
				(b_{[a_2\cdots a_n:b]c}\nabla_a c^{b a_2\cdots a_n:ca}
				+b_{[a_2\cdots a_n:c]b}\nabla_a c^{c a_2\cdots a_n:ba})
			~~ \eqref{eq:id-asym} \\
\label{eq:dlDl-adj-der}
		&= \frac{1}{n} c^{a a_2\cdots a_n : bc} \d_L[b]_{a a_2\cdots a_n : bc}
			+ \delta_L[c]^{a_2\cdots a_n:bc} b_{a_2\cdots a_n:bc} .
\end{align}
We have simply used the definitions of the $\d_L$ and $\delta_L$
differential operators as well as the fact that the contraction of two
tensors, one of which being totally anti-symmetric in a subset of
indices, allows the insertion of an anti-symmetrization over the
corresponding indices of the second tensor. Finally, some of the
anti-symmetrizations annihilated the corresponding tensors, due to their
intercolumn identities and the application of the
identity~\eqref{eq:id-asym}.

\section{Homological algebra}\label{app:homalg}
Below we introduce some basic notions from homological algebra. A
standard text on the subject is~\cite{weibel}, where more details can be
found along with complete proofs.

Let $A_i$, also denoted $A_\bullet$, be a sequence of vector spaces
(real vector spaces, for our purposes) with linear maps $A_i \to
A_{i+1}$ between them. If each successive pair of maps $A_{i-1} \to A_i
\to A_{i+1}$ composes to zero, this sequence is called a \emph{complex
(of vector spaces)} or a \emph{cochain complex}, with an element $a\in
A_i$ being referred to as a \emph{cochain (of degree $i$)}, and the maps
$A_i \to A_{i+1}$ referred to as \emph{cochain differentials}. Any
complex gives rise to cohomologies
\begin{equation}
	H^i(A_\bullet) = \ker (A_i \to A_{i+1}) / \im (A_{i-1} \to A_i).
\end{equation}
If all the cohomologies vanish, $H^i(A_\bullet) = 0$ or the image of
each map is equal to the kernel of the subsequent map, the complex is
called \emph{exact} or an \emph{exact sequence}. Given two complexes
$A_\bullet$ and $B_\bullet$, the vertical maps in the diagram
\begin{equation}
\begin{tikzcd}
	\cdots \ar{r} &
	A_i \ar{r} \ar{d} &
	A_{i+1} \ar{r} \ar{d} &
		\cdots \\
	\cdots \ar{r} &
	B_i \ar{r} &
	B_{i+1} \ar{r} &
	\cdots
\end{tikzcd}
\end{equation}
are called \emph{cochain maps} provided they make the diagram commute.
Furthermore, the diagonal maps in a diagram like
\begin{equation}
\begin{tikzcd}
	\cdots \ar{r}{\d} &
	A_i \ar{r}{\d} \ar{d}{h} \ar[dashed]{dl}[swap]{\delta} &
	A_{i+1} \ar{r}{\d} \ar{d}{h} \ar[dashed]{dl}[swap]{\delta} &
		\cdots \ar[dashed]{dl}[swap]{\delta} \\
	\cdots \ar{r}{\d} &
	B_i \ar{r}{\d} &
	B_{i+1} \ar{r}{\d} &
	\cdots
\end{tikzcd}
\end{equation}
are called \emph{cochain homotopies}. The homotopy maps induce vertical
cochain maps by the formula $h = \d\delta + \delta\d$. It is a basic
fact that cochain maps $A_\bullet \to B_\bullet$ naturally induce maps
in cohomology $H^i(A_\bullet) \to H^i(B_\bullet)$. Of course, identity
chain maps induce identity maps in cohomology and zero chain maps
induces zero maps in cohomology. Also, two cochain maps induce the same
map in cohomology when their difference is induced by a cochain
homotopy.

A \emph{short exact sequence}
\begin{equation}
\begin{tikzcd}
	0 \ar{r} &
	A_\bullet \ar{r} &
	B_\bullet \ar{r} &
	C_\bullet \ar{r} &
	0
\end{tikzcd}
\end{equation}
between complexes $A_\bullet$, $B_\bullet$ and $C_\bullet$ consists of
cochain maps between them such that each instance of
\begin{equation}
\begin{tikzcd}
	0 \ar{r} &
	A_i \ar{r} &
	B_i \ar{r} &
	C_i \ar{r} &
	0
\end{tikzcd}
\end{equation}
is an exact sequence of vector spaces. Another basic fact of homological
algebra is that a short exact sequence of complexes induces the
following \emph{long exact sequence} in cohomology
\begin{equation}\label{eq:long-ex}
\begin{tikzcd}
	\cdots \arrow{r} &
	H^i(A_\bullet) \arrow{r} &
	H^i(B_\bullet) \arrow{r} \arrow[draw=none]{d}[name=Z,shape=coordinate]{}&
	H^i(C_\bullet)
		\arrow[rounded corners,
			to path={ -- ([xshift=2ex]\tikztostart.east)
			|- (Z) [near end]\tikztonodes
			-| ([xshift=-2ex]\tikztotarget.west)
			-- (\tikztotarget)}]{dll}{} \\
	&
	H^{i+1}(A_\bullet) \arrow{r} &
	H^{i+1}(B_\bullet) \arrow{r} &
	H^{i+1}(C_\bullet) \arrow{r} & \cdots
\end{tikzcd} ,
\end{equation}
where the maps $H^i(A_\bullet) \to H^i(B_\bullet)$ and $H^i(B_\bullet)
\to H^i(C_\bullet)$ are induced by the cochain maps from the short exact
sequence and the \emph{connecting} maps $H^i(C_\bullet) \to
H^{i+1}(A_\bullet)$ are induced by the cochain differential.

Finally, another standard result is the so-called \emph{$5$-lemma} (or a
simple variant thereof). It states that the central vertical map in the
commutative diagram
\begin{equation}
\begin{tikzcd}
	A_{-2} \ar{r} \ar{d}{\cong} &
	A_{-1} \ar{r} \ar{d}{\cong} &
	A_{0} \ar{r} \ar{d} &
	A_{1} \ar{r} \ar{d}{\cong} &
	A_{2} \ar{d}{\cong}
	\\
	B_{-2} \ar{r} &
	B_{-1} \ar{r} &
	B_{0} \ar{r} &
	B_{1} \ar{r} &
	B_{2}
\end{tikzcd}
\end{equation}
is an isomorphism, provided that the top and bottom rows are exact
sequences and all the other vertical maps are isomorphisms themselves.

\section{Jets and jet bundles}\label{app:jets}
In this appendix, we briefly introduce jet bundles, fix the relevant
notation and discuss differential operators in the context of jets. For
simplicity, we restrict ourselves to fields taking values in vector
bundles. However, the discussion could be straightforwardly generalized
to general smooth bundles. More details, as well as a coordinate
independent definition, can be found in the standard
literature~\cite{olver-lie,kms}.

Given a vector bundle $F\to M$ over a connected $n$-dimensional smooth
manifold $M$, the \emph{$k$-jet bundle} $J^kF\to M$ is a vector bundle
whose defining characteristic is that for any (possibly non-linear)
differential operator $f\colon \Secs(F)\to \Secs(F')$ of order $k$,
there exists a canonical factorization $f[u] = f\circ j^k u$ for any
section $u\colon M\to F$, where the \emph{$k$-jet prolongation}
$j^k\colon \Secs(F) \to \Secs(J^kF)$ is composed with a smooth bundle
map $f\colon J^kF\to F'$, which by a slight abuse of notation we denote
using the same symbol as the original differential operator. Composing
the differential operator $f$ with an $l$-jet prolongation canonically
defines a new differential operator $p_l f\colon J^{l+k}F \to J^lF'$
called its \emph{$l$-prolongation}, $j^l f[u] = p_l f\circ j^k u$. Given
a trivializable restriction $F_U\to U$ of $F$ to a chart $U\sso M$ with
local coordinates $(x^i)$ and fiber-adapted local coordinates
$(x^i,u^a)$, there is a corresponding adapted chart $J^kF_U\sso J^kF$
with adapted local coordinates $(x^i,u^a_I)$, where $I=i_1\cdots i_l$
runs through multi-indices of orders $|I|=l=0,\ldots,k$. In these
coordinates, the $k$-jet prolongation is given by $j^k u(x) =
(x^i,\partial_I u^a(x))$, while the $l$-prolongation is given by $p_l
f[u](x) = (x^i,\del_I f^b[u](x))$, where $f[u](x) = (x^i,f^b[u](x))$ in
fiber-adapted local coordinates $(x^i,v^b)$ on $F'$. For any $l>k$,
discarding the information about all derivatives of order $>k$ defines a
natural projection $J^lF\to J^kF$.  The projective limit $J^\oo F :=
\varprojlim_{k\to\oo} J^k F$ defines the \emph{$\oo$-jet bundle}. The
$\oo$-jet prolongation $j^\oo$ and $\oo$-prolongation $p_\oo$ are
defined in the obvious way. By composing with the natural projection
$J^\oo F\to J^kF$, the differential operator $f$ also canonically defines
the smooth bundle map $f\colon J^\oo F\to J^kF \stackrel{f}{\to} F'$,
which is again denoted by the same symbol $f$. Conversely, due to the
projective limit construction, any smooth bundle map $f\colon J^\oo F
\to F'$ can only depend on finitely many coordinates of its domain,
which means that there exists a $k\ge 0$ such that this bundle map
canonically factors as $f\colon J^\oo F\to J^k F \stackrel{f}{\to} F'$,
with the smallest such $k$ being the \emph{order} of $f$.

Given vector bundles $F\to M$, $E\to M$ and a differential operator
$e\colon \Secs(F)\to \Secs(E)$, we write down the \emph{partial
differential equation} (PDE) $e[\psi] = 0$, with $\psi \in \Secs(F)$.
Sometimes it is convenient to refer to $F\to M$ as the \emph{field
bundle} and to $E\to M$ as the \emph{equation bundle}. We will only
consider linear PDEs below, where the differential operator $e$ is
linear. We denote the
local spaces of solutions by $\S_e(U)$, where $U\sse M$ is open and $\psi
\in \Secs(F|_U)$ belongs to $\S_e(U)$ iff $e[\psi] = 0$ on $U$. The PDE
is said to be of order $k$ if it can be written as $e[\psi] = e(j^k\psi)$,
where on the right-hand side we have a (linear) bundle map $e\colon J^k
F \to E$.

In adapted coordinates $(x^i,u^a)$ on $F$, the PDE $e[\psi] = 0$ has the
form $e^I(x) \del_I \psi(x) = 0$. When the PDE is of order $k$, the
coefficients $e^I(x)$ vanish for multi-indices with $|I| > k$. The
coefficients of the highest order derivatives, $e^I(x)$ with $|I|=k$, in
fact transform as a tensor under coordinate changes and define a linear
bundle map $\sigma e \colon F\otimes S^kT^*M \to E$ called the
\emph{principal symbol} of $e$. If we fix $(x,p)\in T^*M$, then the
corresponding linear map $\sigma_{x,p} e = \sigma e(x)\cdot p^{\otimes
k} \colon F_x M \to E_x M$ can also be referred to as the value of the
principal symbol of $e$ at $(x,p)$.

The PDE $e[\psi] = 0$ is \emph{equivalent} to the PDE $e'[\psi'] = 0$,
with $e'\colon J^{k'}F' \to E'$, if they have isomorphic solution
spaces.  That is $e[\psi] = 0$ implies that $e'[f[\psi]] = 0$ and
$e'[\psi'] = 0$ implies that $e[f'[\psi']] = 0$, for some differential
operators $f$ and $f'$. In fact, it can be shown that the two PDEs are
equivalent precisely when they fit into the following diagram, where
arrows are differential operators and the bundle labels stand in for the
corresponding spaces of sections,
\begin{equation}
\begin{tikzcd}
	F \ar[swap]{r}{e} \ar[shift left]{d}{f}
		& E \ar[shift left]{d}{g} \ar[dashed,bend right,swap]{l}{q} \\
	F' \ar{r}{e'} \ar[shift left]{u}{f'}
		& E' \ar[shift left]{u}{g'} \ar[dashed,bend left]{l}{q'}
\end{tikzcd} ,
\end{equation}
where differential operators satisfy the following identities:
\begin{align}
	e' \circ f &= g \circ e , &
		f' \circ f &= \id + q \circ e , \\
  e \circ f' &= g' \circ e' , &
		f \circ f' &= \id + q' \circ e' .
\end{align}
The reason we can express equivalence in this way, at least when all the
differential operators are linear, follows from linear algebra on jets.
If we replace the operators $e$, $e'$, $f$ and $f'$ by the corresponding
jet bundle maps, prolonged to the appropriate order, it follows from
basic linear algebra that there exist jet bundle maps that ostensibly
correspond to the operators $g$, $g'$, $q$ and $q'$. It then follows
from a deeper analysis of the properties of linear
PDEs~\cite{goldschmidt-lin,bcggg,seiler-inv} that once these
differential operators are defined using prolongations of sufficiently
high order, the appropriate identities hold at all higher orders. As a
simple example, note that the equation $e[\psi] = 0$ and its
prolongation $p_k e[\psi] = 0$ are equivalent, with $f=f'=\id$.

Consider vector bundles $E, F, G \to M$ and linear differential
operators
\begin{equation}
	f\colon \Secs(G) \to \Secs(F)
	\quad \text{and} \quad
	e \colon \Secs(F) \to \Secs(E),
\end{equation}
of respective orders $k$ and $l$, such that $e\circ f = 0$.
We say that the composition of $e$ and $f$ is \emph{formally exact} if
the composition $p^{k+m} e \circ p^m f$ of jet bundle maps is exact in
the usual linear algebra sense (the image of $p^m f$ is equal to the
kernel of $p^{k+m} e$). Formal exactness is a powerful hypothesis. For
instance, it implies that certain differential operators factorise
through either $e$ or $f$~\cite{goldschmidt-lin,pommaret}. Namely, if
$g$ is any differential operator such that $g \circ f = 0$, then there
must exist another differential operator $g'$ such that $g = g' \circ
e$. Similarly, if $g$ is any differential operator such that $e \circ g
= 0$, then there must exist another differential operator $g'$ such that
$g = f\circ g'$.

\section{Deformations of flat principal bundles}\label{app:bndl-deform}
The material below requires some familiarity with the theory of
$G$-principal bundles~\cite{steenrod,kobayashi,morita,baum}. Its main
point is to show how one can reduce the computation of the degree-$1$
cohomology space of a certain locally constant sheaf on a manifold $M$
to the computation of the degree-$1$ group cohomology of the fundamental
group $\pi = \pi_1(M)$ with coefficients in a certain corresponding
representation. This reformulation is a significant simplification
because group cohomology calculations can often be reduced to finite
dimensional linear algebra and many explicit calculations of that sort
have already been performed and are available in the literature. The
connection between these sheaf and group cohomologies is established by
noticing that both of them describe equivalence classes of infinitesimal
deformations of flat principal bundles. Unfortunately, this argument is
not sufficient to establish an isomorphism between these sheaf and group
cohomologies in higher degrees, but degree-$1$ is already interesting
because it is the one relevant in the physical application we have in
mind (Section~\ref{sec:appl}).

We briefly recall some basic facts about principal
$G$-bundles~\cite{steenrod,kobayashi,morita,baum}. The total space of
the \emph{principal bundle} $P\to M$ has fibers that are right principal
homogeneous spaces of the group $G$. A \emph{right principal homogeneous
space} is defined by the possession a free, transitive action of $G$.
Thus, any principal homogeneous space is diffeomorphic to the manifold
underlying the Lie group $G$ and, if any particular point is identified
with the unit element of $G$, the action of $G$ coincides with
right-multiplication. The fiber-wise right action of $G$ on $P$ allows
us to construct so-called \emph{associated bundles}.  If $F$ is a left
$G$-space, with action $\rho \colon G\to \Aut(F)$, then we define the
corresponding associated bundle, denoted sometimes $F_\rho$ or $F_P$, as
$P\times_\rho F \cong (P \times F)/G$, where the quotient identifies the
points $(pg,f) = (p,gf)$, $p\in P$, $f\in F$, $g\in G$. In particular,
we can define the associated bundles $G_P = P \times_\Ad G$ and $\g_P = P
\times_\Ad \g$, where $\Ad$ denotes respectively the adjoint action of
the Lie group on itself and its Lie algebra, $\Ad(b)a = b a b^{-1}$ and
$\Ad(b)\alpha = b \alpha b^{-1}$, with $a,b\in G$ and $\alpha \in \g$.
When convenient and for simplicity of notation, we shall implicitly
treat Lie group and Lie algebra elements as if they were faithfully
represented as matrices.

The principal $G$-bundle $P\to M$ is called \emph{flat} when it is
endowed with a flat connection or a notion of flat parallel transport,
which are compatible with the structure group action. The details of
these notions are discussed in the next subsections. The arguments
presented therein roughly establish the following
\begin{prop}\label{prp:cohom-equiv}
Let $P\to M$ be a flat principal $G$-bundle and $\pi = \pi(M)$ be the
fundamental group of $M$. We can define the following structures
associated to it: (a) the sheaf $\F_\g$ of locally flat sections of the
associated bundle $\g_P \to M$, (b) the twisted de~Rham complex
$(\Lambda^\bullet M \otimes \g_P, D)$, and (c) the monodromy
representation $\rho \colon \pi \to G$. Then the following cohomology
groups (respectively the sheaf, twisted de~Rham and group cohomologies)
are all isomorphic, by reason of each being isomorphic to the space of
equivalence classes of infinitesimal deformations of the flat principal
$G$-bundle structure of $P\to M$:
\begin{equation}
	H^1(M,\F_\g)
	\cong H^1(\Lambda^\bullet M\otimes \g_P,D)
	\cong H^1(\pi, \Ad_\rho) .
\end{equation}
\end{prop}
We defer to the standard references~\cite{kobayashi,morita,baum} for
detailed proofs.

\subsection{Flat principle bundle cocycle}\label{sec:cocyc-def}
There are multiple ways to construct a principal $G$-bundle over a
manifold $M$. The one that will be important for us here defines also a
bit more structure than principal bundle itself, it also defines a flat
connection thereon. We shall refer to these structures as \emph{flat
principal $G$-bundles}. It is well known that this data can be specified
as follows. Let $\U = (U_i)$ be an open cover of $M$ and $(U,V)\mapsto
t_{U,V} \in G$ an assignment of a structure group element to every
ordered pair of opens $U,V\in \U$. Each $t_{U,V}$ is called a
\emph{transition map}. The transition maps define a principle $G$-bundle
with a flat connection if they satisfy the following \emph{cocycle
identities},
\begin{gather}
\label{eq:t-inv}
	t_{U,V} t_{V,U} = \id , \\
\label{eq:t-cocyc}
	t_{U,V} t_{V,W} t_{W,U} = \id .
\end{gather}
A change of trivialization is an assignment $U\mapsto a_U \in G$ for
every open $U \in \U$. The modified transition functions $t'_{U,V} = a_U
t_{U,V} a_V^{-1}$ define an equivalent flat principal $G$-bundle.

Next, we describe infinitesimal deformations of a flat bundle cocycle
$t_{U,V}$. Namely, suppose that $t_{U,V}(s)$ is a smooth $1$-parameter
family of flat bundle cocycles, with $t_{U,V}(0) = t_{U,V}$. Let us
denote the derivative at $s=0$ as $\dot{t}_{U,V} = \tau_{U,V} t_{U,V}$,
with $\tau_{U,V} \in \g$. Then, the defining relations~\eqref{eq:t-inv}
and~\eqref{eq:t-cocyc} impose the following constraints on the
infinitesimal deformation $\tau_{U,V}$:
\begin{gather}
	\tau_{U,V} = - t_{V,U}^{-1} \tau_{V,U} t_{V,U} , \\
\label{eq:tau-cocyc}
	\tau_{U,V}
	+ t_{U,V} \tau_{V,W} t_{U,V}^{-1}
	- \tau_{W,U}
	= 0 .
\end{gather}
On the other hand, suppose that $a_U(s)$ is a smooth $1$-parameter
family of trivialization changes, with $a_U(0) = \id$. Let us write the
derivative at $s=0$ as $\dot{a}_U = -\sigma_U$. The induced
infinitesimal deformation in the transition functions $t_{U,V}(s) =
a_U(s) t_{U,V} a_V^{-1}$ is
\begin{equation}
\label{eq:sigma-bdr}
	\tau_{U,V} = - \sigma_U + t_{U,V} \sigma_V t_{U,V}^{-1} .
\end{equation}

The point of the above calculations is to show that infinitesimal
deformations of the flat principal bundle cocycle, up to infinitesimal
trivialization changes, correspond precisely to the cohomology classes
of a certain sheaf. To complete the argument, we need only introduce the
basic definitions of \v{C}ech cohomology, which is known to compute the
cohomology vector spaces of a corresponding sheaf~\cite{bredon,ks}. We will
take the sheaf to be $\F_\g$, where $\F_\g(U)$ consists of the locally
flat sections of the bundle $\g_P\to M$, associated to the flat
principal $G$-bundle $P\to M$. Let us now fix an open cover $\U = (U_i)$
of $M$ such that each $U_i$ is contractible and any multiple intersection
of the $U_i$ is also contractible. On a manifold, any open cover can be
refined to such a \emph{good cover}~\cite[Thm.5.1]{bott-tu}. In
particular, the flat principal bundle cocycle can always be refined to a
good cover. The good cover hypothesis ensures that the \v{C}ech
cohomology spaces are in fact isomorphic to the actual sheaf
cohomologies.

We define a \v{C}ech $q$-cochain $\sigma$ as an assignment
$(U_{i_1},\ldots, U_{i_{q+1}}) \mapsto \sigma_{i_1\cdots i_{q+1}} \in
\F_\g(U_i\cap \cdots \cap U_{i_{q+1}})$ to every ordered $(q+1)$-tuple
of opens from $\U$. By local flatness, for any $U \in \U$, $\F_\g(U)
\cong \bar{F}_\g \cong \g$ (cf.~Section~\ref{sec:sh-lc}). It is
convenient to think of a \v{C}ech cocycle $\sigma_{i_1\cdots i_{q+1}}$
as taking values in $\g \cong \F_\g(U_i)$. This means that
$\sigma_{i\cdots}$ and $\sigma_{j\cdots}$ restrict to the same element
of $\F_\g(U_i\cap U_j\cap\cdots)$ only if $\sigma_{i\cdots} =
\Ad(t_{U_i,V_i}) \sigma_{j\cdots} = (t_{U_i,V_j}) \sigma_{j\cdots}
(t_{U_i,V_j}^{-1})$. We shall only need the \v{C}ech differential to be
defined on $0$- and $1$-cochains:
\begin{align}
\notag
	(\delta\sigma)_{ij}
		&= \sigma_j|_{U_i\cap U_j} - \sigma_i|_{U_i\cap U_j} \\
\notag
		&= \Ad(t_{U_i,U_j}) \sigma_j - \sigma_i \\
		&= t_{U_i,U_j} \sigma_j t_{U_i,U_j}^{-1} - \sigma_i , \\
\notag
	(\delta \tau)_{ijk}
		&= \tau_{jk}|_{U_i\cap U_j\cap U_k} - \tau_{ik}|_{U_i\cap U_j\cap U_k}
			+ \tau_{ij}|_{U_i\cap U_j\cap U_k} \\
\notag
		&= \Ad(t_{U_i,U_j}) \tau_{jk} - \tau_{ik} + \tau_{ij} \\
		&= t_{U_i,U_j} \tau_{jk} t_{U_i,U_j}^{-1} - \tau_{ik} + \tau_{ij} .
\end{align}
The space of closed \v{C}ech $q$-cocycles modulo the \v{C}ech
coboundaries then is isomorphic to the sheaf cohomology group in degree
$q$, which in our case is $H^q(\F_\g)$.

It should now be clear, from Equations~\eqref{eq:tau-cocyc}
and~\eqref{eq:sigma-bdr}, that the infinitesimal deformation of the flat
bundle cocycle defines a \v{C}ech $1$-cocycle $\tau_{ij} =
\tau_{U_i,U_j}$ and an infinitesimal change in trivialization defines a
\v{C}ech coboundary $\tau_{ij} = (\delta \sigma)_{ij}$, with $\sigma_i =
\sigma_{U_i}$.

\subsection{Flat connection on a principal bundle}\label{sec:conn-def}
Another, ultimately equivalent, way to specify a principal bundle with a
flat connection is as follows.

A \emph{principal $G$-connection} on a principal $G$-bundle $P$ is a
$\g$-valued $1$-form $\omega$ on the total space $P$ (an element of
$\Omega^1(P)\otimes \g$) such that (i) $\omega$ is $\Ad$-equivariant
($R_a^*\omega = \Ad(a^{-1})\omega$, where $R_a\colon P \to P$ is the
action of $a\in G$ on $P$ by right multiplication) and (ii)
$\omega(\beta) = \beta$ for any vertical vector $\beta \in TP$. Recall
that vertical vectors are those annihilated by the tangent map of the
projection $P\to M$ and that the vertical subspace of $TP$ at any point
of $P$ may be naturally identified with $\g$, which we have used in the
preceding definition. The defining condition on the form $\omega$ is
clearly linear inhomogeneous. Thus, the space of all principal
$G$-connections forms an affine subspace of $\Omega^1(P)\otimes \g$.
So, the difference $A = \omega' - \omega$ between any two principal
connections belongs to the subspace of $\Omega^1(P)\otimes \g$ that is
$\Ad$-equivariant and horizontal (annihilates vertical vectors). This
subspace is in fact isomorphic, by pullback along the projection $P \to
P/G \cong M$, to the space of sections $\Secs(\Lambda^1 M\otimes \g_P)$
of the associated bundle $\Lambda^1 M\otimes \g_P \to M$. In fact, we
can identify the spaces of sections $\Secs(\Lambda^p M \otimes \g_P)$
with the $\Ad$-equivariant, horizontal subspaces of $\Omega^p(P)\otimes
\g$. The first order differential operator $DA = \d A + [\omega\wedge
A]$ (see below for notation) preserves these subspaces and hence can be
projected down to a first order differential operator $D\colon
\Secs(\Lambda^p M \otimes \g_P) \to \Secs(\Lambda^{p+1} M \otimes
\g_P)$, which we shall refer to as the \emph{twisted differential}
(cf.~Section~\ref{sec:tw-dr}).

The curvature $\Omega$ of a principal $G$-connection $\omega$ is defined
to be the following $\g$-valued $2$-form on $P$:
\begin{equation}
	\Omega = \d \omega + \frac{1}{2}[\omega \wedge \omega] ,
\end{equation}
where the bracketed wedge product is definite to satisfy $[(\lambda
\otimes \alpha) \wedge (\mu \otimes \beta)] = \lambda \wedge \mu \otimes
[\alpha, \beta]$ for any $\lambda, \mu \in \Omega^1(P)$ and
$\alpha,\beta \in \g$. Since $\Omega$ is $\Ad$-equivariant and
horizontal, we can equally write $\Omega\in \Secs(\Lambda^2 M \otimes
\g_P)$. The twisted differential $D$ is not nilpotent, $D^2\ne 0$.
However, its square is $C^\oo(M)$-linear and so is a differential
operator of order $0$. In fact, we can compute it to be
\begin{equation}
	D^2 A = [\Omega \wedge A] ,
\end{equation}
for any $A\in \Secs(\Lambda^p M \otimes \g_P)$. If $\Omega = 0$, then
the connection is said to be \emph{flat}. This is a sufficient condition
for the twisted differential to become nilpotent, $D^2=0$. A necessary
and sufficient condition would simply be that the curvature $\Omega$
takes values in the center of $\g$, upon local trivialization of $\g_P$.

Given any two flat connections $\omega$ and $\omega'$, their difference
can be represented by a section $\omega'-\omega = A\in \Secs(\Lambda^1 M
\otimes \g_P)$ (or rather its pullback to $P$) that necessarily
satisfies the following equation:
\begin{align}
	0
	&= \d\omega' + \frac{1}{2}[\omega'\wedge\omega'] \\
	&= \d(\omega + A) + \frac{1}{2}[(\omega + A)\wedge (\omega + A)] \\
	&= \d A + [\omega \wedge A] + \frac{1}{2} [A\wedge A] \\
	&= D A + \frac{1}{2} [A\wedge A] .
\end{align}
Where the last expression can be interpreted as computed on $M$ rather
than on $P$. Equating this last expression to zero gives a differential
equation on sections $A \in \Secs(\Lambda^1 M \otimes \g_P)$ identifying
those that parametrize the space of flat principal $G$-connections
(relative to $\omega$, which defines the twisted differential $D$).

An automorphism of a principal $G$-bundle $P\to M$ is a bundle map
$f\colon P\to P$ that covers the identity on $M$ and is equivariant with
respect to the right action of $G$~\cite{abcmm,baum}. It is a standard
fact that such maps can be expressed as functions $a_f\colon P\to G$
that are $\Ad$-equivariant (where $\Ad$ is the left action of $G$ on $G$
by conjugation) with respect to the right action of $G$ on $P$. In turn,
the set of such maps is in bijection with the space of sections of the
associated bundle $G_P = P\times_\Ad G$. Since the map $f\colon P\to P$
is an automorphism, the pullback connection $f^* \omega$ is considered
equivalent to the original one. Given its equivariance, the map $f$
corresponds to a section $a_f\in \Secs(G_P)$. Equivalently, given a
section $a\in \Secs(G_P)$, we can define the corresponding automorphism
map $f_a\colon P\to P$. It is not hard to compute that
\begin{equation}
	f_a^* \omega = \Ad(a^{-1}) \omega + a^{-1} \d a .
\end{equation}
Naturally, the pullback $f^*\omega$ of a flat connection $\omega$ is
also a flat connection. 

Next, we describe infinitesimal deformations of a flat principle
$G$-connection. Given a flat connection described by a $\g$-valued
$1$-form $\omega$ on $P$, any smooth $1$-parameter family of principal
connection can be written as $\omega + A(s)$, with $A(0) = 0$, where
$A(s)$ can, for fixed $s$, be considered as a section of the associated
bundle $\Lambda^1 M \otimes \g_P \to M$. This family consists of flat
connections if and only if the equation $D A(s) + \frac{1}{2}[A(s)\wedge
A(s)] = 0$ is satisfied, where $D$ is the twisted differential defined
by $\omega$.  If $\dot{A}(0) = A$, then the preceding identity imposes
the condition $DA = 0$ on this infinitesimal deformation. Also, if
$a(s)\in \Secs(G_P)$, with $a(0) = \id$ and $\dot{a}(0) = \alpha\in
\Secs(\g_P)$, defines a smooth $1$-parameter family of automorphisms
$f_{a(s)} \colon P\to P$, then the corresponding infinitesimal
deformation of the original flat connection is is equal to $A = \d\alpha
- [\alpha,\omega] = D\alpha$.

It should now be clear that infinitesimal deformations of a given flat
principal $G$-connection, up to infinitesimal automorphisms of the
underlying principal $G$-bundle, are in bijections with the cohomology
vector space $H^1(\Lambda^\bullet\otimes \g_P,D)$ of the twisted de~Rham
complex defined by the original flat connection.

\subsection{Monodromy representation}\label{sec:monodr-def}
A connection, in the sense of Ehresemann, can be defined as a splitting
of the tangent space of $P$ into $T_{x,a}P \cong T_xM \oplus \g$, for
$(x,a) \in P$, with the $\g$ summand canonically identified with the
subspace of vertical vectors, such that the splitting is smooth in $x$
and equivariant in $a$. The $T_xM$ summand is called the
\emph{horizontal subspace} of $T_{x,a}P$. This formulation leads
naturally to the idea of parallel transport. Given a point $(x,a) \in P$
and a smooth curve $\gamma\colon [0,1] \to M$ such that $\gamma(0) = x$
and $\gamma(1) = 1$, there exists a unique lift $\tilde{\gamma}$ of
$\gamma$ to $P$ such that $\tilde{\gamma}(0) = (x,a)$ and the tangent
vector $\dot{\tilde{\gamma}}$ is always horizontal. With $x$ and $y$
fixed, the endpoint $(y,b) = \tilde{\gamma}(1)$ defines $b$ as the image
of $a$ parallel transported along $\gamma$. Since the splitting of $TP$
is equivariant with respect to the right action of $G$ on $P$, so is
parallel transport. It is easy to see that parallel transport does not
depend on the parametrization of $\gamma$, is well defined also when
$\gamma$ is piecewise smooth, and respects concatenation,
$a_{\gamma\eta} = a_\gamma a_\eta$ for $\gamma(1) = \eta(0)$ and
$\gamma\eta$ being the concatenated curve. In particular, if $\gamma$ is
a closed curve based at $x\in M$ ($\gamma(0) = \gamma(1) = x$), then the
effect of parallel transport is equivalent to the right action on the
fiber $P_x$ by some element $a_\gamma \in G$.

Let $\pi = \pi_1(M,x)$ be the fundamental group of $M$ based at some
point $x$. The connection splitting of $TP$ is called \emph{flat} when
the parallel transport along any closed contractible curve $\gamma$ is
trivial, $a_\gamma = \id$, and thus the group element $a_\gamma$
effecting parallel transport along a closed curve $\gamma$ based at $x$
depends only on its homotopy type $[\gamma]\in \pi$. In other words,
parallel transport defines a homomorphism $\rho \colon \pi \to G$,
$\rho([\gamma]) = a_\gamma$, which we call the \emph{monodromy
representation} of the fundamental group of $M$ in the structure group
of $P\to M$ (cf.~the introduction to Section~\ref{sec:killing}).

Thus, any flat principal $G$-bundle gives rise to a representation
$\rho\colon \pi \to G$. Two isomorphic flat principal bundles give rise
to equivalent monodromy representations, where two representations
$\rho'$ and $\rho$ are equivalent if there exists an element $a\in G$
such that $\rho'([\gamma]) = a\rho([\gamma])a^{-1}$. Conversely, any
homomorphism $\rho\colon \pi \to G$ allows us to construct a flat
principal $G$-bundle with a monodromy representation equivalent to
$\rho$. Namely, consider the universal cover $\tilde{M} \to M$ as a
principal $\pi$-bundle and define the total space of the corresponding
principal $G$-bundle as $P = \tilde{M} \times_\rho G$. A flat connection
can be defined on the trivial principal $G$-bundle $\tilde{M}\times G$
using the construction of Section~\ref{sec:cocyc-def} applied to an
cover by contractible open sets and transition maps defined by $\rho$.
This flat connection then projects down to $P$.

Next, we describe infinitesimal deformations of a fixed monodromy
representation $\rho$. Let $\rho_s\colon \pi \to G$ be a smooth
$1$-parameter family of monodromy representations, with $\rho(s) = \rho$
and $\dot{\rho}_s(a) = \tau(s) \rho(a)$ for some $\tau\colon \pi\to \g$.
The representation property $\rho_s([\gamma][\eta]) = \rho_s([\gamma])
\rho_s([\eta])$ imposes the following constraint on the infinitesimal
deformation:
\begin{equation} \label{eq:grp-cocyc}
	\tau([\gamma])
	+ \rho([\gamma])\tau([\eta]) \rho([\gamma])^{-1}
	- \tau([\gamma][\eta])
	= 0.
\end{equation}
A family of trivial deformations is given by $\rho_s([\gamma]) = a_s
\rho([\gamma]) a_s^{-1}$ for a smooth $1$-parameter family $a_s \in G$,
with $a_0 = \id$ and $\dot{a}_0 = -\sigma \in \g$. The corresponding
infinitesimal deformation of the representation is given by
\begin{equation} \label{eq:grp-cobdr}
	\tau([\gamma]) = -\sigma + \rho([\gamma]) \sigma \rho([\gamma])^{-1} .
\end{equation}

The point of the above calculations is to show that these infinitesimal
deformations can be identified with certain \emph{group cohomology}
classes. To see that, we need to introduce some basic
definitions~\cite[Ch.6]{weibel}, \cite{group-cohom}. Group cohomology is
defined given a group and a representation thereof. We will give the
definitions by directly taking the group to be $\pi$ and the
representation to be the composite adjoint representation of $\pi$ on
$\g$, $\Ad_\rho = \Ad \circ \rho \colon \pi \to \GL(\g)$. The vector
space $C^p(\pi,\Ad_\rho)$ of $p$-cochains consists of functions
$\sigma\colon \pi^p \to \g$, where $\pi^p = \pi \times \cdots \times
\pi$ is the $p$-fold product. The cochain differentials $\delta\colon
C^p(\pi,\Ad_\rho) \to C^{p+1}(\pi,\Ad_\rho)$ are defined by the formula
\begin{multline}
	\delta\sigma([\gamma_1],\ldots,[\gamma_{p+1}])
	= (-1)^{p+1} \sigma([\gamma_1],\ldots,[\gamma_p]) \\
		{} + \Ad_\rho([\gamma_1]) \sigma([\gamma_2],\ldots,[\gamma_{p+1}]) \\
		+ \sum_{q=1}^p (-1)^q \sigma(\ldots,[\gamma_q][\gamma_{q+1}],\ldots) .
\end{multline}
For $0$- and $1$-cochains, we have the following explicit formulas:
\begin{align}
\notag
	\delta\sigma([\gamma])
		&= -\sigma + \Ad_\rho([\gamma])\sigma \\
		&= - \sigma + \rho([\gamma]) \sigma \rho([\gamma])^{-1} , \\
\notag
	\delta\tau([\gamma],[\eta])
		&= \tau([\gamma]) + \Ad_\rho([\gamma]) \tau([\eta]) - \tau([\gamma][\eta]) \\
		&= \tau([\gamma]) + \rho([\gamma]) \tau([\eta]) \rho([\gamma])^{-1}
			- \tau([\gamma][\eta]) .
\end{align}
It is worth noting that the degree-$0$ group cohomology is isomorphic to
the subspace of the representation on which the group acts trivially,
$H^0(\pi,\Ad_\rho) \cong \g^\pi$.

It should now be clear from Equations~\eqref{eq:grp-cocyc}
and~\eqref{eq:grp-cobdr} that infinitesimal deformations of a monodromy
representations $\rho\colon \pi \to G$, up to deformations by
conjugation, are in bijection with the group cohomology
$H^1(\pi,\Ad_\rho)$ of the group $\pi$ with coefficients in the
composite adjoint representation of $\pi$ on $\g$.

\bibliographystyle{utphys-alpha}
\bibliography{paper-causcohom}

\newcommand{\noopsort}[1]{}
\providecommand{\href}[2]{#2}\begingroup\raggedright\begin{thebibliography}{10}

\bibitem{abcmm}
M.~C. Abbati, R.~Cirelli, A.~Mania', and P.~Michor, ``The lie group of
  automorphisms of a principle bundle,''
  \href{http://dx.doi.org/10.1016/0393-0440(89)90015-6}{{\em Journal of
  Geometry and Physics} {\bfseries 6} (1989) 215--235}.

\bibitem{acm}
J.~Ambj{\o}rn, M.~Carfora, and A.~Marzuoli, {\em The Geometry of Dynamical
  Triangulations}.
\newblock Springer, 1997.

\bibitem{stven}
C.~Amrouche, P.~G. Ciarlet, L.~Gratie, and S.~Kesavan, ``On {Saint Venant's}
  compatibility conditions and {Poincar\'{e}'s} lemma,''
  \href{http://dx.doi.org/10.1016/j.crma.2006.03.026}{{\em Comptes Rendus
  Mathematique} {\bfseries 342} (2006) 887--891}.

\bibitem{anderson-big}
I.~M. Anderson, ``The variational bicomplex.'' Unpublished draft, 1989.

\bibitem{anderson-small}
I.~M. Anderson,
  \href{http://dx.doi.org/10.1090/conm/132/1188434}{``Introduction to the
  variational bicomplex,''} in {\em Mathematical aspects of classical field
  theory}, M.~J. Gotay, J.~E. Marsden, and V.~Moncrief, eds., vol.~132 of {\em
  Contemporary Mathematics}, pp.~51--73.
\newblock American Mathematical Society, Providence, Rhode Island, 1992.

\bibitem{baum}
H.~Baum, \href{http://dx.doi.org/10.1007/978-3-540-38293-5}{{\em
  Eichfeldtheorie: Eine Einf{\"{u}}hrung in Die Differentialgeometrie Auf
  Faserb{\"{u}}ndeln}}.
\newblock Springer, London, 2010.

\bibitem{bss}
C.~Becker, A.~Schenkel, and R.~J. Szabo, ``Differential cohomology and locally
  covariant quantum field theory.'' 2014.
\newblock \href{http://arxiv.org/abs/1406.1514}{{\ttfamily arXiv:1406.1514}}.

\bibitem{benini}
M.~Benini, ``Optimal space of linear classical observables for {Maxwell}
  {$k$}-forms via spacelike and timelike compact {de Rham} cohomologies.''
  2014.
\newblock \href{http://arxiv.org/abs/1401.7563}{{\ttfamily arXiv:1401.7563}}.

\bibitem{bdm}
M.~Benini, C.~Dappiaggi, and S.~Murro, ``Radiative observables for linearized
  gravity on asymptotically flat spacetimes and their boundary induced
  states.'' 2014.
\newblock \href{http://arxiv.org/abs/1404.4551}{{\ttfamily arXiv:1404.4551}}.

\bibitem{bds}
M.~Benini, C.~Dappiaggi, and A.~Schenkel, ``Quantized {Abelian} principal
  connections on {Lorentzian} manifolds,''
  \href{http://dx.doi.org/10.1007/s00220-014-1917-0}{{\em Communications in
  Mathematical Physics} {\bfseries 330} (2013) 123--152},
  \href{http://arxiv.org/abs/1303.2515}{{\ttfamily arXiv:1303.2515}}.

\bibitem{berger}
M.~Berger, ``Encounter with a geometer: {Eugenio Calabi},'' in {\em Manifolds
  and Geometry}, P.~de~Bartolomeis, F.~Tricerri, and E.~Vesentini, eds.,
  vol.~26 of {\em Symposia Mathematica}, pp.~20--60.
\newblock Cambridge, 1996.

\bibitem{bbl}
L.~B. Bergery, J.~P. Bourguignon, and J.~Lafontaine,
  \href{http://dx.doi.org/10.1090/pspum/027.1/0388467}{``D\'{e}formations
  localement triviales des vari\'{e}t\'{e}s riemanniennes,''} in {\em
  Differential Geometry, Part 1}, vol.~27 of {\em Proceedings of Symposia in
  Pure Mathematics}, pp.~3--32.
\newblock AMS, Providence, RI, 1975.

\bibitem{bgg}
I.~N. Bernstein, I.~M. Gelfand, and S.~I. Gelfand, ``Differential operators on
  the base affine space and a study of {$\mathfrak{g}$}-modules,'' in {\em Lie
  groups and their representations: Summer School of the Bolyai J\'{a}nos
  Mathematical Society}, I.~M. Gelfand, ed., pp.~21--64.
\newblock Adam Hilger, Halsted, NY, 1975.

\bibitem{bcjr}
D.~Bini, C.~Cherubini, R.~T. Jantzen, and R.~Ruffini, ``{de Rham} wave equation
  for tensor valued $p$-forms,''
  \href{http://dx.doi.org/10.1142/s0218271803003785}{{\em International Journal
  of Modern Physics D} {\bfseries 12} (2003) 1363--1384}.

\bibitem{bott-tu}
R.~Bott and L.~W. Tu, {\em Differential Forms in Algebraic Topology}, vol.~82
  of {\em Graduate Texts in Mathematics}.
\newblock Springer, New York, 1982.

\bibitem{bredon}
G.~E. Bredon, {\em Sheaf Theory}.
\newblock Graduate Texts in Mathematics. Springer, New York, 1997.

\bibitem{bfr}
R.~Brunetti, K.~Fredenhagen, and P.~L. Ribeiro, ``Algebraic structure of
  classical field theory {I}: Kinematics and linearized dynamics for real
  scalar fields.'' 2012.
\newblock \href{http://arxiv.org/abs/1209.2148}{{\ttfamily arXiv:1209.2148}}.

\bibitem{bcggg}
R.~L. Bryant, S.~S. Chern, R.~B. Gardner, H.~L. Goldschmidt, and P.~A.
  Griffiths, {\em Exterior Differential Systems}, vol.~18 of {\em Mathematical
  Sciences Research Institute Publications}.
\newblock Springer, 2011.

\bibitem{calabi}
E.~Calabi, \href{http://dx.doi.org/10.1090/pspum/003}{``On compact,
  {R}iemannian manifolds with constant curvature. {I},''} in {\em Differential
  Geometry}, C.~B. Allendoerfer, ed., vol.~3 of {\em Proceedings of Symposia in
  Pure Mathematics}, pp.~155--180.
\newblock AMS, Providence, RI, 1961.

\bibitem{chr-kl}
D.~Christodoulou and S.~Klainerman, {\em The global nonlinear stability of the
  {M}inkowski space}, vol.~41 of {\em Princeton Mathematical Series}.
\newblock Princeton University Press, Princeton, NJ, 1993.

\bibitem{ferr-fluid}
B.~Coll and J.~J. Ferrando, ``Thermodynamic perfect fluid. its rainich
  theory,'' \href{http://dx.doi.org/10.1063/1.528477}{{\em Journal of
  Mathematical Physics} {\bfseries 30} (1989) 2918--2922}.

\bibitem{dl}
C.~Dappiaggi and B.~Lang, ``Quantization of {Maxwell's} equations on curved
  backgrounds and general local covariance,''
  \href{http://dx.doi.org/10.1007/s11005-012-0571-8}{{\em Letters in
  Mathematical Physics} {\bfseries 101} (2012) 265--287},
  \href{http://arxiv.org/abs/1104.1374}{{\ttfamily arXiv:1104.1374}}.

\bibitem{dimca}
A.~Dimca, {\em Sheaves in Topology}.
\newblock Universitext. Springer, Berlin, 2004.

\bibitem{eastwood}
M.~Eastwood, ``Variations on the de {R}ham complex,'' {\em Notices of the
  American Mathematical Society} {\bfseries 46} (1999) 1368--1376.
  \url{http://www.ams.org/notices/199911/fea-eastwood.pdf}.

\bibitem{eilenberg-maclane}
S.~Eilenberg and S.~MacLane, ``Relations between homology and homotopy groups
  of spaces,'' \href{http://dx.doi.org/10.2307/1969165}{{\em The Annals of
  Mathematics} {\bfseries 46} (1945) 480--509}.

\bibitem{emparan-reall}
R.~Emparan and H.~S. Reall, ``Black holes in higher dimensions,''
  \href{http://dx.doi.org/10.12942/lrr-2008-6}{{\em Living Reviews in
  Relativity} {\bfseries 11} no.~6, (2008) },
  \href{http://arxiv.org/abs/0801.3471}{{\ttfamily arXiv:0801.3471}}.

\bibitem{saez-schw}
J.~J. Ferrando and J.~A. S\'{a}ez, ``An intrinsic characterization of the
  {Schwarzschild} metric,''
  \href{http://dx.doi.org/10.1088/0264-9381/15/5/014}{{\em Classical and
  Quantum Gravity} {\bfseries 15} (1998) 1323--1330}.

\bibitem{saez-kerr}
J.~J. Ferrando and J.~A. S\'{a}ez, ``An intrinsic characterization of the
  {Kerr} metric,'' \href{http://dx.doi.org/10.1088/0264-9381/26/7/075013}{{\em
  Classical and Quantum Gravity} {\bfseries 26} (2009) 075013},
  \href{http://arxiv.org/abs/0812.3310}{{\ttfamily arXiv:0812.3310}}.

\bibitem{fewster-hunt}
C.~J. Fewster and D.~S. Hunt, ``Quantization of linearized gravity in
  cosmological vacuum spacetimes,''
  \href{http://dx.doi.org/10.1142/S0129055X13300033}{{\em Reviews in
  Mathematical Physics} {\bfseries 25} (2013) 1330003},
  \href{http://arxiv.org/abs/1203.0261}{{\ttfamily arXiv:1203.0261}}.

\bibitem{fl}
C.~J. Fewster and B.~Lang, ``Dynamical locality of the free {Maxwell} field.''
  2014.
\newblock \href{http://arxiv.org/abs/1403.7083}{{\ttfamily arXiv:1403.7083}}.

\bibitem{forger-romero}
M.~Forger and S.~V. Romero, ``Covariant {Poisson} brackets in geometric field
  theory,'' \href{http://dx.doi.org/10.1007/s00220-005-1287-8}{{\em
  Communications in Mathematical Physics} {\bfseries 256} (2005) 375--410},
  \href{http://arxiv.org/abs/math-ph/0408008}{{\ttfamily
  arXiv:math-ph/0408008}}.

\bibitem{fkwc}
S.~A. Fulling, R.~C. King, B.~G. Wybourne, and C.~J. Cummins, ``Normal forms
  for tensor polynomials. {I}. {T}he {R}iemann tensor,''
  \href{http://dx.doi.org/10.1088/0264-9381/9/5/003}{{\em Classical and Quantum
  Gravity} {\bfseries 9} (1992) 1151}.

\bibitem{fulton}
W.~Fulton, {\em Young Tableaux: With Applications to Representation Theory and
  Geometry}, vol.~35 of {\em London Mathematical Society Student Texts}.
\newblock Cambridge University Press, New York, 1997.

\bibitem{gasqui-goldschmidt-fr}
J.~Gasqui and H.~Goldschmidt, ``D\'{e}formations infinit\'{e}simales des
  espaces riemanniens localement sym\'{e}triques. {I},''
  \href{http://dx.doi.org/10.1016/0001-8708(83)90090-7}{{\em Advances in
  Mathematics} {\bfseries 48} (1983) 205--285}.

\bibitem{gasqui-goldschmidt}
J.~Gasqui and H.~Goldschmidt,
  \href{http://dx.doi.org/10.1007/978-94-009-3057-5\_14}{``Complexes of
  differential operators and symmetric spaces,''} in {\em Deformation Theory of
  Algebras and Structures and Applications}, M.~Hazewinkel and M.~Gerstenhaber,
  eds., vol.~247 of {\em NATO ASI Series}, pp.~797--827.
\newblock Kluwer, Dordrecht, 1988.

\bibitem{gelfand-manin}
S.~I. Gelʹfand and Y.~Manin, {\em Methods of Homological Algebra: Springer
  monographs in mathematics}.
\newblock Springer Monographs in Mathematics. Springer, New York, NY, 2nd~ed.,
  2003.

\bibitem{goldman}
W.~M. Goldman, ``The symplectic nature of fundamental groups of surfaces,''
  \href{http://dx.doi.org/10.1016/0001-8708(84)90040-9}{{\em Advances in
  Mathematics} {\bfseries 54} (1984) 200--225}.

\bibitem{goldschmidt-lin}
H.~Goldschmidt, ``Existence theorems for analytic linear partial differential
  equations,'' \href{http://dx.doi.org/10.2307/1970689}{{\em The Annals of
  Mathematics} {\bfseries 86} (1967) 246--270}.

\bibitem{goldschmidt-calabi}
H.~Goldschmidt, ``Duality theorems in deformation theory,''
  \href{http://dx.doi.org/10.1090/s0002-9947-1985-0805952-x}{{\em Transactions
  of the American Mathematical Society} {\bfseries 292} (1985) 1}.

\bibitem{ghv}
W.~Greub, S.~Halperin, and K.~Vanstone, {\em Connections, Curvature, and
  Cohomology. {Vol. I}: De {Rham} Cohomology of Manifolds and Vector Bundles},
  vol.~47-I of {\em Pure and Applied Mathematics}.
\newblock Academic Press, 1972.

\bibitem{hack-lingrav}
T.-P. Hack, ``Quantization of the linearised {Einstein-Klein-Gordon} system on
  arbitrary backgrounds and the special case of perturbations in inflation.''
  2014.
\newblock \href{http://arxiv.org/abs/1403.3957}{{\ttfamily arXiv:1403.3957}}.

\bibitem{hs}
T.-P. Hack and A.~Schenkel, ``Linear bosonic and fermionic quantum gauge
  theories on curved spacetimes,''
  \href{http://dx.doi.org/10.1007/s10714-013-1508-y}{{\em General Relativity
  and Gravitation} {\bfseries 45} (2013) 877--910},
  \href{http://arxiv.org/abs/1205.3484}{{\ttfamily arXiv:1205.3484}}.

\bibitem{hoermander-III}
L.~H\"{o}rmander, {\em The Analysis of Linear Partial Differential Operators
  {III}}, vol.~274 of {\em Grundlehren Der Mathematischen Wissenschaften}.
\newblock Springer, 1985.

\bibitem{hoermander-I}
L.~H\"{o}rmander, {\em The Analysis of Linear Partial Differential Operators
  {I}: Distribution Theory and Fourier Analysis}, vol.~256 of {\em Grundlehren
  Der Mathematischen Wissenschaften}.
\newblock Springer, 2012.

\bibitem{ks}
M.~Kashiwara and P.~Schapira, {\em Sheaves on Manifolds}, vol.~292 of {\em
  Grundlehren der mathematischen Wissenschaften in Einzeldarstellungen}.
\newblock Springer, 1990.

\bibitem{kh-big}
I.~Khavkine, ``Characteristics, conal geometry and causality in locally
  covariant field theory.'' 2012.
\newblock \href{http://arxiv.org/abs/1211.1914}{{\ttfamily arXiv:1211.1914}}.

\bibitem{kh-cohom}
I.~Khavkine, ``Cohomology with causally restricted supports.'' 2014.
\newblock \href{http://arxiv.org/abs/1404.1932}{{\ttfamily arXiv:1404.1932}}.

\bibitem{kh-peierls}
I.~Khavkine, ``Covariant phase space, constraints, gauge and the {P}eierls
  formula,'' \href{http://dx.doi.org/10.1142/s0217751x14300099}{{\em
  International Journal of Modern Physics A} {\bfseries 29} (2014) 1430009},
  \href{http://arxiv.org/abs/1402.1282}{{\ttfamily arXiv:1402.1282}}.

\bibitem{kh-linstab}
I.~Khavkine, ``Topology, rigid cosymmetries and linearization instability in
  higher gauge theories,''
  \href{http://dx.doi.org/10.1007/s00023-014-0321-9}{{\em Annales Henri
  Poincar\'{e}} (2014) }, \href{http://arxiv.org/abs/1303.2406}{{\ttfamily
  arXiv:1303.2406}}.

\bibitem{kobayashi}
S.~Kobayashi, {\em Differential Geometry of Complex Vector Bundles}.
\newblock Princeton University Press, 1987.

\bibitem{kodama}
H.~Kodama, ``Canonical structure of locally homogeneous systems on compact
  closed {3-Manifolds} of types {$E^3$}, {Nil} and {Sol},''
  \href{http://dx.doi.org/10.1143/ptp.99.173}{{\em Progress of Theoretical
  Physics} {\bfseries 99} (1998) 173--236},
  \href{http://arxiv.org/abs/gr-qc/9705066}{{\ttfamily arXiv:gr-qc/9705066}}.

\bibitem{kth}
T.~Koike, M.~Tanimoto, and A.~Hosoya, ``Compact homogeneous universes,''
  \href{http://dx.doi.org/10.1063/1.530819}{{\em Journal of Mathematical
  Physics} {\bfseries 35} (1994) 4855--4888},
  \href{http://arxiv.org/abs/gr-qc/9405052}{{\ttfamily arXiv:gr-qc/9405052}}.

\bibitem{kms}
I.~Kola{\v{r}}, P.~W. Michor, and J.~Slov{\'{a}}k, {\em Natural Operations in
  Differential Geometry}.
\newblock Springer, 1993.

\bibitem{lang}
S.~Lang, {\em Differential and Riemannian Manifolds}, vol.~160 of {\em Graduate
  Texts in Mathematics}.
\newblock Springer, 1995.

\bibitem{lichnerowicz}
A.~Lichnerowicz, ``Propagateurs, commutateurs et anticommutateurs en
  relativit\'{e} g\'{e}n\'{e}rale,'' in {\em Relativity, Groups and Topology},
  C.~DeWitt and B.~S. DeWitt, eds., pp.~821--861.
\newblock Gordon and Breach, New York, 1964.

\bibitem{mtw}
C.~W. Misner, K.~S. Thorne, and J.~A. Wheeler, {\em Gravitation}.
\newblock Physics Series. W. H. Freeman, San Francisco, 1973.

\bibitem{morita}
S.~Morita, {\em Geometry of Characteristic Classes}, vol.~199 of {\em
  Translations of mathematical monographs}.
\newblock American Mathematical Society, 2001.

\bibitem{rsset}
{$n$Lab}, ``reduced simplicial set,'' 2014.
\newblock \url{http://ncatlab.org/nlab/show/reduced+simplicial+set}. [Online;
  accessed 02-July-2014].

\bibitem{lrr}
M.~A.~A. \noopsort{Leeuwen}{van Leeuwen}, ``{Littlewood-Richardson} rule.''
\newblock \url{http://www-math.univ-poitiers.fr/~maavl/LiE/form.html}. Online
  calculator based on~\cite{lie-book}.

\bibitem{lie-book}
M.~A.~A. \noopsort{Leeuwen}{van Leeuwen}, A.~M. Cohen, and B.~Lisser, {\em
  {LiE}, A Package for {Lie} Group Computations}.
\newblock Computer Algebra Nederland, Amsterdam, 1992.
\newblock \url{http://www-math.univ-poitiers.fr/~maavl/LiE/}.

\bibitem{olver-lie}
P.~J. Olver, \href{http://dx.doi.org/10.1007/978-1-4612-4350-2}{{\em
  Applications of Lie groups to differential equations}}, vol.~107 of {\em
  Graduate Texts in Mathematics}.
\newblock Springer-Verlag, New York, second~ed., 1993.

\bibitem{roura}
G.~P\'{e}rez-Nadal, A.~Roura, and E.~Verdaguer, ``Stress tensor fluctuations in
  {de Sitter} spacetime,''
  \href{http://dx.doi.org/10.1088/1475-7516/2010/05/036}{{\em Journal of
  Cosmology and Astroparticle Physics} {\bfseries 2010} (2010) 036},
  \href{http://arxiv.org/abs/0911.4870}{{\ttfamily arXiv:0911.4870}}.

\bibitem{pommaret}
J.-F. Pommaret, {\em Systems of partial differential equations and {Lie}
  pseudogroups}, vol.~14 of {\em Mathematics and Its Applications}.
\newblock Gordon and Breach, New York, 1978.

\bibitem{pommaret-mech}
J.-F. Pommaret, {\em Lie Pseudogroups and Mechanics}, vol.~16 of {\em
  Mathematics and its applications}.
\newblock Gordon and Breach, New York, 1988.

\bibitem{postnikov-dok}
M.~M. Postnikov, ``Determination of the homology groups of a space by means of
  the homotopy invariants,'' {\em Doklady Akademii Nauk SSSR (N.S.)} {\bfseries
  76} (1951) 359--362.

\bibitem{postnikov-ru}
M.~M. Postnikov, {\em Issledovaniya po gomotopi\v{c}esko\u{i} teorii
  nepreryvnyh otobra\v{z}eni\u{i}. {I}. {A}lgebrai\v{c}eskaya teoriya sistem.
  {II}. {N}atural'naya sistema i gomotopi\v{c}eski\u{i} tip}, vol.~46 of {\em
  Trudy Matematicheskogo Instituta Steklova}.
\newblock Izdatelstvo Akademii Nauk SSSR, Moscow, 1955.
\newblock \url{http://mi.mathnet.ru/tm1182}.

\bibitem{postnikov-en}
M.~M. Postnikov, {\em Investigations in the homotopy theory of continuous
  mappings. {I}. {T}he algebraic theory of systems. {II}. {T}he natural system
  and homotopy type}, vol.~7 of {\em American Mathematical Society Translations
  (2)}.
\newblock American Mathematical Society, 1957.

\bibitem{quillen}
D.~G. Quillen, {\em Formal properties of over-determined systems of partial
  differential equations}.
\newblock PhD thesis, Harvard University, 1964.

\bibitem{reed-simon}
M.~Reed and B.~Simon, {\em Methods of Modern Mathematical Physics {I}:
  Functional Analysis}.
\newblock Academic Press, New York, 1981.

\bibitem{ryan}
M.~Ryan, ``{Teukolsky} equation and {Penrose} wave equation,''
  \href{http://dx.doi.org/10.1103/physrevd.10.1736}{{\em Physical Review D}
  {\bfseries 10} (1974) 1736--1740}.

\bibitem{sdh}
K.~Sanders, C.~Dappiaggi, and T.-P. Hack, ``Electromagnetism, local covariance,
  the {Aharonov-Bohm} effect and {Gauss'} law,''
  \href{http://dx.doi.org/10.1007/s00220-014-1989-x}{{\em Communications in
  Mathematical Physics} {\bfseries 328} (2014) 625--667},
  \href{http://arxiv.org/abs/1211.6420}{{\ttfamily arXiv:1211.6420}}.

\bibitem{schwartz}
L.~Schwartz, {\em Th\'{e}orie des distributions}, vol.~1--2.
\newblock Hermann, Paris, 1951.

\bibitem{seiler-inv}
W.~M. Seiler, {\em Involution: The Formal Theory of Differential Equations and
  its Applications in Computer Algebra}, vol.~24 of {\em Algorithms and
  Computation in Mathematics}.
\newblock Springer, 2010.

\bibitem{serre}
J.-P. Serre, ``Un th\'{e}or\`{e}me de dualit\'{e},''
  \href{http://dx.doi.org/10.1007/bf02564268}{{\em Commentarii Mathematici
  Helvetici} {\bfseries 29} (1955) 9--26}.

\bibitem{shl-tarkh}
A.~A. Shlapunov and N.~N. Tarkhanov, ``A homotopy operator for {Spencer's}
  sequence in the {$C^\infty$}-case,''
  \href{http://dx.doi.org/10.3103/s1055134409020035}{{\em Siberian Advances in
  Mathematics} {\bfseries 19} (2009) 91--127}.

\bibitem{smith}
R.~T. Smith, ``Examples of elliptic complexes,''
  \href{http://dx.doi.org/10.1090/s0002-9904-1976-14028-1}{{\em Bulletin of the
  American Mathematical Society} {\bfseries 82} (1976) 297--299}.

\bibitem{spencer-deform1}
D.~C. Spencer, ``Deformation of structures on manifolds defined by transitive,
  continuous pseudogroups {Part I}: Infinitesimal deformations of structure,''
  \href{http://dx.doi.org/10.2307/1970277}{{\em The Annals of Mathematics}
  {\bfseries 76} (1962) 306--398}.

\bibitem{spencer-deform2}
D.~C. Spencer, ``Deformation of structures on manifolds defined by transitive,
  continuous pseudogroups {Part II}: Deformations of structure,''
  \href{http://dx.doi.org/10.2307/1970367}{{\em The Annals of Mathematics}
  {\bfseries 76} (1962) 399--445}.

\bibitem{spencer}
D.~C. Spencer, ``Overdetermined systems of linear partial differential
  equations,'' \href{http://dx.doi.org/10.1090/s0002-9904-1969-12129-4}{{\em
  Bulletin of the American Mathematical Society} {\bfseries 75} (1969)
  179--240}.

\bibitem{spivak}
M.~Spivak, {\em A Comprehensive Introduction to Differential Geometry}, vol.~1.
\newblock Publish or Perish, 3rd~ed., 1999.

\bibitem{steenrod}
N.~E. Steenrod, {\em The Topology of Fibre Bundles}.
\newblock Princeton Mathematical Series. Princeton University Press, 1951.

\bibitem{stepanov}
S.~E. Stepanov, ``The {Killing-Yano} tensor,''
  \href{http://dx.doi.org/10.1023/a\%253a1022645304580}{{\em Theoretical and
  Mathematical Physics} {\bfseries 134} (2003) 333--338}.

\bibitem{sw-pert}
J.~M. Stewart and M.~Walker, ``Perturbations of {Space-Times} in general
  relativity,'' \href{http://dx.doi.org/10.1098/rspa.1974.0172}{{\em
  Proceedings of the Royal Society of London. A. Mathematical and Physical
  Sciences} {\bfseries 341} (1974) 49--74}.

\bibitem{kth2}
M.~Tanimoto, T.~Koike, and A.~Hosoya, ``Dynamics of compact homogeneous
  universes,'' \href{http://dx.doi.org/10.1063/1.531853}{{\em Journal of
  Mathematical Physics} {\bfseries 38} (1996) 350--368},
  \href{http://arxiv.org/abs/gr-qc/9604056}{{\ttfamily arXiv:gr-qc/9604056}}.

\bibitem{tarkhanov}
N.~N. Tarkhanov, \href{http://dx.doi.org/10.1007/978-94-011-0327-5}{{\em
  Complexes of Differential Operators}}, vol.~340 of {\em Mathematics and Its
  Applications}.
\newblock Kluwer, Dordrecht, 1995.

\bibitem{treves}
F.~Treves, {\em Topological Vector Spaces, Distributions and Kernels}.
\newblock Dover Books on Mathematics. Dover Publications, 2013.

\bibitem{wald}
R.~M. Wald, {\em General Relativity}.
\newblock University of Chicago Press, Chicago, 1984.

\bibitem{weibel}
C.~A. Weibel, \href{http://dx.doi.org/10.1017/CBO9781139644136}{{\em An
  introduction to homological algebra}}, vol.~38 of {\em Cambridge Studies in
  Advanced Mathematics}.
\newblock Cambridge University Press, Cambridge, 1994.

\bibitem{weil}
A.~Weil, ``On discrete subgroups of {Lie} groups ({II}),''
  \href{http://dx.doi.org/10.2307/1970212}{{\em The Annals of Mathematics}
  {\bfseries 75} (1962) 578--602}.

\bibitem{whitehead}
G.~W. Whitehead, {\em Elements of Homotopy Theory}, vol.~61 of {\em Graduate
  Texts in Mathematics}.
\newblock Springer, New York, NY, 2012.

\bibitem{group-cohom}
Wikipedia, ``Group cohomology --- {Wikipedia, The Free Encyclopedia},'' 2014.
\newblock
  \url{http://en.wikipedia.org/w/index.php?title=Group_cohomology&oldid=611025879}.
  [Online; accessed 27-June-2014].

\bibitem{levi-civita}
Wikipedia, ``{Levi-Civita} symbol --- {Wikipedia, The Free Encyclopedia},''
  2014.
\newblock
  \url{http://en.wikipedia.org/w/index.php?title=Levi-Civita_symbol&oldid=607146551}.
  [Online; accessed 16-June-2014].

\bibitem{maschke}
Wikipedia, ``Maschke's theorem --- {Wikipedia, The Free Encyclopedia},'' 2014.
\newblock
  \url{http://en.wikipedia.org/w/index.php?title=Maschke%27s_theorem&oldid=590172149}.
  [Online; accessed 29-June-2014].

\bibitem{schur}
Wikipedia, ``Schur's lemma --- {Wikipedia, The Free Encyclopedia},'' 2014.
\newblock
  \url{http://en.wikipedia.org/w/index.php?title=Schur%27s_lemma&oldid=607166328}.
  [Online; accessed 29-June-2014].

\bibitem{wolf-cc}
J.~A. Wolf, {\em Spaces of Constant Curvature}, vol.~372 of {\em AMS-Chelsea}.
\newblock American Mathematical Society, 6th~ed., 2011.

\end{thebibliography}\endgroup

\end{document}